\documentclass[a4paper,10pt]{article}

\usepackage[utf8]{inputenc} 
\usepackage[T1]{fontenc}    
\usepackage{hyperref}       
\usepackage{url}            
\usepackage{booktabs}       
\usepackage{amsfonts}       
\usepackage{nicefrac}       
\usepackage{microtype}      
\usepackage{lipsum}
\usepackage{fancyhdr}       
\usepackage{graphicx}       
\usepackage{orcidlink}
\graphicspath{{media/}}     
\usepackage{graphicx}
\usepackage{epstopdf, epsfig}

\usepackage{amsfonts}
\usepackage{amsmath}
\usepackage{amssymb}
\usepackage{amsopn}
\usepackage{amsthm}

\usepackage{algorithmic}
\usepackage{enumitem}

\usepackage{tikz}
\usetikzlibrary{cd}
\usepackage{onimage}
\usepackage[caption=false]{subfig}
\ifpdf
  \DeclareGraphicsExtensions{.eps,.pdf,.png,.jpg}
\else
  \DeclareGraphicsExtensions{.eps}
\fi

\newtheorem{lemma}{Lemma}
\newtheorem{proposition}{Proposition}
\newtheorem{theorem}{Theorem}
\newtheorem{corollary}{Corollary}
\newtheorem{remark}{Remark}
\newtheorem{example}{Example}

\usepackage{array}
\newcommand*{\vertbar}{\rule[-1ex]{0.5pt}{2.5ex}}

\DeclareMathOperator*{\argmin}{\arg\!\min\enskip}

\DeclareMathOperator{\Tr}{tr}

\DeclareMathOperator{\vct}{vec}

\DeclareMathOperator{\grad}{\nabla}
\DeclareMathOperator{\vdiv}{\nabla\cdot}
\DeclareMathOperator{\Lap}{\Delta}
\DeclareMathOperator{\td}{\mathrm{d}}
\DeclareMathOperator{\ddt}{\frac{\td}{\td t}}

\DeclareMathOperator{\Id}{Id}

\newcommand{\mat}[1]{{\boldsymbol{#1}}}
\newcommand{\vect}[1]{{\boldsymbol{#1}}}
\newcommand{\mcal}[1]{{\mathcal{#1}}}

\newcommand{\R}{\mathbb{R}}

\newcommand{\Wei}{\mbox{\textit{Wi}}} 
\newcommand{\Rey}{\mbox{\textit{Re}}} 
\newcommand{\aligntridown}{\raise.4ex\hbox{$\bigtriangledown$}}

\newcommand{\revtwo}[1]{\textcolor{black}{#1}}

\newenvironment{revtwobox}{\color{black}}{}

\begin{document}

{\vskip -1cm}

\title{Machine Learning in Viscoelastic Fluids via Energy-Based Kernel Embedding
}
\author{
{Samuel E. Otto \orcidlink{0000-0002-5086-8057}}\thanks{\,Present address: Sibley School of Mechanical and Aerospace Engineering, Cornell University, Ithaca, NY, USA, e-mail: s.otto@cornell.edu}\\
{\normalsize AI Institute in Dynamic Systems,}\\
{\normalsize University of Washington, Seattle, WA, USA}\\ [2ex]
{Cassio M. Oishi \orcidlink{0000-0002-0904-6561}}\thanks{\,e-mail: cassio.oishi@unesp.br} \\
{\normalsize Departamento de Matem\'atica e Computa\c c\~ao, Faculdade de Ci\^encias e Tecnologia, }\\
{\normalsize São Paulo State University, Presidente Prudente, Brazil} \\ [2ex]
{Fabio Amaral  \orcidlink{0000-0001-6945-8376}}\thanks{\,e-mail: fabio.amaral@unesp.br} \\
{\normalsize Departamento de Matem\'atica e Computa\c c\~ao, Faculdade de Ci\^encias e Tecnologia, }\\
{\normalsize São Paulo State University, Presidente Prudente, Brazil} \\ [2ex]
{Steven L. Brunton \orcidlink{0000-0002-6565-5118}}\thanks{\,e-mail: sbrunton@uw.edu}\\
{\normalsize Department of Mechanical Engineering,}\\
{\normalsize University of Washington, Seattle, WA, USA}\\ [2ex]
{J. Nathan Kutz \orcidlink{0000-0002-6004-2275}}\thanks{\,e-mail: kutz@uw.edu}\\
{\normalsize Department of Applied Mathematics,}\\
{\normalsize University of Washington, Seattle, WA, USA}\\[2ex]
}

\date{}

\maketitle

\begin{abstract}

The ability to measure differences in collected data is of fundamental importance for quantitative science and machine learning, motivating the establishment of metrics grounded in physical principles. 
In this study, we focus on the development of such metrics for viscoelastic fluid flows governed by a large class of linear and nonlinear stress models. 
\revtwo{To do this, we introduce {\em energy-compatible} families of kernel functions corresponding to a given viscoelastic stress model.
Each kernel implicitly embeds flowfield snapshots into a {\em Reproducing Kernel Hilbert Space} (RKHS) in which distances and angles are computed and whose squared norm equals the total mechanical energy.}
Additionally, we present a solution to the preimage problem for \revtwo{these} kernels, enabling accurate reconstruction of flowfields from their RKHS representations.
Through numerical experiments on an unsteady viscoelastic lid-driven cavity flow, we demonstrate the utility of \revtwo{energy-compatible} kernels for extracting energetically-dominant coherent structures in viscoelastic flows across a range of Reynolds and Weissenberg numbers.
Specifically, the features extracted by Kernel Principal Component Analysis (KPCA) of flowfield snapshots using \revtwo{energy-compatible} kernel functions yield reconstructions with superior accuracy in terms of mechanical energy compared to conventional methods such as ordinary Principal Component Analysis (PCA) with na\"{i}vely-defined state vectors or KPCA with ad-hoc choices of kernel functions.
Our findings underscore the importance of principled choices of metrics in both scientific and machine learning investigations of complex fluid systems.
\end{abstract}

\textit{Keywords: viscoelastic flow, energy-based inner product, kernel method, machine learning, reproducing kernel Hilbert space, principal component analysis
}

\section{Introduction}

A basic component of quantitative science is the ability to measure differences in collected data.
Quantities such as mass, momentum, and energy provide unified descriptions of physical systems, and are therefore desirable quantities to measure when investigating the behavior of a system of interest.
Likewise, in machine learning it is often necessary to endow the space in which data lies with geometric notions such as distance and angle, or otherwise to embed the data in a space with these notions.
The assumed geometry, including for instance, the way measurements are normalized, can drastically affect the outcome of learning, potentially highlighting spurious features while ignoring important ones that appear insignificant due to a poor choice of metric.
Therefore, in machine learning for physics applications it is important to ground the ways we compare data in principled physical notions such as energy.

One of the most important tasks in machine learning involves extracting low-dimensional features (variables) that describe a system or allow one to predict quantities of interest.
Since its first applications to fluid dynamics in \cite{Lumey}, the Proper Orthogonal Decomposition (POD), also known as principal component analysis (PCA) or the Karhunen-Loéve (KL) expansion, has been widely used for encoding high dimensional data in a low dimensional representation (\cite{Aubry,Sirovich_1991}).
Other modal analysis techniques such as those reviewed in \cite{Taira} have emerged as powerful tools to reduce the dimension of complex flow by splitting them into simpler components or ``modes''. 

Many of these modal analysis techniques, including POD, rely on the choice of an inner product defined on the flow's state space.
The fact that significantly different results can be obtained using different inner products motivates the introduction of principled choices based on physics.
For incompressible Newtonian fluid flows the integrated dot product of velocity fields is a natural inner product whose resulting (squared) norm is the flow's kinetic energy.
Applying PCA with this inner product extracts the orthogonal modes and mode coefficients that are optimal for reconstructing flowfields in an energetic sense.
An inner product with analogous properties for isothermal compressible flows was introduced by Rowley et al. \cite{Rowley2004model} and used to construct POD-Galerkin reduced-order models.
Moreover, this work showed that compatibility of the inner product with a conserved or dissipated energy function guarantees that the stability of an equilibrium is preserved by Galerkin projection.
Energy and dimensionally-consistent inner products for more general classes of compressible flows were later investigated in \cite{PARISH2023112387}.
Following this same line of work, energy-consistent inner products have been introduced for magnetohydrodynamics in \cite{Kaptanoglu2021pre} and for rotating shallow-water equations in \cite{Sockwell2019mass}.

In the new era of machine learning applications in fluid mechanics \cite{Brunton2020arfm}, there is growing interest in nonlinear dimensionality reduction techniques rooted in manifold learning and artificial neural networks \cite{Mendez2022LinearAN,DUBOIS2022110733,Farzamnik2023from}. 
These methodologies allow for greater dimensionality reduction, especially in advection-dominated flow problems, by capturing curved manifolds in the state space that are poorly approximated by low-dimensional subspaces, i.e., by superpositions of modes \cite{Ohlberger2016reduced}.
Complementary to learned manifolds, dimensionality reduction methods condense states along ``fibers'' determined by the choice of reduced modeling variables.
By allowing these variables to be nonlinear functions of the state, nonlinear dimensionality reduction methods also allow for nonlinear fibers that appropriately group states with similar dynamical behaviors \cite{Otto2023nonlinear}.

Kernel Principal Component Analysis (KPCA) was introduced by Schölkopf et al. \cite{Scholkopf1998nonlinear} and has become one of the most widely used methods for nonlinear dimensionality reduction. The key insight is that PCA --- a linear dimensionality reduction method --- can be applied through the use of a kernel function in a high or infinite-dimensional Reproducing Kernel Hilbert Space (RKHS) into which states have been lifted as illustrated in Fig.~\ref{fig:energy_compatible_lifting}.
Recent works, such as \cite{SALVADOR20211, Csala2022comparing, Otto2022model}, have leveraged the enhanced capabilities of kernel-based nonlinear dimensionality reduction for modeling complex fluid flows. 
However, to the best of our knowledge, all formulations presented to date have only been used in the context of Newtonian fluid flows. 
Moreover, with a wide range of kernel functions to choose from and the ability to combine kernel functions to produce new ones, it is a matter of great practical importance to narrow this selection based on principled physical considerations.

Viscoelastic fluids play a pervasive role across various industrial sectors, encompassing consumer goods, food, healthcare, and more.
These applications motivate the use of machine learning models and nonlinear dimensionality reduction to shed light on the complex interplay of inertial and elastic effects that give rise to unique unsteady nonlinear dynamics in viscoelastic fluid flows (see \cite{Alves2021numerical}).
This is especially pertinent in the context of Elastic Turbulence (ET) (\cite{Groisman,PhysRevLett.110.174502}) and Elasto-Inertial Turbulence (EIT) (\cite{doi:10.1073/pnas.1219666110,doi:10.1146/annurev-fluid-032822-025933}), two phenomena posing distinctive challenges in the realm of viscoelastic fluid dynamics. 
Roughly speaking, defining the elasticity parameter as $E=\frac{Wi}{Re}$, where $Re$ represents the Reynolds number and $Wi$ is the Weissenberg number, ET manifests itself in the turbulent flow of viscoelastic fluids when $E$ is large, i.e., for inertialess or creeping flows. 
On the other hand, EIT characterizes the turbulent behavior of highly viscoelastic fluids influenced by inertial forces, particularly in inertia-dominated flows at moderate or high Reynolds numbers. 
Both ET and EIT, along with the identification of regime transitions (\cite{LI_XI_GRAHAM_2006,10.1063/1.4895780,10.1122/1.4798549,THOMAS_KHOMAMI_SURESHKUMAR_2009}), present challenges and opportunities to enhance our understanding of these phenomena through machine learning models based on data gathered from computational fluid dynamics simulations.
More broadly, machine learning models have the potential to enhance our ability to optimize and control non-Newtonian fluid systems arising in a wide range of engineering applications.

\begin{figure}
    \centering
    \begin{tikzonimage}[width=0.7\textwidth]{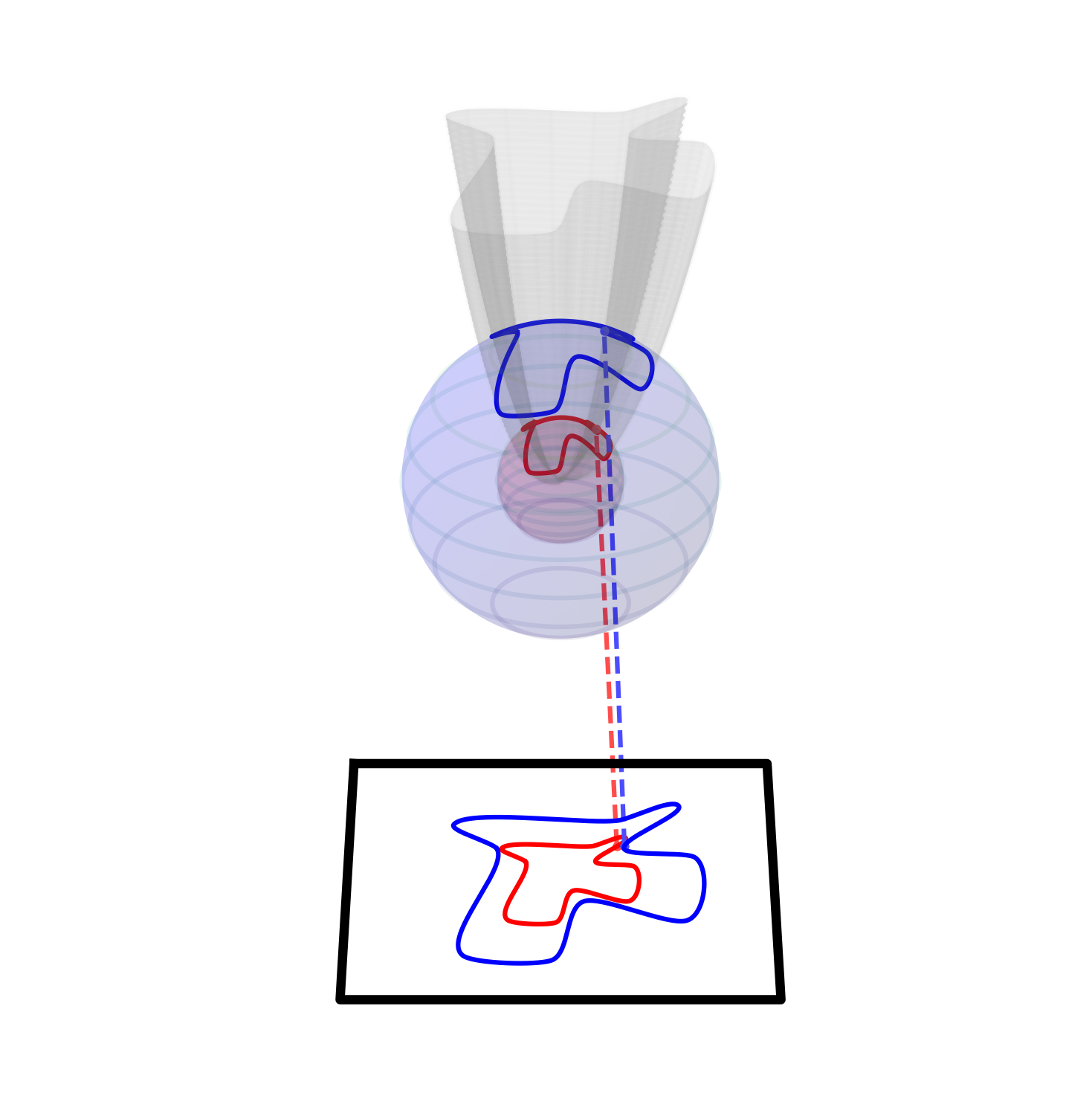}
        \node at (0.36,0.25) {\large $\mcal{F}$};
        \node at (0.36,0.8) {\large $\Phi(\mcal{F})$};
        \draw[->] (0.35, 0.28) arc[radius=0.58, start angle=180+25, end angle=180-25];
        \node at (0.825,0.8) {\large RKHS, $\mcal{H}$};
        \node[anchor=west] at (0.78,0.2) {$\mcal{E}(\vect{q}) = E_2$};
        \node[anchor=west] at (0.655,0.625) {$\| \phi \|_{\mcal{H}}^2 = E_2$};
        \draw[-] (0.645, 0.195) -- (0.78, 0.2);
        \node[anchor=west] at (0.78,0.1) {$\mcal{E}(\vect{q}) = E_1$};
        \node[anchor=west] at (0.575,0.55) {$\| \phi \|_{\mcal{H}}^2 = E_1$};
        \draw[-] (0.587, 0.195) -- (0.78, 0.1);
    \end{tikzonimage}
    \caption{Lifting states into a Reproducing Kernel Hilbert Space (RKHS) with squared norm equal to total mechanical energy. We illustrate how constant-energy surfaces in the state space with energies $E_1$ and $E_2$ are lifted onto spheres centered about the origin in the RKHS with squared radii $E_1$ and $E_2$. Improperly chosen metrics in the state space can cause states with different mechanical energies to appear similar. This problem is rectified by measuring distances between lifted states in the RKHS.}
    \label{fig:energy_compatible_lifting}
\end{figure}
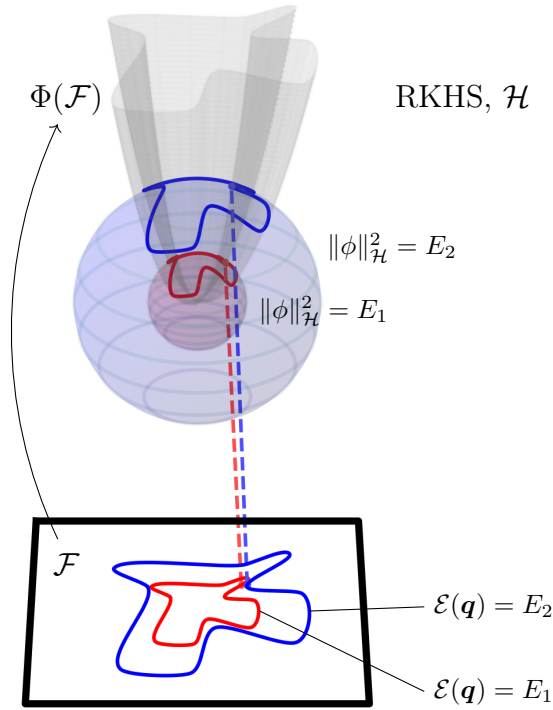

Following the line of work initiated by Rowley et al. \cite{Rowley2004model} and Schölkopf et al. \cite{Scholkopf1998nonlinear}, in this paper we show that the mechanical energy of viscoelastic fluid flows can be used to formulate kernel functions giving rise to metrics for distances and angles between states that are well-suited to physics-informed machine learning applications.
\revtwo{Specifically, we show that there is a family of kernel functions adapted to a given viscoelastic stress model where each kernel implicitly embeds the flow’s state in a Reproducing Kernel Hilbert Space (RKHS) with squared norm equal to the mechanical energy.}
The induced distance turns the set of states with finite energy into a complete, separable metric space.
\revtwo{These kernel functions can then be used in kernel-based machine learning algorithms (see \cite{Hofmann2008kernel}) exploiting the implicit geometry of the RKHS for classification, regression, and dimensionality reduction.
We note that while there are seemingly natural choices for energy-compatible kernel functions adapted to each viscoelastic stress model, such kernels are not unique. Hence, our contribution does not provide a single standardized choice of kernel, but rather an entire family of energy-compatible kernels to choose from. To illustrate the utility of such kernels in applications without enmeshing ourselves in endless comparisons, we choose a single representative kernel adapted to each stress model.}

Beyond providing principled choices of energy-compatible kernel functions for viscoelastic flows, we also provide a solution of the preimage problem (see \cite{Kwok2004preimage, Honeine2011preimage}) for these kernels.
That is, the problem of reconstructing the flowfield from its lifted representation in the RKHS, or its truncated representation in terms of kernel principal components.
We prove that the velocity and matrix square root of the conformation tensor field can always be reconstructed from the RKHS representation of a state via a bounded linear operator.
As a consequence, we show that these fields can be linearly reconstructed from truncated kernel principal components with guaranteed accuracy depending on the truncated singular values.
Other fields can also be linearly reconstructed depending on the viscoelastic stress model, but the velocity and square root of the conformation tensor are sufficient to explicitly reconstruct the entire flowfield.

To illustrate the utility of energy-compatible kernel functions and the importance of making principled choices for measuring distances between states, we perform KPCA-based dimensionality reduction and reconstruction of an unsteady viscoelastic lid-driven cavity flow.
The flow is simulated at both low and moderate Reynolds numbers and at both low and high Weissenberg numbers using different viscoelastic stress models.
The selection of the specific pairs of Reynolds and Weissenberg numbers is grounded in a careful consideration of previous works that have investigated the complexities of viscoelastic flows under inertial effects (\cite{RICHTER_IACCARINO_SHAQFEH_2010,Patel_Rothstein_Modarres-Sadeghi_2023,Hamid_Sasmal_2023}) as well as investigations specifically focused on problems related to creeping flows (\cite{dzanic_from_sauret_2022,Song_Liu_Lu_Khomami_2022})
While ordinary PCA is optimal for reconstructing the flowfields with respect to the Frobenius norm, we show that this can lead to poor reconstruction of the flow's mechanical energy.
In contrast, by choosing a kernel function adapted to the viscoelastic stress model, Reynolds number, and Weissenberg number of the simulation, kernel PCA is able to extract low-dimensional coordinates that faithfully reconstruct the flowfield with superior accuracy in an energetic sense.

\section{Viscoelastic flow models and energy functions}

We consider nondimensional equations for viscoelastic flow \cite{Alves2021numerical} on a set $\Omega \subset \R^d$ whose state $\vect{q} = (\vect{u}, \mat{c})$ consists of the velocity field $\vect{u}:\Omega \to \R^d$ and the conformation tensor field $\mat{c}:\Omega \to \mathbb{S}_+^{d} \subset \R^{d\times d}$.
The conformation tensor field takes values in the space $\mathbb{S}_+^{d}$ of positive semidefinite real symmetric matrices. 
We use $\mathbb{S}^{d}$ to denote the real symmetric matrices and we use $\mathbb{S}_{++}^{d}$ to denote the real symmetric positive definite matrices.
The time evolution of the state $\vect{q}(t)$ is governed by the momentum equation
\begin{equation}
    \ddt \vect{u} + \vect{u} \cdot \grad \vect{u} + \grad p 
    = \frac{\beta}{\Rey} \Lap \vect{u} + \frac{1}{\Rey} \vdiv \pmb{\tau} + \vect{f},
    \label{mom}
\end{equation}
subject to the incompressibility constraint
\begin{equation}
    \vdiv \vect{u} = 0,
    \label{mass}
    \end{equation}
and the conformation tensor transport equation
\begin{equation}
    \overset{\bigtriangledown}{\mat{c}} = \mat{s}(\mat{c}),
    \label{cons}
\end{equation}
where $\beta$ is the viscosity ratio satisfying $0 < \beta < 1$. The 
upper convected derivative \cite{doi:10.1137/20M1364990}, adopted in Eq.~\eqref{cons} is defined as
\begin{equation}\label{eq2.7}
    \overset{\bigtriangledown}{\mat{c}} = \ddt \mat{c} + \vect{u} \cdot \grad \mat{c} - \mat{c}\grad \vect{u} - (\grad \vect{u})^T \mat{c}
\end{equation}

These equations are coupled through a stress model
\begin{equation}
    \pmb{\tau} = - (1-\beta) \mat{s}(\mat{c}),
\end{equation}
determined by a function $\mat{s}:D(\mat{s})\subset \mathbb{S}_+^{d} \to \mathbb{S}^{d}$, \revtwo{where $D(\mat{s})$ is the domain of $\mat{s}$}. There are a wide variety of available stress models taking the general form
\begin{equation}
    \mat{s}(\mat{c}) = -\frac{h_0(\mat{c})}{\Wei} \left[ h_1(\mat{c}) \mat{c} - h_2(\mat{c}) \mat{I} + \kappa(\mat{c} - \mat{I})^2 \right],
    \label{eqn:general_form_for_stress_models}
\end{equation}
where $\kappa$ is a constant and $h_0$, $h_1$, and $h_2$ are scalar-valued functions of the conformation tensor.
Some examples are given in table~\ref{tab:stress_models}.
Using the stress model, the total mechanical energy of the system is given by
\begin{equation}
\revtwo{
    \mcal{E}(\vect{q}) = \frac{1}{2} \int_{\Omega} \left[ \big\vert \vect{u}(\vect{x}) \big\vert^2 + \frac{1}{\Rey} \Tr(\pmb{\tau}) + \frac{(1-\beta) d}{\Rey \Wei} \right] \td \vect{x}.}
    \label{eqn:energy}
\end{equation}
Here, we are adopting the energy defined in \cite{Balci2011symmetric} with the addition of a model-dependent constant $c \geq 0$ selected to ensure that the energy is always non-negative. More details about the energy estimation for viscoelastic fluids can be found in \cite{LOZINSKI2003161}.
For the models in table~\ref{tab:stress_models}, this constant may be taken to be equal to the dimension $d$ of the flow domain.
To simplify the notation, we define
\[ \theta := \frac{1-\beta}{\Rey \Wei}. \]

\begin{table}
    \begin{center}
    \def~{\hphantom{0}}
    \begin{tabular}{l|cccc|c}
        Model  & $h_0(\mat{c})$ & $h_1(\mat{c})$ & $h_2(\mat{c})$ & $\kappa$ & constraints on $D(\mat{s})$  \\
        \hline
        Oldroyd-B & $1$ & $1$ & $1$ & $0$ & None  \\
        Giesekus & $1$ & $1$ & $1$ & $\alpha$ & None  \\ 
        FENE-CR & $1$ & $\frac{L^2}{L^2 - \Tr(\mat{c})}$ & $\frac{L^2}{L^2 - \Tr(\mat{c})}$ & $0$ & $\Tr(\mat{c}) < L^2$  \\
        FENE-P & $1$ & $\frac{L^2}{L^2 - \Tr(\mat{c})}$ & $1$ & $0$ & $\Tr(\mat{c}) < L^2$  \\
        Linear PTT & $1 + \varepsilon\Tr(\mat{c}-\mat{I})$ & $1$ & $1$ & $0$ & None  \\
        Nonlinear PTT & $\exp{\left[\varepsilon\Tr(\mat{c}-\mat{I})\right]}$ & $1$ & $1$ & $0$ & None 
    \end{tabular}
    \caption{Some common stress models for viscoelastic flow (see \cite{Alves2021numerical}) in the general form of Eq.~\eqref{eqn:general_form_for_stress_models}. With the exception of the linear Oldroyd-B model, the nonlinear models introduce an additional nondimensional parameter, each defined accordingly: $\alpha$ represents the mobility factor for the Giesekus model, whereas $L^2$ and $\varepsilon$ correspond to the extensibility parameters for the finitely extensible nonlinear elastic with Peterlin closure (FENE) and the Phan-Thien–Tanner(PTT) models, respectively.} 
    \label{tab:stress_models}
    \end{center}
\end{table}

We denote the space of states with finite total mechanical energies by
\[
    \mcal{F} = \left\{ \vect{q} \ : \ \mcal{E}(\vect{q}) < \infty \right\}.
\]
We declare two elements in $\mcal{F}$ to be equal when they agree almost everywhere on $\Omega$.
In this paper we introduce energy-based distance functions that allow states in $\mcal{F}$ to be compared.
The next section explains how this can be accomplished for a large class of stress models including those listed in table~\ref{tab:stress_models} by embedding $\mcal{F}$ in a Reproducing Kernel Hilbert Space (RKHS) whose squared norm equals the total mechanical energy, as illustrated in Fig.~\ref{fig:energy_compatible_lifting}. 
This enables a variety of kernel-based machine learning algorithms (see \cite{Hofmann2008kernel}) to be applied in viscoelastic flow problems with principled choices for the kernel functions based on flow physics.

\section{Energy-based reproducing kernels}
First we review the concept of a Reproducing Kernel Hilbert Space (RKHS).
An RKHS $\mcal{H}$ over $\mcal{F}$ consists of functions $\phi: \mcal{F} \to \R$ where pointwise evaluation $\phi \mapsto \phi(\vect{q})$ is a bounded linear map for every $\vect{q}\in\mcal{F}$.
By the Riesz lemma \cite{Reed1980functional}, there is a unique element $K_{\vect{q}} \in \mcal{H}$ so that the value of every $\phi\in\mcal{H}$ at $\vect{q}\in\mcal{F}$ is given by the inner product
\begin{equation}
    \phi(\vect{q}) = \left\langle K_{\vect{q}}, \ \phi \right\rangle_{\mcal{H}}.
\end{equation}
The function $k:\mcal{F}\times\mcal{F} \to \R$ defined by
\begin{equation}
    k(\vect{q}_1, \vect{q}_2) = K_{\vect{q}_2}(\vect{q}_1) = \left\langle K_{\vect{q}_1}, \ K_{\vect{q}_2} \right\rangle_{\mcal{H}}
\end{equation}
is called the ``reproducing kernel'' of $\mcal{H}$ and
the map $\Phi: \vect{q} \mapsto K_{\vect{q}}$ is called the ``feature map''.
It is easy to verify that for every finite collection of states $\vect{q}_1, \ldots, \vect{q}_m \in \mcal{F}$ and coefficients $a_1, \ldots, a_m \in \R$ the kernel satisfies the positive-definiteness condition
\begin{equation}
    \sum_{i=1}^{m}\sum_{i=1}^{m} a_i a_j k(\vect{q}_i, \vect{q}_j) \geq 0.
    \label{eqn:PD_cond}
\end{equation}
Conversely, any function $k:\mcal{F} \times \mcal{F} \to \R$ satisfying the above positive-definiteness condition is the reproducing kernel of a unique RKHS thanks to the Moore-Aronszajn theorem (see \cite{Berlinet2011reproducing, Aronszajn1950theory}).

For a large class of stress models, including those in table~\ref{tab:stress_models}, we show that there is an RKHS $\mcal{H}$ whose feature map $\Phi:\mcal{F} \to \mcal{H}$ is injective and respects the total energy in the sense that 
\begin{equation}
    \mcal{E}(\vect{q}) 
    = \| \Phi(\vect{q})\|_{\mcal{H}}^2 
    = k(\vect{q}, \vect{q}).
    \label{eqn:energy_compatibility}
\end{equation}
This notion of compatibility with the total mechanical energy is illustrated in Fig.~\ref{fig:energy_compatible_lifting}.
Using the distance between lifted states in the RKHS, the state space $\mcal{F}$ becomes a metric space with
\begin{equation}
    \label{eqn:kernel_metric}
    d_{\mcal{E}}(\vect{q}_1, \vect{q}_2) 
    := \left\Vert \Phi(\vect{q}_1) - \Phi(\vect{q}_2) \right\Vert_{\mcal{H}}
    = \sqrt{k(\vect{q}_1, \vect{q}_1) - 2 k(\vect{q}_1, \vect{q}_2) + k(\vect{q}_2, \vect{q}_2)}.
\end{equation}
The injective property of the feature map is required to ensure that $d_{\mcal{E}}(\vect{q}_1, \vect{q}_2) = 0$ if and only if $\vect{q}_1 = \vect{q}_2$ in $\mcal{F}$, i.e, almost everywhere in $\Omega$.
We say that an RKHS is an ``injective RKHS'' with an ``injective kernel function'' when the associated feature map is injective.
Most importantly, the metric in Eq.~\eqref{eqn:kernel_metric} can be computed using the kernel function without doing explicit computations in the abstract, possibly infinite-dimensional space $\mcal{H}$.

Our main result, stated in the following theorem, provides general conditions on the stress model ensuring there is an RKHS with injective feature map compatible with the total energy.
Moreover, the theorem provides an explicit formula for the corresponding reproducing kernel.
Before stating the result, we need some preliminary definitions.
Let $\sigma(\mat{c}) \subset \R$ denote the set of eigenvalues (spectrum) of a symmetric matrix $\mat{c} \in \mathbb{S}^d$ and let $\mat{P}_{\mat{c}}(\lambda)$ denote the orthogonal projection onto the eigenspace of $\mat{c}$ corresponding to an eigenvalue $\lambda \in \sigma(\mat{c})$.
Then the action of a function $f: \sigma(\mat{c}) \to \R$ on the matrix $\mat{c}$ is defined by
\begin{equation}
    \label{eqn:matrix_functional_calculus}
    f(\mat{c}) := \sum_{\lambda \in \sigma(\mat{c})} f(\lambda) \mat{P}_{\mat{c}}(\lambda).
\end{equation}
It is easy to see that $1(\mat{c}) = \mat{I}$, $\Id_{\R}(\mat{c}) = \mat{c}$, $(f+g)(\mat{c}) = f(\mat{c}) + g(\mat{c})$, and $(f g)(\mat{c}) = f(\mat{c}) g(\mat{c})$, making $f \mapsto f(\mat{c})$ a commutative algebra homomorphism commonly referred to as the ``functional calculus'' of $\mat{c}$.
This extends the usual notions of matrix polynomials, matrix square roots, and matrix exponentials.
\begin{theorem}
    \label{thm:viscoelastic_kernels}
    Suppose that there are (measurable) functions $\{ f_i \}_{i=0}^{\infty}$ on $\sigma(D(\mat{s})) = \bigcup_{\mat{c}\in D(\mat{s})} \sigma(\mat{c})$ and nonnegative constants $\{ c_{i,p} \}_{i,p=0}^{\infty}$ so that
    \begin{equation}
        \label{eqn:stress_model_series}
        h(\mat{c}) := 
        -\Wei \cdot \Tr[\mat{s}(\mat{c})] + c
        = \sum_{i=0}^{\infty}\sum_{p=0}^{\infty} c_{i,p} \left[ \Tr \big(f_i(\mat{c})^2\big) \right]^p
    \end{equation}
    for every $\mat{c}\in D(\mat{s})$.
    Then the series
    \begin{equation}
        \label{eqn:integrand_kernel}
        \tilde{k}(\mat{c}_1, \mat{c}_2)
        := \sum_{i=0}^{\infty}\sum_{p=0}^{\infty} c_{i,p} \left[ \Tr \big( f_i(\mat{c}_1) f_i(\mat{c}_2) \big) \right]^p
    \end{equation}
    converges absolutely for every $\mat{c}_1, \mat{c}_2 \in D(\mat{s})$ and the function $k:\mcal{F}\times\mcal{F} \to \R$ defined by
    \begin{equation}
    \label{eqn:main_kernel_fun}
        k(\vect{q}_1, \vect{q}_2) := \frac{1}{2} \int_{\Omega} \left[ \vect{u}_1(\vect{x}) \cdot \vect{u}_2(\vect{x}) + \theta \tilde{k}\big(\mat{c}_1(\vect{x}), \mat{c}_2(\vect{x})\big) \right] \td \vect{x}
    \end{equation}
    is a positive-definite kernel satisfying $\mcal{E}(\vect{q}) = k(\vect{q}, \vect{q})$.
    Moreover, if there is a coefficient $c_{i,p} > 0$ with $p \geq 1$, $f_i$ injective on $\sigma(D(\mat{s}))$, and $p$ odd or $f_i$ nonnegative, then the feature map of the corresponding RKHS is injective.
    We provide a proof in \ref{app:viscoelastic_kernels_proof}.
\end{theorem}

\begin{revtwobox} 
\begin{remark}
    The functions $f_i$ in the Theorem are not necessarily unique, and different choice of $f_i$ can result in different kernel functions.
    For example, we will see below that for the Oldroyd-B model we can take $f_0(x) = \sqrt{x}$ with $c_{0,1} = 1$.
    Another equally valid choice leading to a different kernel function would be $f_0(x) = -\chi_{\mcal{I}}(x) \sqrt{x} + \chi_{\R\setminus \mcal{I}}(x) \sqrt{x}$, where $\mcal{I}$ is a measurable subset of $\R$ and $\chi_{\mcal{I}}(x) = 1$ when $x \in \mcal{I}$ and $\chi_{\mcal{I}}(x) = 0$ otherwise. The choices for $\mcal{I}$ are endless, and we do not explore them further.
\end{remark}
\end{revtwobox}  

Furthermore, $\mcal{F}$ is a complete metric space, as we state in the next theorem.
Intuitively, this means that $\mcal{F}$ has no missing points that can be approached, but never reached.
Completeness is important for understanding limits that appear, for example, when studying the long-time behavior of the system.
The fact that $\mcal{F}$ is complete provides additional evidence that our kernel-based metric is a natural one for the states of viscoelastic flows.

\begin{theorem}
    \label{thm:completeness_of_metric}
    With the same assumptions as Theorem~\ref{thm:viscoelastic_kernels},
    suppose that there is a coefficient $c_{i,p} > 0$ with $p \geq 1$, $f_i$ injective on $\sigma(D(\mat{s}))$, and $p$ odd or $f_i$ nonnegative.
    Then $\mcal{F}$ is a complete metric space with metric given by Eq.~\eqref{eqn:kernel_metric} and $\Phi(\mcal{F})$ is a closed subset of $\mcal{H}$. 
    A proof is provided in \ref{app:completeness_proof}.
\end{theorem}

To see how these theorems work, we use them to show that the stress models in table~\ref{tab:stress_models} admit injective positive-definite kernel functions turning $\mcal{F}$ into a complete metric space.
However, we note that in certain cases the parameters in the stress models must be constrained.
Interestingly, for the Oldroyd-B, Giesekus, and linear PTT stress models it is also possible to explicitly construct feature maps $\Psi: \mcal{F} \to L^2(\Omega)$ acting point-wise so that
\begin{equation}
    \label{eqn:explicit_feature_map}
    k(\vect{q}_1, \vect{q}_2) 
    = \langle \Psi(\vect{q}_1),\ \Psi(\vect{q}_2) \rangle_{L^2(\Omega)}
    := \int_{\Omega} \Psi(\vect{q}_1)(\vect{x})^T \Psi(\vect{q}_2)(\vect{x}) \td \vect{x}.
\end{equation}
As a result of the Moore-Aronszajn theorem (see Theorem~\ref{thm:Moore-Aronszajn} in \ref{app:viscoelastic_kernels_proof}), for these cases the RKHS $\mcal{H}$ is identified isometrically with the closed subspace of $L^2(\Omega)$ spanned by $\{ \Psi(\vect{q}) \}_{\vect{q} \in \mcal{F}}$.
We let $\vct(\mat{A})$ denote any vectorization of a matrix $\mat{A}$ and we recall that $\vct(\mat{A}_1)^T \vct(\mat{A}_2) = \Tr(\mat{A}_1^T\mat{A}_2)$.

\begin{example}
For the \textbf{Oldroyd-B} model, we have
\begin{equation}
    h(\mat{c})
    = \Tr(\mat{c} - \mat{I}) + d = \Tr(\mat{c}),
\end{equation}
which can be written in the form of Eq.~\eqref{eqn:stress_model_series} with $f_0:x \mapsto \sqrt{x}$ and coefficient $c_{0,1} = 1$.
Thus, Eq.~\eqref{eqn:main_kernel_fun} with
\begin{equation}
    \tilde{k}(\mat{c}_1, \mat{c}_2) := \Tr\big(\sqrt{\mat{c}_1}\sqrt{ \mat{c}_2}\big),
    \label{eqn:Oldroyd-B_integrand_kernel}
\end{equation}
is an injective positive semi-definite kernel function satisfying Eq.~\eqref{eqn:energy_compatibility} for the Oldroyd-B.
Moreover, it is easy to see that
    \begin{equation}\label{eqn:OldB_explicit_feature_map}
        \Psi(\vect{q})(\vect{x})
        = \frac{1}{\sqrt{2}} 
        \begin{bmatrix}
            \vect{u}(\vect{x}) \\
            \sqrt{\theta} \vct\big(\sqrt{\mat{c}(\vect{x}}) \big)
        \end{bmatrix}
        \in \R^{d + d^2}.
    \end{equation}
provides an explicit feature map satisfying Eq.~\eqref{eqn:explicit_feature_map}.
\end{example}

\begin{example}
For the \textbf{Giesekus} model, we have
\begin{equation}
    h(\mat{c})
    = \Tr(\mat{c} - \mat{I}) + \alpha \Tr\left[ (\mat{c} - \mat{I})^2 \right] + d 
    = \Tr(\mat{c}) + \alpha \Tr\left[ (\mat{c} - \mat{I})^2 \right],
\end{equation}
which can be written in the form of Eq.~\eqref{eqn:stress_model_series} with $f_0 : x \mapsto \sqrt{x}$ and $f_1: x \mapsto (x-1)$ with coefficients $c_{0,1} = 1$ and $c_{1,1} = \alpha$.
Thus, Eq.~\eqref{eqn:main_kernel_fun} with
\begin{equation}
    \tilde{k}(\mat{c}_1, \mat{c}_2) := \Tr\big(\sqrt{\mat{c}_1}\sqrt{ \mat{c}_2}\big) + \alpha \Tr\big[ (\mat{c}_1 - \mat{I}) (\mat{c}_2 - \mat{I}) \big],
    \label{eqn:Giesekus_integrand_kernel}
\end{equation}
is an injective positive semi-definite kernel function satisfying Eq.~\eqref{eqn:energy_compatibility} for the Giesekus model with parameter $\varepsilon \geq 0$.
Moreover, 
\begin{equation}\label{eqn:Giesekus_explicit_feature_map}
    \Psi(\vect{q})(\vect{x})
    = \frac{1}{\sqrt{2}} 
    \begin{bmatrix}
        \vect{u}(\vect{x}) \\
        \sqrt{\theta} \vct\big(\sqrt{\mat{c}(\vect{x}}) \big) \\
        \sqrt{\theta \alpha} \vct\big(\mat{c}(\vect{x}) - \mat{I} \big)
    \end{bmatrix}
    \in \R^{d + 2 d^2},
\end{equation}
defines an explicit feature map satisfying Eq.~\eqref{eqn:explicit_feature_map}.
\end{example}

\begin{example}
For the \textbf{FENE-CR} model, we have
\begin{equation}
    h(\mat{c})
    = \frac{L^2\Tr(\mat{c} - \mat{I})}{L^2 - \Tr(\mat{c})} + d
    = \frac{(L^2 - d)\Tr(\mat{c})}{L^2 - \Tr(\mat{c})}
    = (L^2 - d) \sum_{p=1}^{\infty} \frac{1}{L^{2p}} \Tr(\mat{c})^p,
\end{equation}
where the geometric series converges when $\vert \Tr(\mat{c}) \vert < L^2$, hence on $D(\mat{s})$ for this model.
This expression takes the form of Eq.~\eqref{eqn:stress_model_series} with $f_0: x\mapsto \sqrt{x}$ and $c_{0,p} = (L^2 - d) L^{-2p}$ for every $p \geq 1$.
Therefore, if the parameter $L$ of the FENE-CR model is chosen so that $L^2 > d$, then Eq.~\eqref{eqn:main_kernel_fun} with
\begin{equation}
    \tilde{k}(\mat{c}_1, \mat{c}_2) = \frac{(L^2-d)\Tr\big(\sqrt{\mat{c}_1}\sqrt{\mat{c}_2}\big)}{L^2 - \Tr\big(\sqrt{\mat{c}_1}\sqrt{\mat{c}_2}\big)},
    \label{eqn:FENE-CR_integrand_kernel}
\end{equation}
defines an injective positive semi-definite kernel function satisfying Eq.~\eqref{eqn:energy_compatibility}.
\end{example}

\begin{example}
For the \textbf{FENE-P} model, we have
\begin{equation}
    h(\mat{c})
    = \frac{L^2 \Tr(\mat{c})}{L^2 - \Tr(\mat{c})}
    = L^2 \sum_{p=1}^{\infty} \frac{1}{L^{2p}} \Tr(\mat{c})^p,
\end{equation}
where the geometric series converges when $\vert \Tr(\mat{c}) \vert < L^2$, hence on $D(\mat{s})$ for this model.
This expression takes the form of Eq.~\eqref{eqn:stress_model_series} with $f_0: x\mapsto \sqrt{x}$ and $c_{0,p} = L^{-2(p-1)}$ for every $p \geq 1$.
Therefore, Eq.~\eqref{eqn:main_kernel_fun} with
\begin{equation}
    \tilde{k}(\mat{c}_1, \mat{c}_2) = \frac{L^2\Tr\big(\sqrt{\mat{c}_1}\sqrt{\mat{c}_2}\big)}{L^2 - \Tr\big(\sqrt{\mat{c}_1}\sqrt{\mat{c}_2}\big)},
    \label{eqn:FENE-P_integrand_kernel}
\end{equation}
defines an injective positive semi-definite kernel function satisfying Eq.~\eqref{eqn:energy_compatibility} for the FENE-P model.
We note that unlike the FENE-CR model, there are no constraints on the model parameter $L$.
\end{example}

\begin{example}
For the \textbf{linear PTT} model, we have
\begin{equation}
    h(\mat{c}) 
    = \left[ 1 + \varepsilon \Tr(\mat{c}-\mat{I}) \right]\Tr(\mat{c} - \mat{I}) + d
    = \varepsilon d^2 + (1 - 2\varepsilon d) \Tr(\mat{c}) + \varepsilon \Tr(\mat{c})^2,
\end{equation}
which takes the form of Eq.~\eqref{eqn:stress_model_series} with $f_0: x\mapsto \sqrt{x}$ and coefficients $c_{0,0} = \varepsilon d^2$, $c_{0,1} = 1-2\varepsilon d$, and $c_{0,2} = \varepsilon$.
Therefore, if the parameter $\varepsilon$ of the linear PTT model satisfies $0\leq \varepsilon \leq (2 d)^{-1}$, then Eq.~\eqref{eqn:main_kernel_fun} with
\begin{equation}
    \tilde{k}(\mat{c}_1, \mat{c}_2) = 
    \left[ 1 + \varepsilon \Tr\big(\sqrt{\mat{c}_1}\sqrt{\mat{c}_2}-\mat{I}\big) \right]\Tr\big(\sqrt{\mat{c}_1}\sqrt{\mat{c}_2} - \mat{I}\big) + d,
    \label{eqn:linear_PTT_integrand_kernel}
\end{equation}
defines an injective positive semi-definite kernel function satisfying Eq.~\eqref{eqn:energy_compatibility}.
To provide an explicit feature map for this model, we first observe that
\begin{equation}
\begin{aligned}
    \tilde{k}(\mat{c}_1, \mat{c}_2) 
    &= d^2 + (1-2\varepsilon d) \Tr\big(\sqrt{\mat{c}_1}\sqrt{\mat{c}_2}\big) + \varepsilon \Tr\big(\sqrt{\mat{c}_1}\sqrt{\mat{c}_2}\big)^2 \\
    &= d^2 + (1-2\varepsilon d) \Tr\big(\sqrt{\mat{c}_1}\sqrt{\mat{c}_2}\big) + \varepsilon \Tr\Big[ \big(\sqrt{\mat{c}_1} \sqrt{\mat{c}_2} \big) \otimes \big(\sqrt{\mat{c}_1} \sqrt{\mat{c}_2} \big) \Big] \\
    &= d^2 + (1-2\varepsilon d) \Tr\big(\sqrt{\mat{c}_1}\sqrt{\mat{c}_2}\big) + \varepsilon \Tr\Big[ \big(\sqrt{\mat{c}_1} \otimes \sqrt{\mat{c}_1} \big) \big(\sqrt{\mat{c}_2} \otimes \sqrt{\mat{c}_2} \big)\Big],
\end{aligned}
\end{equation}
where $\otimes$ denotes the Kronecker product of matrices.
Here, we have used the trace and mixed product properties of the Kronecker product.
Thanks to the above expression,
\begin{equation}\label{eqn:LinearPTT_explicit_feature_map}
    \Psi(\vect{q})(\vect{x})
    = \frac{1}{\sqrt{2}} 
    \begin{bmatrix}
        d \sqrt{\theta} \\
        \vect{u}(\vect{x}) \\
        \sqrt{\theta (1 - 2\varepsilon d)} \vct\big(\sqrt{\mat{c}(\vect{x}}) \big) \\
        \sqrt{\theta \varepsilon} \vct\big( \sqrt{\mat{c}(\vect{x})} \otimes \sqrt{\mat{c}(\vect{x})} \big)
    \end{bmatrix}
    \in \R^{1 + d + d^2 + d^4},
\end{equation}
defines an explicit feature map satisfying Eq.~\eqref{eqn:explicit_feature_map} for the linear PTT model.
\end{example}

\begin{example}
For the \textbf{nonlinear PTT} model, we have
\begin{equation}
\begin{aligned}
    h(\mat{c}) &= \exp\left[ \varepsilon\Tr(\mat{c} - \mat{I}) \right]\Tr(\mat{c} - \mat{I}) + d \\
        &= d + e^{-\varepsilon d}\left[ \Tr(\mat{c}) e^{\varepsilon\Tr(\mat{c})} - d e^{\varepsilon \Tr(\mat{c})} \right] \\
        &= d + e^{-\varepsilon d}\left[ \sum_{p=1}^{\infty} \frac{1}{(p-1)!}\varepsilon^{p-1}\Tr(\mat{c})^p - d - d\sum_{p=1}^{\infty} \frac{1}{p!} \varepsilon^p \Tr(\mat{c})^p \right] \\
        &= \underbrace{d\left(1-e^{-\varepsilon d}\right)}_{c_{0,0}} + \sum_{p=1}^{\infty}\underbrace{\frac{\varepsilon^{p-1} e^{-\varepsilon d}}{(p-1)!}\left(1 - \frac{\varepsilon d}{p} \right)}_{c_{0,p}}\Tr(\mat{c})^p,
\end{aligned}
\end{equation}
which takes the form of Eq.~\eqref{eqn:stress_model_series} with $f_0: x\mapsto \sqrt{x}$.
We observe that if $0 \leq \varepsilon \leq d^{-1}$ then all of the coefficients $c_{0,p}$ are non-negative and $c_{0,p}$ is strictly positive for every $p\geq 2$.
Therefore, if the parameter $\varepsilon$ of the nonlinear PTT model satisfies $0\leq \varepsilon \leq d^{-1}$, then
Eq.~\eqref{eqn:main_kernel_fun} with
\begin{equation}
    \tilde{k}(\mat{c}_1, \mat{c}_2) = 
    \exp\left[\varepsilon \Tr\big(\sqrt{\mat{c}_1}\sqrt{\mat{c}_2}-\mat{I}\big) \right]\Tr\big(\sqrt{\mat{c}_1}\sqrt{\mat{c}_2} - \mat{I}\big) + d,
    \label{eqn:nonlinear_PTT_integrand_kernel}
\end{equation}
defines an injective positive semi-definite kernel function satisfying Eq.~\eqref{eqn:energy_compatibility}.
\end{example}

\begin{revtwobox}   
\subsection{Forming new energy-compatible kernel functions}
\label{subsec:creating_new_kernels}
A large family of energy-compatible, injective, positive-definite kernel functions can be generated from a single kernel function $k:\mcal{F} \times \mcal{F} \to \R$ satisfying these properties.
New kernel functions are obtained by multiplying $k$ pointwise by any positive-definite kernel functions $\tilde{k}:\mcal{F} \times \mcal{F} \to \R$ satisfying $\tilde{k}(\vect{q}, \vect{q}) = 1$ for every $\vect{q} \in \mcal{F}$.
Examples include radial kernel functions such as Gaussian and Mat\'{e}rn kernels, as well as the cosine kernel.
It is well-known \cite{Aronszajn1950theory} that the pointwise product $k \tilde{k}:\mcal{F} \times \mcal{F} \to \R$ defined by
\begin{equation}
    \label{eqn:product_kernel}
    k \tilde{k} (\vect{q}_1, \vect{q}_2) := k(\vect{q}_1, \vect{q}_2) \tilde{k}(\vect{q}_1, \vect{q}_2)
\end{equation}
is a new positive-definite kernel function.
Energy-compatibility of this product kernel is obvious and
the following proposition shows that it gives rise to an injective feature map.

\begin{proposition}\label{prop:product_kernel_injectivity}
        Let $k, \tilde{k}: \mcal{F} \times \mcal{F} \to \R$ be positive-definite kernel functions where $k$ gives rise to an injective feature map and $\tilde{k}(\vect{q}, \vect{q}) = 1$ for every $\vect{q} \in \mcal{F}$. Then the positive-definite kernel function defined by \eqref{eqn:product_kernel} gives rise to an injective feature map.
        A proof is given in Appendix~\ref{app:viscoelastic_kernels_proof}
    \end{proposition}
\end{revtwobox}

\subsection{Comparing states with different parameter values}

So far, our kernel functions and the associated distance function on $\mcal{F}$ only allows us to compare states at the same values of the parameters $\beta$, $\Rey$, and $\Wei$.
However, this situation is easily rectified by defining a new kernel function
\begin{equation}
        \check{k}\left((\vect{q}_1,\theta_1),\ (\vect{q}_2, \theta_2) \right) 
        := \frac{1}{2} \int_{\Omega} \left[ \vect{u}_1(\vect{x}) \cdot \vect{u}_2(\vect{x}) + \sqrt{\theta_1 \theta_2} \tilde{k}\big(\mat{c}_1(\vect{x}), \mat{c}_2(\vect{x})\big) \right] \td \vect{x},
        \label{eqn:parametric_kernel_fun}
\end{equation}
which depends on the values of the parameter $\theta = (1-\beta)/(\Rey \Wei)$.
This is a positive-definite kernel function of a unique RKHS $\check{\mcal{H}}$ over $\mcal{F}\times (0,\infty)$ thanks to the kernel product rule stated in Lemma~\ref{lem:combining_kernels} in \ref{app:viscoelastic_kernels_proof}.
Under the same conditions stated in Theorem~\ref{thm:viscoelastic_kernels}, the associated feature map $\check{\Phi}: \mcal{F}\times (0,\infty) \to \check{\mcal{H}}$ is ``conditionally injective'', meaning that
\begin{equation}
    \Phi_\theta : \vect{q} \mapsto \check{\Phi}(\vect{q}, \theta)
\end{equation}
is injective for every $\theta > 0$.
To see this, we observe that $k(\vect{q}_1, \vect{q}_2) = \check{k}\left((\vect{q}_1, \theta),\ (\vect{q}_2, \theta) \right)$.
Since $k$ is an injective kernel function, it follows that 
\[
    \| \Phi_\theta(\vect{q}_1) - \Phi_\theta(\vect{q}_2) \|_{\check{\mcal{H}}}
    = d_{\mcal{E}}(\vect{q}_1, \vect{q}_2) > 0
\]
when $\vect{q}_1 \neq \vect{q}_2$ in $\mcal{F}$.

\section{Reconstructing fields and preimages for kernel functions}
\label{subsec:reconstrucion}
Once states are lifted into the RKHS, it is often of practical importance to be able to reconstruct them.
This is referred to as the preimage problem in kernel-based machine learning \cite{Kwok2004preimage, Honeine2011preimage}.
The following theorem provides a partial solution for the class of viscoelastic kernels presented above.
It describes spatial fields that can be linearly reconstructed from the lifted representations of states in the RKHS.
In particular, Corollary~\ref{cor:reconstruction} means that for each of the stress models listed in Table~\ref{tab:stress_models}, 
the field $\vect{y} = (\vect{u}, \sqrt{\mat{c}})$ lies in $L^2(\Omega;\ \R^d \times \R^{d\times d})$ and can be reconstructed via a bounded linear operator acting on $\Phi(\vect{q})$ for every $\vect{q} \in \mcal{F}$.
For the Oldroyd-B, Giesekus, and linear PTT models, this property is evident from their explicit feature maps given above in Eqs.~\eqref{eqn:OldB_explicit_feature_map},~\eqref{eqn:Giesekus_explicit_feature_map},~and~\eqref{eqn:LinearPTT_explicit_feature_map}.

\begin{theorem}
    \label{thm:reconstruction}
    With the same assumptions as Theorem~\ref{thm:viscoelastic_kernels}, let $\vect{a}_0\in \R^d$ and $\mat{A}_{i,p} \in \R^{d^p \times d^p}$ satisfy
    \begin{equation}
        A := \sqrt{2 \| \vect{a}_0 \|_2^2 + \frac{2}{\theta} \sum_{i=0}^{\infty} \sum_{p=0}^{\infty} \| \mat{A}_{i,p} \|_F^2} < \infty
    \end{equation}
    and consider the function $\psi: \R^d \times D(\mat{s}) \to \R$ defined by
    \begin{equation}
        \psi(\vect{u}, \mat{c}) = \vect{a}_0^T \vect{u} + \sum_{i=0}^{\infty} \sum_{p=0}^{\infty} \sqrt{c_{i,p}} \Tr\left[ \mat{A}_{i,p}^T f_{i}(\mat{c})^{\otimes p} \right].
    \end{equation}
    Then there is a unique bounded linear operator $R_{\psi}: \mcal{H} \to L^2(\Omega)$ satisfying
    \begin{equation}
        R_{\psi} \Phi(\vect{q}) = \psi \circ \vect{q}
    \end{equation}
    for every $\vect{q} = (\vect{u}, \mat{c}) \in \mcal{F}$ and $\| R_{\psi} \| \leq A$.
    We provide a proof in \ref{app:reconstruction_proofs}.
\end{theorem}

\begin{corollary}
    \label{cor:reconstruction}
    With the same assumptions as Theorem~\ref{thm:viscoelastic_kernels}, suppose that there is a nonzero coefficient $c_{i,p} > 0$.
    Then there is a unique bounded linear operator $R_{i,p}: \mcal{H} \to L^2(\Omega;\ \R^{d} \times \R^{d^p \times d^p} )$ satisfying
    \begin{equation}
        R_{i,p} \Phi(\vect{q}) = \left( \vect{u},\ f_i(\mat{c})^{\otimes p} \right)
    \end{equation}
    for every $\vect{q} = (\vect{u}, \mat{c}) \in \mcal{F}$. \revtwo{Its} norm is bounded by
    $
        \| R_{i,p} \| \leq \sqrt{2d + \frac{2 d^{2p}}{\theta c_{i,p}}}
    $.
\end{corollary}

\begin{remark}
    Incidentally, the proof of Theorem~\ref{thm:reconstruction} shows that the metric space $(\mcal{F}, d_{\mcal{E}})$ and the RKHS $\mcal{H}$ are separable.
    This is a useful property for studying operators on $\mcal{H}$, such as the covariance operator in kernel principal component analysis described below.
\end{remark}

These results justify linear reconstruction of the field $\vect{y} = (\vect{u}, \sqrt{\mat{c}})$ using low-dimensional coordinates obtained via kernel Principal Component Analysis (KPCA) \cite{Scholkopf1998nonlinear}.
Here, there is a probability measure $\mu$ over states in $\mcal{F}$, and we seek a collection of low-dimensional features describing states from this distribution.
For the sake of simplicity, we consider the uncentered version of KPCA where the covariance operator $C_{\mu}:\mcal{H} \to \mcal{H}$ is defined by
\begin{equation}
    C_{\mu} f = \int_{\mcal{F}} \Phi(\vect{q}) \langle \Phi(\vect{q}), f \rangle_{\mcal{H}} \td \mu(\vect{q}).
\end{equation}
If the probability measure has finite average energy, i.e., $\int_{\mcal{F}} \mcal{E}(\vect{q}) \td\mu(\vect{q}) < \infty$, then the covariance operator is self-adjoint, positive semidefinite, and trace-class by Theorem~4.1 in \cite{Minh2018covariances}.
Therefore, $\mcal{H}$ admits an orthonormal basis of eigenvectors $\{ u_j \}_{j=1}^{\infty}$ of $C_{\mu}$ with eigenvalues $\lambda_j = \sigma_j^2$, $\sigma_j \geq 0$, arranged in descending order and satisfying \revtwo{$\sum_{j=1}^{\infty} \sigma_j^2 < \infty$. Using KPCA, the principal components}
\begin{equation}
    z_j(\vect{q}) := \langle u_j,\ \Phi(\vect{q}) \rangle_{\mcal{H}}
\end{equation}
can be computed without performing explicit operations in the RKHS.
To do this, KPCA uses a kernel integral operator $K_{\mu} : L^2(\mcal{F},\mu) \to L^2(\mcal{F},\mu)$ defined by
\begin{equation}
    (K_{\mu} f)(\vect{q}) = \int_{\mcal{F}} k(\vect{q}, \vect{q}') f(\vect{q}') \td \mu(\vect{q}).
\end{equation}
This operator is self-adjoint and has the same eigenvalues $\lambda_j$ as $C_{\mu}$, with eigenfunctions $\{ v_j \}_{j=1}^{\infty}$ forming an orthonormal basis for $L^2(\mcal{F},\mu)$.
Crucially, these eigenfunctions can be computed using the kernel function, enabling one to compute the KPCA coordinates
\begin{equation}
    z_j(\vect{q}) 
    = \sigma_{j}^{-1} \int_{\mcal{F}} v_j(\vect{q}') k(\vect{q}', \vect{q}) \td \mu(\vect{q}').
\end{equation}

Since $C_{\mu}$ is a compact operator on a Hilbert space, it can be approximated in norm by finite-rank operators.
One often uses empirical covariance operators obtained by drawing finitely many independent samples $\{ \vect{q}_i \}_{i=1}^m$ from $\mu$ and forming the covariance $C_{\mu_m}$ for the empirical measure $\mu_m = m^{-1} \sum_{i=1}^m \delta_{\vect{q}_i}$ with $\delta_{\vect{q}_i}$ denoting the Dirac measure centered at $\vect{q}_i$.
With the identification $L^2(\mcal{F}, \mu_m) \cong \R^m$, which is valid when the sampled states are distinct, the corresponding kernel integral operator $K_{\mu_m}$ becomes the kernel matrix $\mat{K} \in \R^{m\times m}$ with entries $[\mat{K}]_{i,j} = m^{-1} k(\vect{q}_i, \vect{q}_j)$ and eigenvectors $\vect{v}_j \in \R^m$.
The KPCA features for the empirical distribution are then given by
\begin{equation}
    z_j(\vect{q}) = m^{-1/2} \sigma_j^{-1} 
    \begin{bmatrix}
        k(\vect{q}_1, \vect{q}) &
        \cdots &
        k(\vect{q}_m, \vect{q})
    \end{bmatrix}
    \vect{v}_j.
\end{equation}

The following result says that certain spatial fields such as $(\vect{u}, \sqrt{\mat{c}})$ can be linearly reconstructed with guaranteed accuracy using the leading KPCA features.
Here, we arrange the first $r$ eigenvectors into an operator $U_r:\R^r \to \mcal{H}$ defined by $(w_1, \ldots, w_r) \mapsto w_1 u_1 + \cdots w_r u_r$ and we denote the vector of leading principal components $\vect{z}_r(\vect{q}) = ( z_1(\vect{q}), \ldots, z_r(\vect{q}) ) = U_r^* \Phi(\vect{q})$.
\begin{proposition}
    \label{prop:KPCA_field_reconstruction}
    With the same assumptions as Theorem~\ref{thm:viscoelastic_kernels}, suppose that there is a nonzero coefficient $c_{i,p} > 0$ and the average energy $\int_{\mcal{F}} \mcal{E}(\vect{q}) \td \mu(\vect{q}) < \infty$ is finite.
    For a state vector $\vect{q} = (\vect{u}, \mat{c}) \in \mcal{F}$, the reconstruction error
    is bounded by
    \begin{equation}
        \label{eqn:single_state_KPCA_rec_bound}
        \left\| \big( \vect{u},\ f_i(\mat{c})^{\otimes p} \big) - R_{i,p} U_r \vect{z}_r(\vect{q}) \right\|_{L^2(\Omega)}^2 
        \leq 2 \left(d + \frac{d^{2p}}{\theta c_{i,p}}\right) \Big( \mcal{E}(\vect{q}) - \| \vect{z}_r(\vect{q}) \|_2^2 \Big).
    \end{equation}
    The average reconstruction error for states drawn according to $\mu$ is bounded by
    \begin{equation}
        \label{eqn:average_KPCA_rec_bound}
        \int_{\mcal{F}} \left\| \big( \vect{u},\ f_i(\mat{c})^{\otimes p} \big) - R_{i,p} U_r \vect{z}_r(\vect{q}) \right\|_{L^2(\Omega)}^2 \td \mu(\vect{q})
        \leq 2 \left( d + \frac{d^{2p}}{\theta c_{i,p}}\right) \sum_{j=r+1}^{\infty} \sigma_j^2 .
    \end{equation}
    We provide a proof in \ref{app:reconstruction_proofs}.
\end{proposition}

\begin{revtwobox} 
\begin{remark}
    When the centered version of KPCA is used the reconstruction will involve an affine, rather than purely linear map of the leading principal components.
\end{remark}
\end{revtwobox} 

\section{A case study: transient lid-driven cavity flow}

Let us consider the lid-driven cavity problem defined over a square spatial domain $[0,1]\times [0,1]$, and driven by the top lid moving to the right. 
While the cavity flow problem may not have immediate practical applications, it holds significant importance as a viscoelastic benchmark (see \cite{SOUSA2016129,Castillo2015,chai2021efficient}). 
This simple geometry can be employed to address challenges associated with the singularity of the stress field on the corner of the cavity as well as to study the emergence of elastic turbulence for high Weissenberg number (\cite{comminal2015robust}). 
In particular, we have illustrated the streamlines for two situations in Fig. \ref{stream} for the lid-driven cavity flow considering the Oldroyd-B model.

To generate the numerical data for our investigation, we employed the same code designed for solving viscoelastic fluid flows as in \cite{Martins2015,evans2017stresses,evans2017transient,Franca}. 
In this methodology, we employ finite-difference approximations for the discretization of the momentum and mass equations (\ref{mom}) and (\ref{mass}), as well as the constitutive equation (\ref{cons}). 
Managing the interdependence between the velocity and pressure fields, a key challenge in solving the governing equations (\ref{mom}) and (\ref{mass}), is addressed through the application of a projection scheme. 
Widely recognized in computational fluid dynamics, this scheme effectively decouples the velocity and pressure fields, simplifying the computational process and enhancing the stability and accuracy of our numerical simulations. 
Specifically, the momentum equation (\ref{mom}) is time-discretized using a semi-implicit strategy, while an explicit time discretization is employed for solving the constitutive equation (\ref{cons}). 
Additionally, for the discretization of the convective terms presented in equations (\ref{mom}) and (\ref{cons}), we apply a high-order upwind methodology. 
In order to obtain efficient solutions for high Weissenberg number simulations, we incorporate a stabilization scheme known as log-conformation (\cite{Alves2021numerical}). 
The log-conformation method proves particularly effective in handling flows with high Weissenberg numbers, ensuring stability and accuracy in capturing the intricate rheological behavior of viscoelastic fluids.

In our numerical tests, we use a uniform mesh with a spatial discretization of $\Delta x= 1/80 = 0.0125$ within the square domain. 
Time integration was performed with a fixed time-step of $\Delta t = 10^{-5}$, ensuring temporal stability and capturing the transient dynamics of the viscoelastic fluid flow. 
The fixed mesh size and time-step, were selected based on a careful balance between computational efficiency and the accuracy required for our investigation. 

Since we aim to analyze the dynamical behavior of the flow, we adopt the same smooth, transient velocity profile for the top lid used in \cite{Castillo2015}:
\begin{equation}
	u(x,1,t) = 16x^2(1 - x)^2\sin(\pi t).\\
\end{equation}
The remaining cavity walls are stationary and we impose no-slip boundary conditions.

\begin{figure}
    \centering
    a){\includegraphics[width=.35\textwidth,trim= 0in 0.4in 0in .1in]{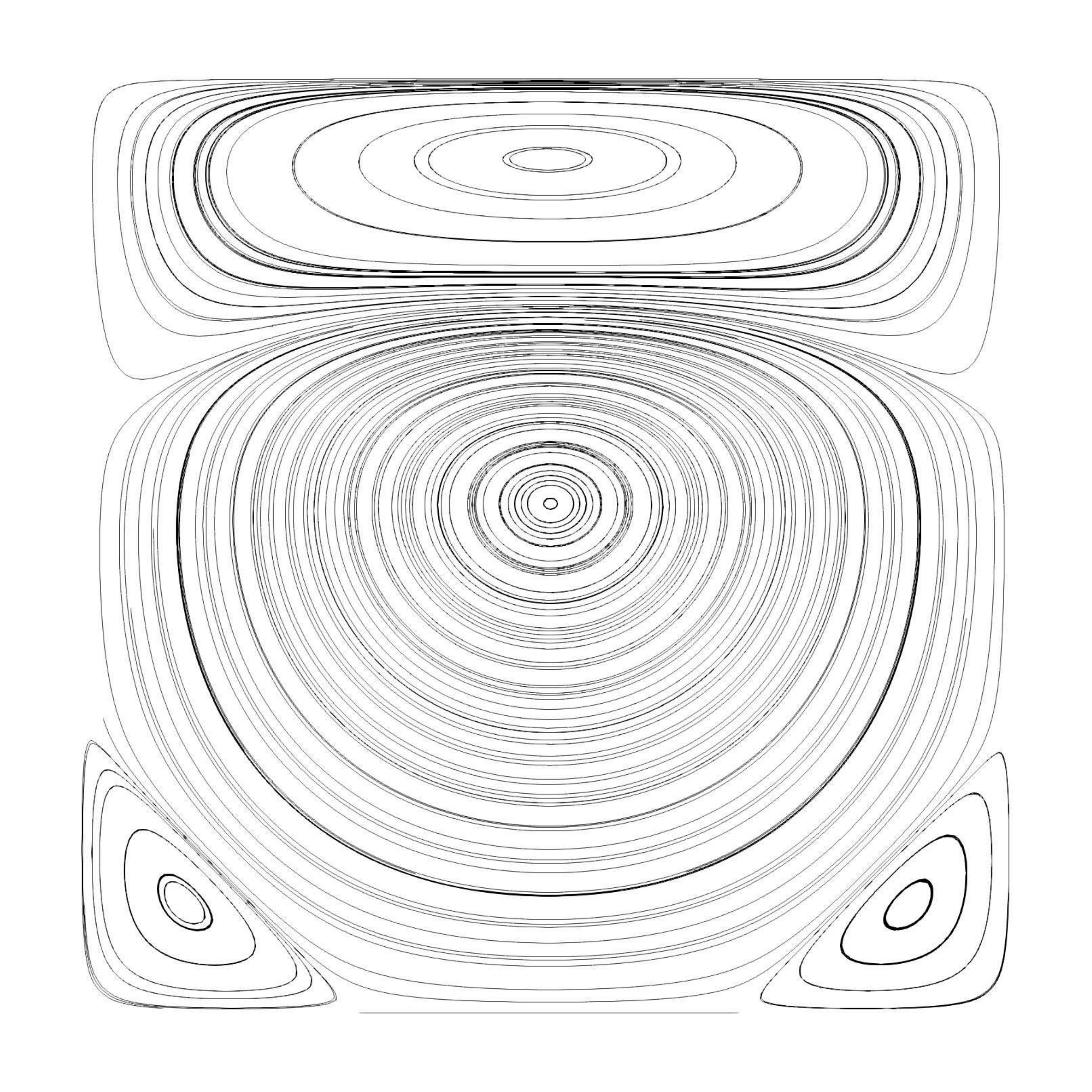}
    \includegraphics[width=.4\textwidth]{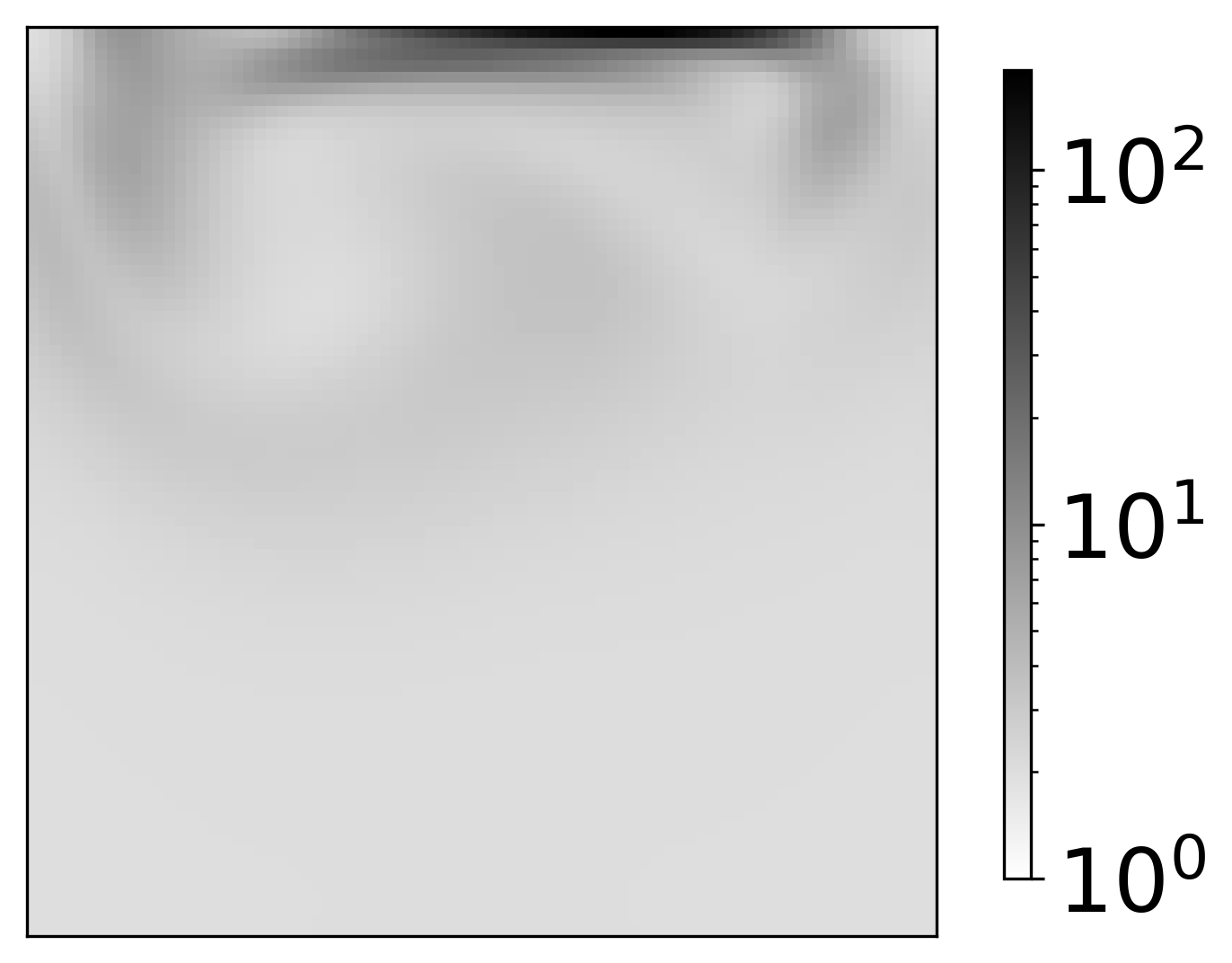}}\\
    b){\includegraphics[width=.35\textwidth,trim= 0in 0.4in 0in .1in]{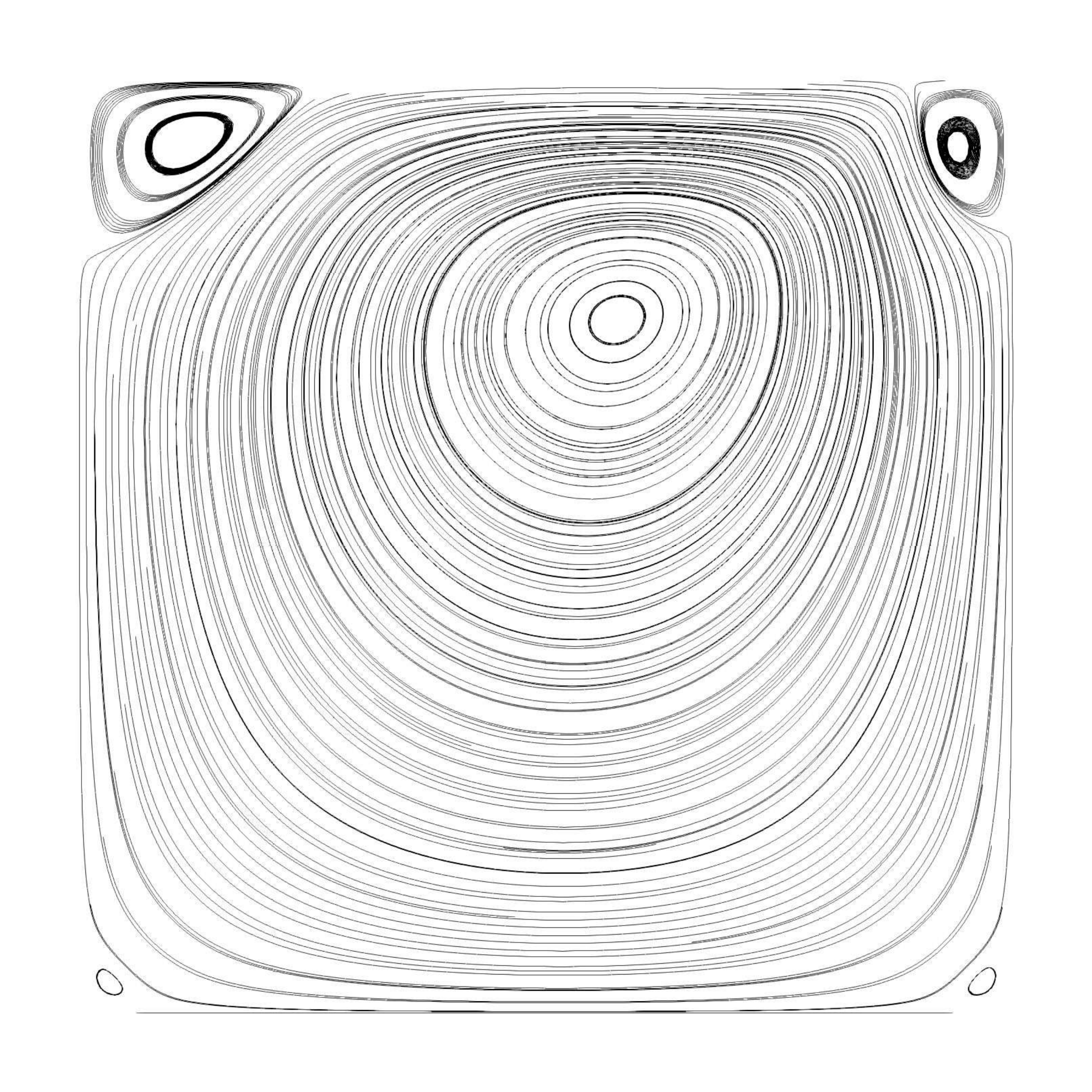} 
    \includegraphics[width=.4\textwidth]{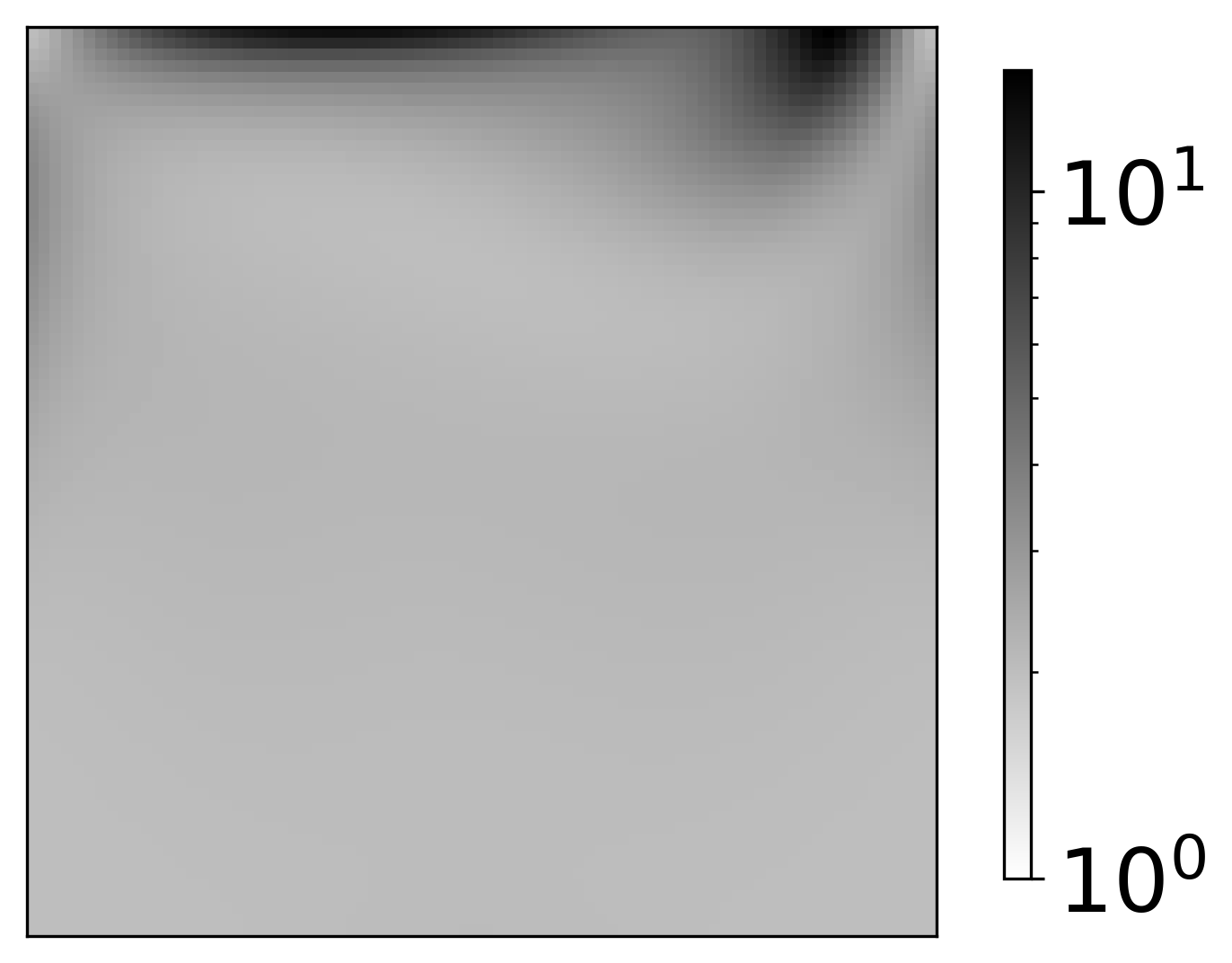}}\\
    \caption{Streamlines (left) and $\Tr(\mat{c})$ (right) of the cavity flow using the viscoelastic Oldroyd-B model: a) $Re=10 , Wi=10, \beta=0.9$ and b) $Re=0.001, Wi=5 ,\beta=0.5$.}
    \label{stream}
\end{figure}

\subsection{Dataset, algorithm and measurement of errors}

The data from simulations are organized into a matrix $\textbf{Y} \in \mathbb{R}^{D\times m}$, where $D$ represents the product of the number of mesh points and the number of variables in each spatial location.
Each column is a modified state vector
\begin{equation}
    \vect{y} =\begin{pmatrix}
        u   \\
        v   \\
        b_{xx}   \\
        b_{xy}  \\
        b_{yy}
    \end{pmatrix}
    \qquad \mbox{where} \qquad \mat{b} = \sqrt{\mat{c}}
\label{q_vector}\end{equation}
corresponding to one of the $m$ snapshots collected from a simulation.
This definition of a modified state vector is due the fact it can be linearly reconstructed from the lifted states, as discussed in section \ref{subsec:reconstrucion}. 
We also drop the $b_{yx}$ entry of the square root conformation tensor since it is equal to $b_{xy}$ by symmetry.
The original state vector $\vect{q}$ can be obtained from $\vect{y}$ by the relation $\mat{c} = \mat{b}^2$ \cite{Balci2011symmetric}. 
Hence, the matrix $\textbf{Y}$ is structured as
\begin{equation}
\textbf{Y} = 
\left[
  \begin{array}{cccc}
    \vertbar & \vertbar & &   \vertbar \\
    \vect{y}_1   & \vect{y}_2  & \dots & \vect{y}_m \\
    \vertbar & \vertbar  & & \vertbar 
  \end{array}
\right],
\label{eqn:snapshot_matrix}
\end{equation}  
where $\vect{y}_i$ is the modified state vector defined by Eq.~\ref{q_vector} for the $i$th state vector snapshot $\vect{q}_i$ gathered from a simulation.


We use the Kernel Principal Component Analysis (KPCA) algorithm developed in \cite{Scholkopf1998nonlinear} to extract low-dimensional variables describing the simulation data.
In summary, the kernel matrix $\mat{K} \in \mathbb{R}^{m\times m}$ with entries $[\mat{K}]_{i,j} = k(\vect{q}_i, \vect{q}_j)$ is centered to form
\begin{equation}
    \mat{K}_c = \left(\mat{I} - \frac{1}{m} \vect{1}_m \vect{1}_m^T \right) \mat{K} \left(\mat{I} - \frac{1}{m} \vect{1}_m \vect{1}_m^T \right),
\end{equation}
where $\vect{1}_m$ denotes the vector with unit entries in $\R^m$.
Computing the symmetric eigenvalue decomposition of the centered kernel matrix 
\begin{equation}
    \mat{K}_c = \mat{V} \mat{\Lambda} \mat{V}^T,
    \label{K_alg}
\end{equation}
and forming its rank-$r$ truncation $\mat{V}_r \mat{\Lambda}_r \mat{V}_r^T$ retaining the largest $r$ eigenvalues
yields
\begin{equation}
\mat{Z}_r = 
\left[
  \begin{array}{cccc}
    \vertbar & \vertbar &  &   \vertbar \\
    \vect{z}_r(\vect{q}_1)   & \vect{z}_r(\vect{q}_2)  & \dots & \vect{z}_r(\vect{q}_m) \\
    \vertbar & \vertbar  &  & \vertbar 
  \end{array}
\right] = \sqrt{\mat{\Lambda}_r} \mat{V}_r^T.
\label{eqn:KPCA_matrix}
\end{equation}
The columns of this matrix are vectors comprising the leading $r$ kernel principal components of each snapshot.
The kernel matrix is constructed using Eq.~\eqref{eqn:integrand_kernel} for the selected stress model. 
For example, if the Oldroyd-B model is chosen, the kernel function Eq.~\eqref{eqn:integrand_kernel} is computed using Eq.~\eqref{eqn:Oldroyd-B_integrand_kernel}. 
It is worth noting that the linear PCA reduction is also obtained as a special case of KPCA using a kernel given by the classical $L^2(\Omega; \R^d\times \R^{d\times d})$ inner product of state vectors.

We reconstruct the modified snapshots via an affine function of the kernel principal components
\begin{equation}
    \vect{\hat{y}}_i = \vect{r}_0 + \mat{R} \vect{z}_r(\vect{q}_i),
\end{equation}
by solving the least-squares problem
  \begin{equation}
    \argmin_{\mat{R}, \vect{r}_0} \Vert \textbf{Y} - \vect{r}_0 \vect{1}_m^T - \mat{R}\mat{Z}_r \Vert_{F}^{2},
        \label{opt2}
    \end{equation}
where $\mat{R}  \in \mathbb{R}^{D\times r}$ and $\vect{r}_0 \in \R^D$. 
Affine reconstruction of modified state vectors is justified per the discussion in Section~\ref{subsec:reconstrucion}.
Approximations of the original state vectors $\vect{\hat{q}}$ are then obtained from $\vect{\hat{y}}$ by the relation $\mat{\hat{c}} = \mat{\hat{b}}^2$.

In order to study the effect of the kernel function in dimensionality reduction, we consider several ways of measuring the reconstruction error.
The first way to measure the error is using the ordinary $L^2$ norm of the reconstructed fields
\begin{equation}
E_F = \frac{\sum_{i=1}^{m} \Vert \vect{y}_i - \vect{\hat{y}}_i \Vert_{L^2(\Omega)}^{2}}{\sum_{i=1}^{m}\Vert \vect{y}_i \Vert_{L^2(\Omega)}^{2}}.
\label{E1}
\end{equation}
Note that the numerator and denominator can be computed as Frobenius norms of respective snapshot matrices.
The problem with using this notion of error is that it is not grounded in the mechanical energy of the system, and could assign improper weight to the kinetic and elastic components of the error.
One approach to rectify this is to measure the error in the total mechanical energy of the reconstruction by
\begin{equation}
E_E =  \frac{\sum_{i=1}^{m}| \mcal{E}(\vect{q}_{i}) - \mcal{E}(\vect{\hat{q}_{i}})|}{\sum_{i=1}^{m} \mcal{E}(\vect{q}_{i})}.
\label{E3}
\end{equation}
However, this notion of error can be zero even when the reconstructed and ground truth snapshots are different.
Both issues are resolved by measuring the error in the kernel Hilbert space associated with the ground-truth stress model by
\begin{equation}
E_H =  \frac{\sum_{i=1}^{m} \Vert \Phi(\vect{q}_i) - \Phi(\vect{\hat{q}}_i)\Vert_{\mcal{H}}^{2}}{\sum_{i=1}^{m} \Vert \Phi(\vect{q}_i) \Vert_{\mcal{H}}^2}
= \frac{\sum_{i=1}^{m} d_{\mcal{E}}(\vect{q}_i, \vect{\hat{q}}_i)^{2}}{\sum_{i=1}^{m} \mcal{E}(\vect{q}_i) },
\label{E2}
\end{equation}
where quantities in the numerator and denominator are computed using Eq.~\eqref{eqn:kernel_metric}, Eq.~\eqref{eqn:energy_compatibility}, and the kernel function.

In the following two subsections, we study how the choice of kernel affects the ability of kernel-based dimensionality reduction and reconstruction methods to capture essential features of the lid-driven cavity flow.
To do this, we simulate the cavity flow using different choices of ground-truth stress model, Reynolds number, and Weissenberg number.
These ground-truth quantities are used to define the error metrics $E_E$ and $E_H$ in Eqs.~\ref{E3}~and~\ref{E2}.
Kernel PCA-based dimensionality reduction and reconstruction are then performed using the ground-truth kernel function and compared to other ad-hoc choices such as ordinary PCA and kernel functions associated with incorrect stress models.

\subsection{Results using the Oldroyd-B model}

We begin by studying the case where the linear Oldroyd-B model is used to define the ground-truth, and we compare against ordinary PCA.
This model holds significant popularity in non-Newtonian fluid mechanics, making it an ideal starting point for making comparisons.
We note that KPCA using the Oldroyd-B stress model is equivalent to performing PCA using properly weighted modified state vectors $\vect{\tilde{y}} = \Psi(\vect{q})$ defined by the explicit feature map in Eq.~\eqref{eqn:OldB_explicit_feature_map}.
To compare, we na\"{i}vely perform PCA using the state vectors defined by Eq.~\ref{q_vector} without adjusting how the kinetic and elastic components are weighted.
These two methods become nearly identical when $\theta = 1$, with the only difference being that $b_{xy}$ and $b_{yx}$ both appear in the properly weighted modified state vector, while only $b_{xy}$ appears in the modified state vector used to perform PCA.
The results discussed below show that using the correct weighting, i.e., the kernel function associated with the Oldroyd-B model is crucial for PCA to properly capture the most energetic flow structures.

In Fig.~\ref{comparison_OldB_theta_modes}, we plot our three error metrics against the number of kernel principal components (modes) used for reconstruction at different values of $\theta$.
These values were calculated using the triples $(\Rey = 10, \Wei = 10, \beta = 0.9)$, $(\Rey = 1, \Wei = 0.5, \beta = 0.5)$, and $(\Rey = 0.001, \Wei = 0.5, \beta = 0.5)$.
Ordinary PCA minimizes the reconstruction error with respect to the Frobenius norm, so it is no surprise that its error with respect to $E_F$ is lower.
However, this norm does not reflect the energetic content of the flow.
When $\theta \neq 1$, reduction and reconstruction based on the correct Oldroyd-B kernel yields smaller errors both in the kernel Hilbert space, as measured by $E_H$, and in the mechanical energy of the reconstructed snapshots, as measured by $E_E$.
As expected, the two methods have nearly identical performance in all error metrics when $\theta = 1$, as discussed above.

\begin{figure}[!ht]
    \centering
 a)
    \includegraphics[width=.3\textwidth]{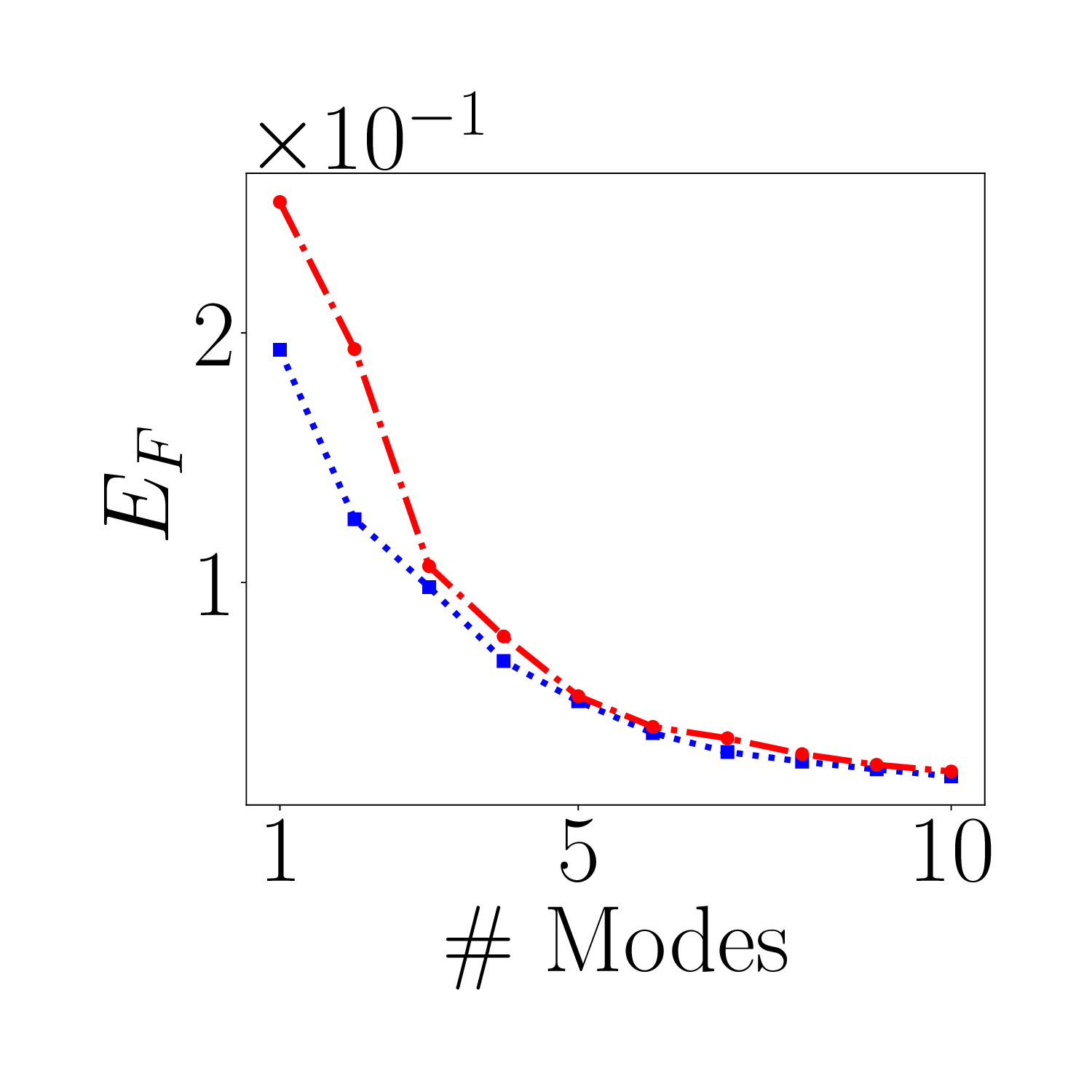}
    \includegraphics[width=.3\textwidth]{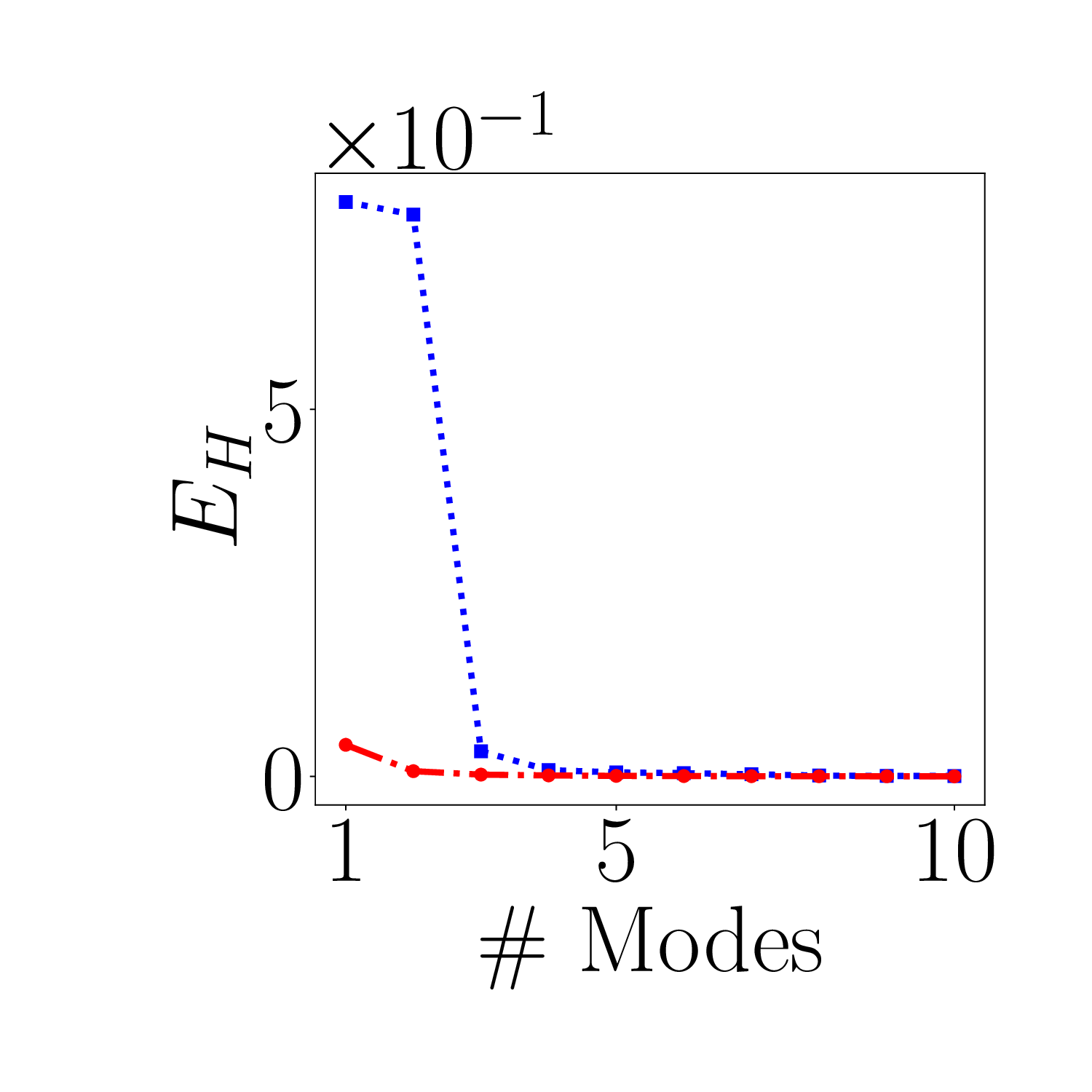} 
    \includegraphics[width=.3\textwidth]{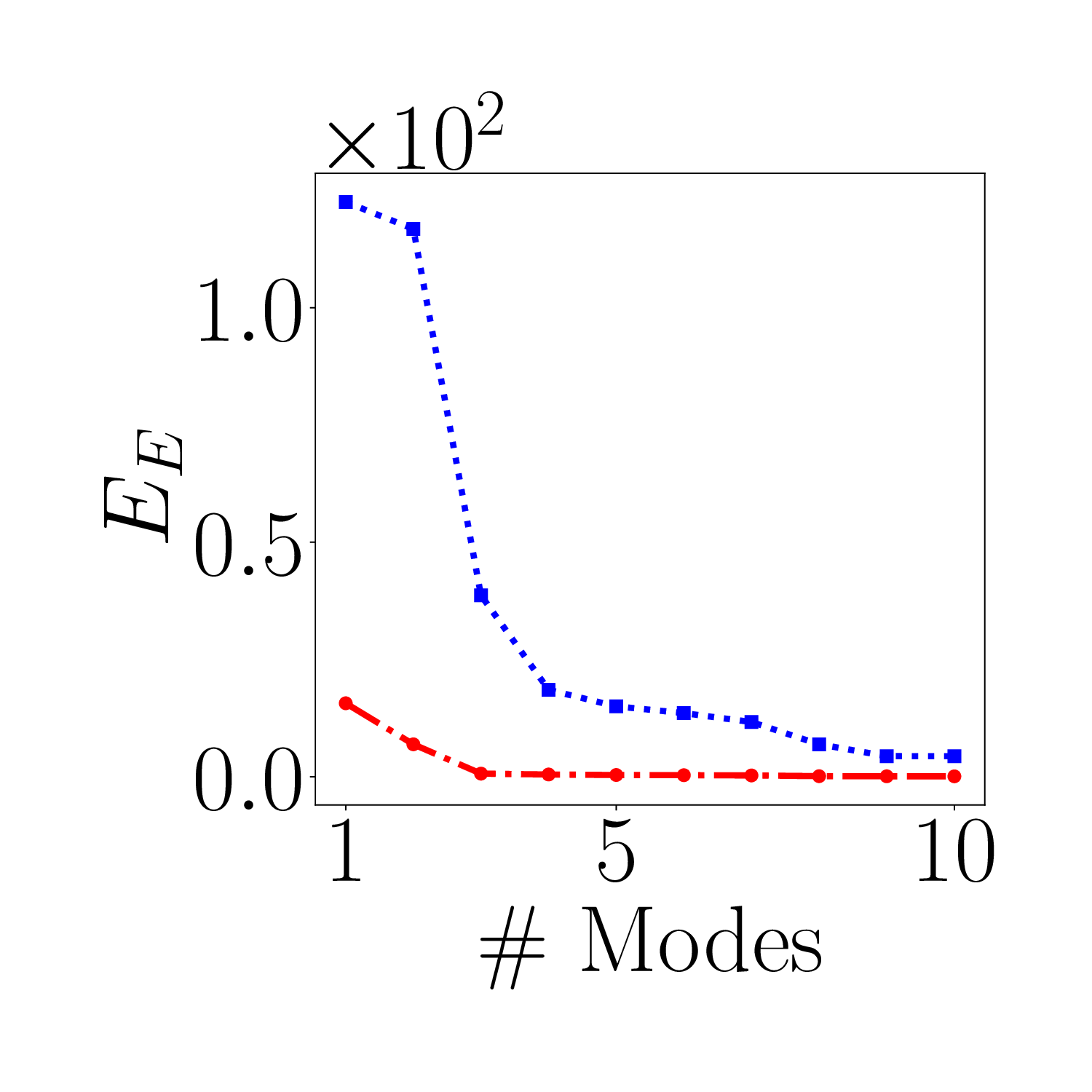} \\
b)
    \includegraphics[width=.3\textwidth]{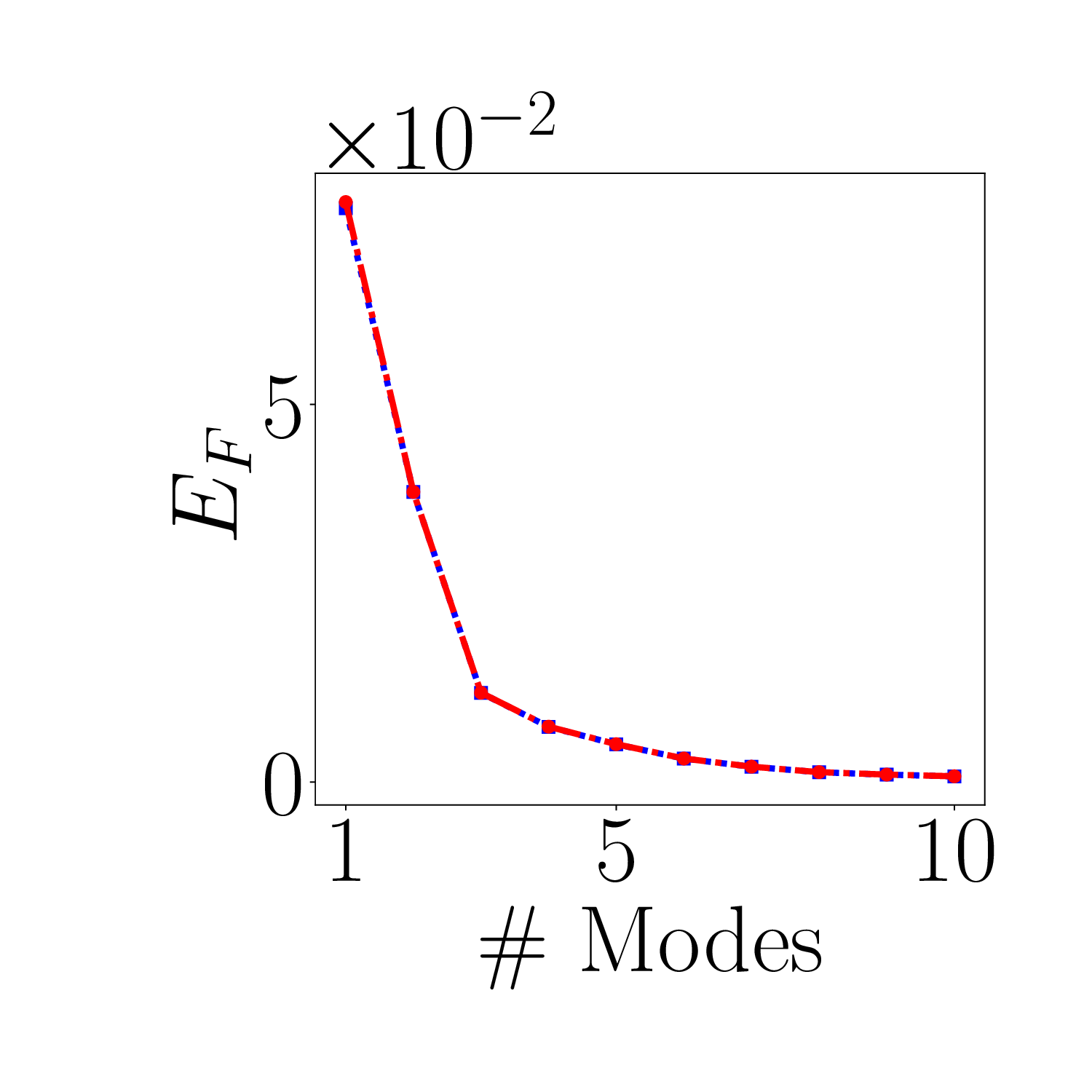}
    \includegraphics[width=.3\textwidth]{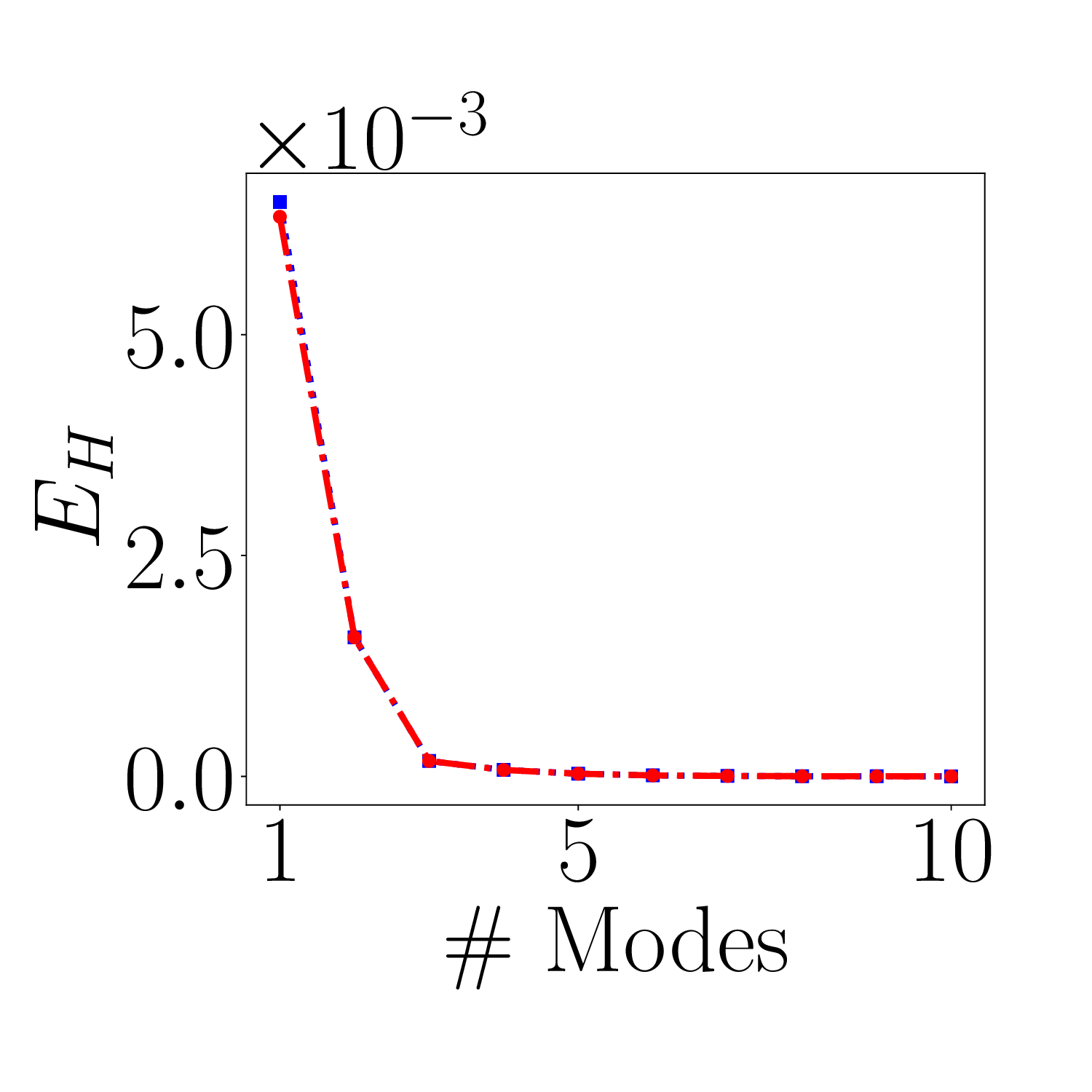} 
    \includegraphics[width=.3\textwidth]{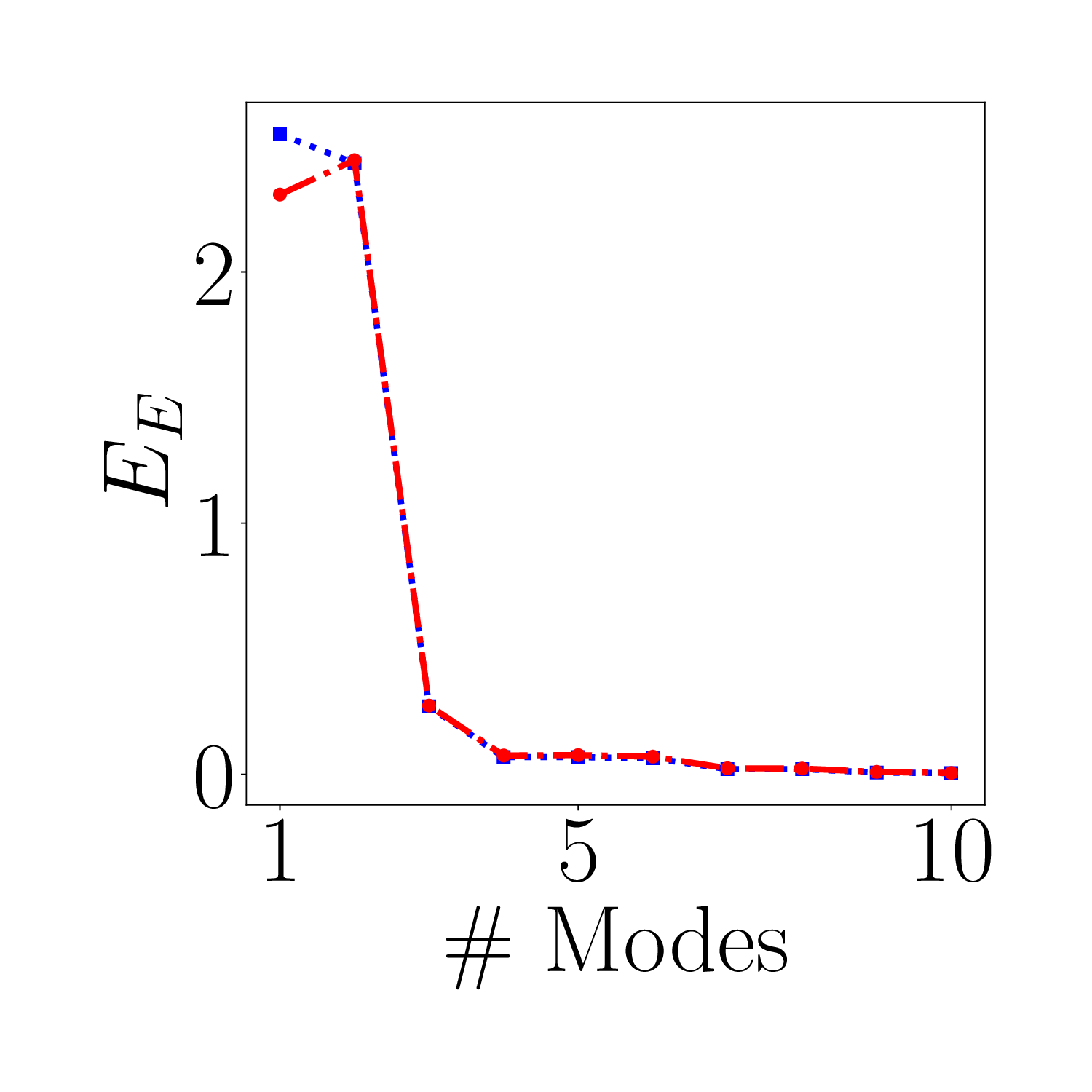} \\
c)
    \includegraphics[width=.3\textwidth]{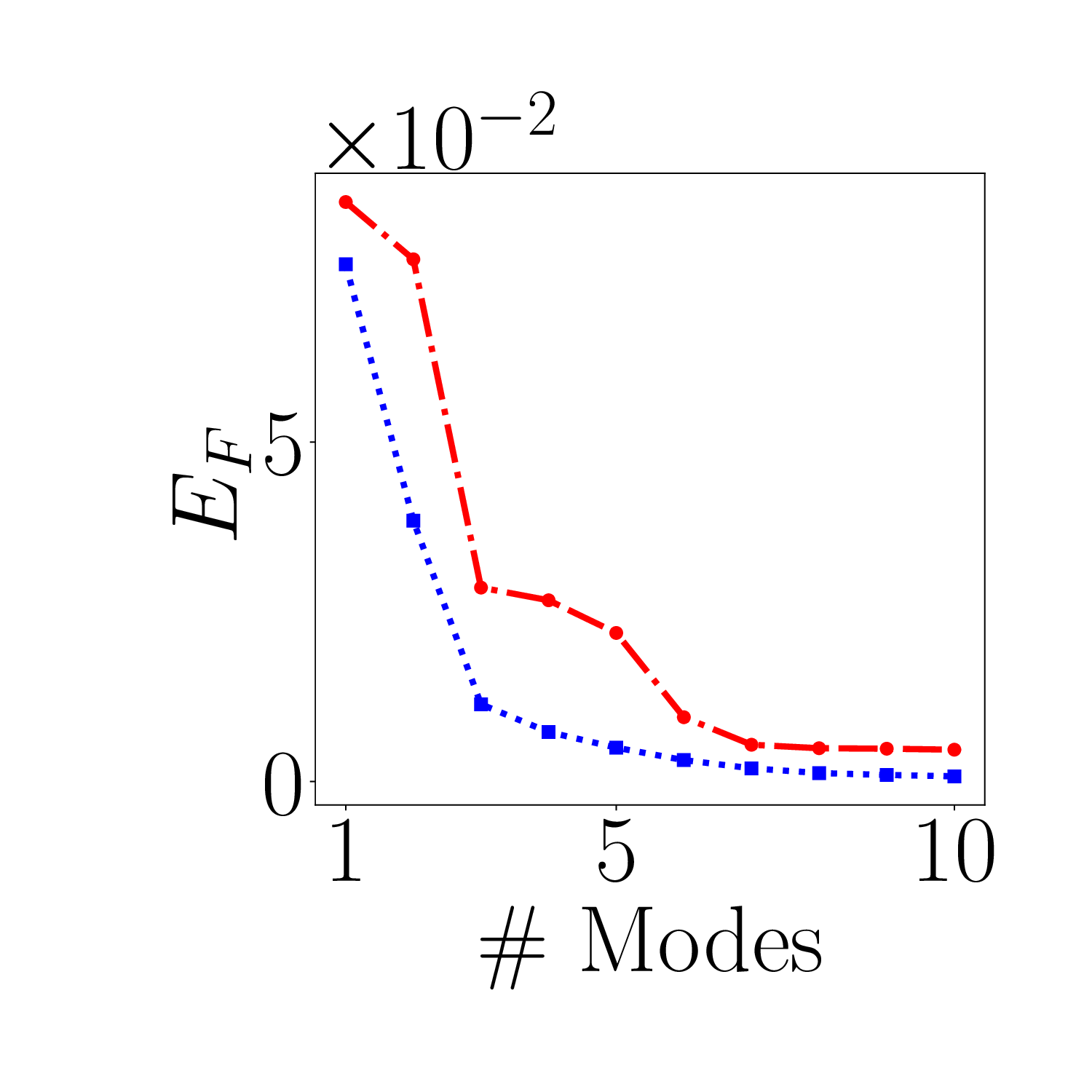}
    \includegraphics[width=.3\textwidth]{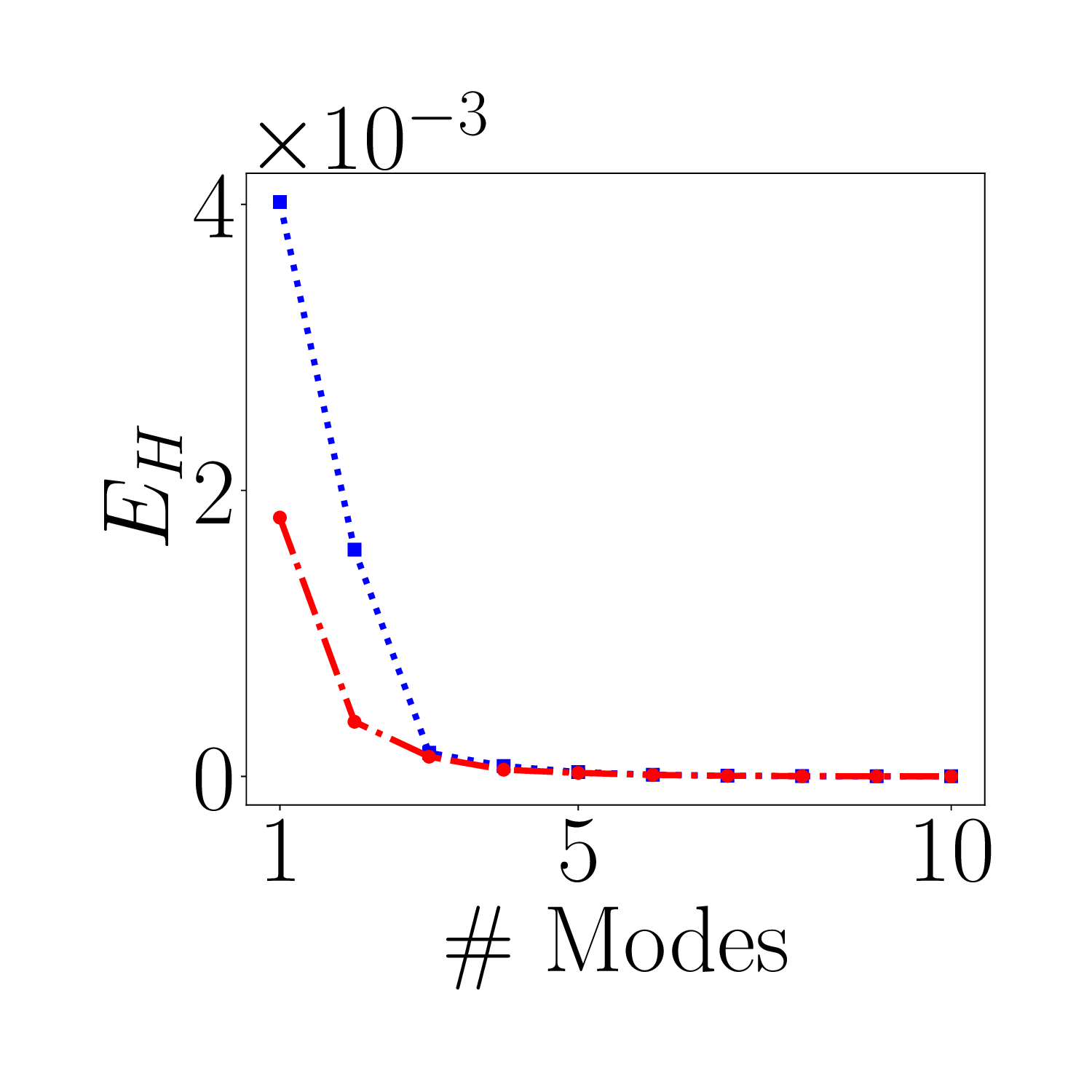}
    \includegraphics[width=.3\textwidth]{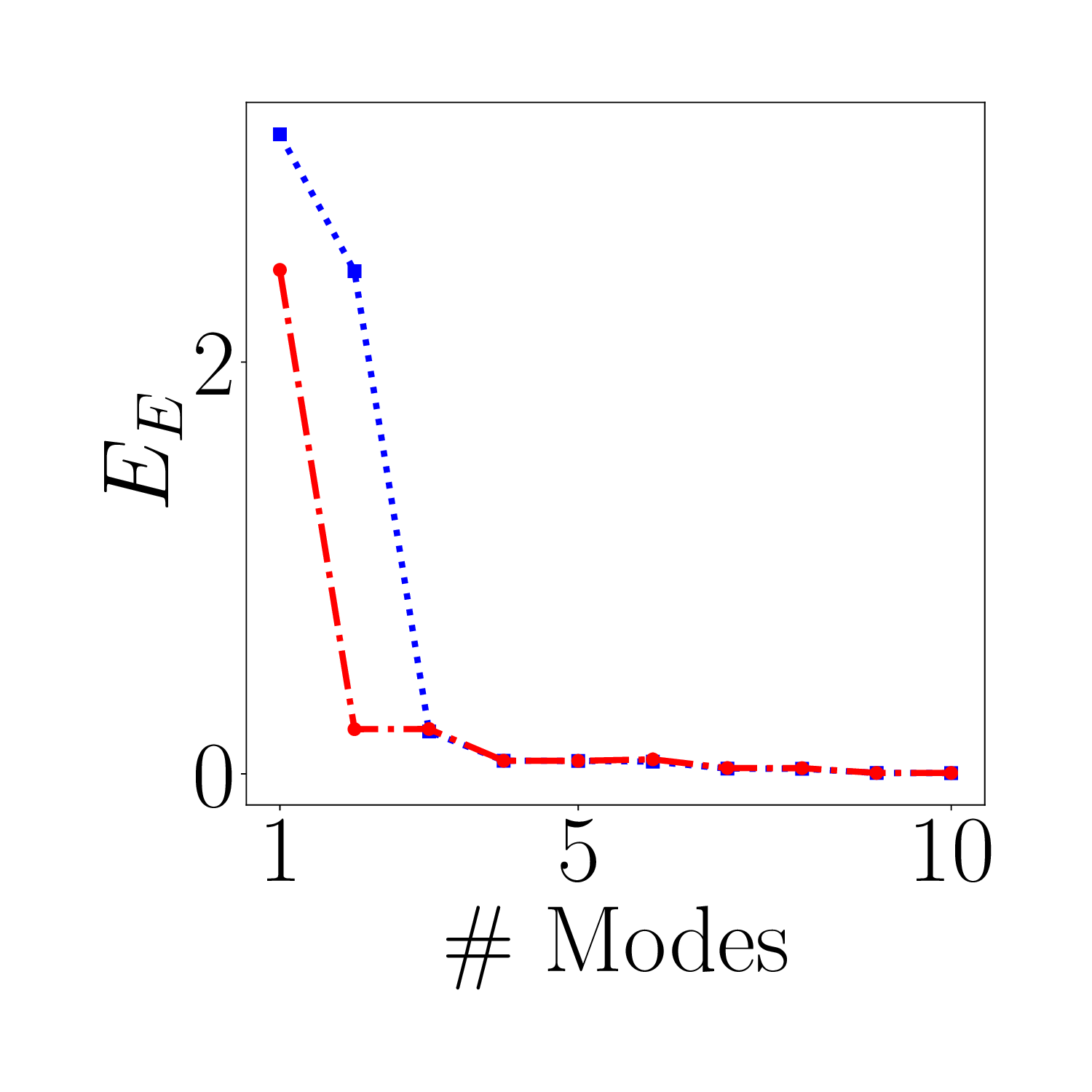}\\
    \includegraphics[width=.5\textwidth]{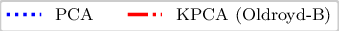}
    
    \caption{Comparing reconstruction performance using PCA and Kernel PCA (KPCA) with the Oldroyd-B kernel for simulations performed using the Oldroyd-B stress model for different values of $\theta$ in rows (a) $\theta = 0.001$, (b) $\theta = 1$ and (c) $\theta = 1000$. Error is measured using Eq.~\ref{E1} (left), Eq.~\ref{E2} (center), and Eq.~\ref{E3} (right).}
    \label{comparison_OldB_theta_modes}
\end{figure}

Now we select the leading $r=2$ modes and compare how well the transient dynamics of the flow are captured using the Oldroyd-B kernel and ordinary PCA-based methods.
In Fig.~\ref{energy_oldB_time} we plot the ground truth mechanical energy and the mechanical energies of our $2$-mode reconstructions across four values of $\theta$.
In each case, the reduction and reconstruction of the flow using the correct Oldroyd-B kernel accurately captures the total mechanical energy of the flow at all times, whereas the na\"{i}ve PCA-based method fails to captured the total energy.
The error using PCA is most profound for the small $\theta$ cases shown in Fig.~\ref{energy_oldB_time}(a) and Fig.~\ref{energy_oldB_time}(b).

\begin{figure}[!ht]
  \centering
a)  \includegraphics[scale=0.35]{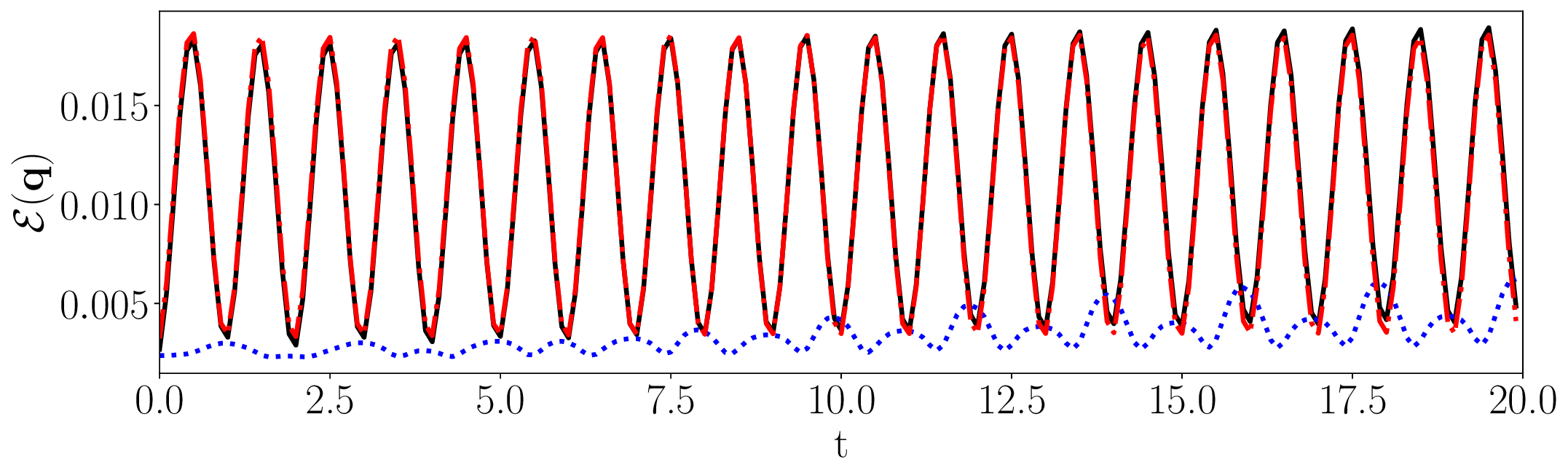}\\
b)  \includegraphics[scale=0.35]{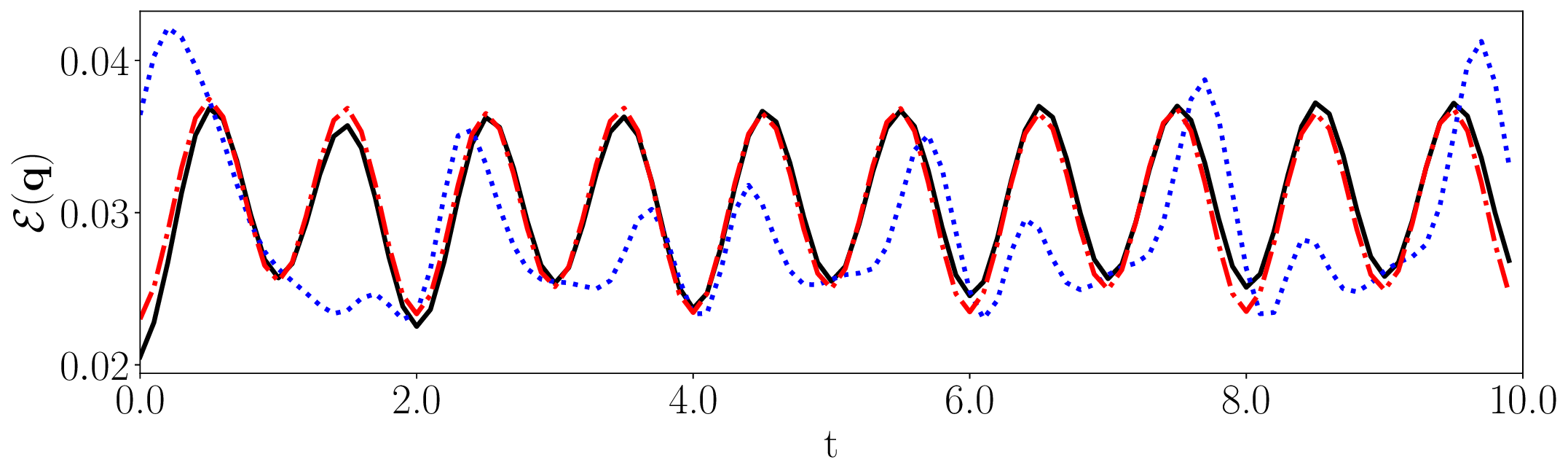}\\
c)  \includegraphics[scale=0.35]{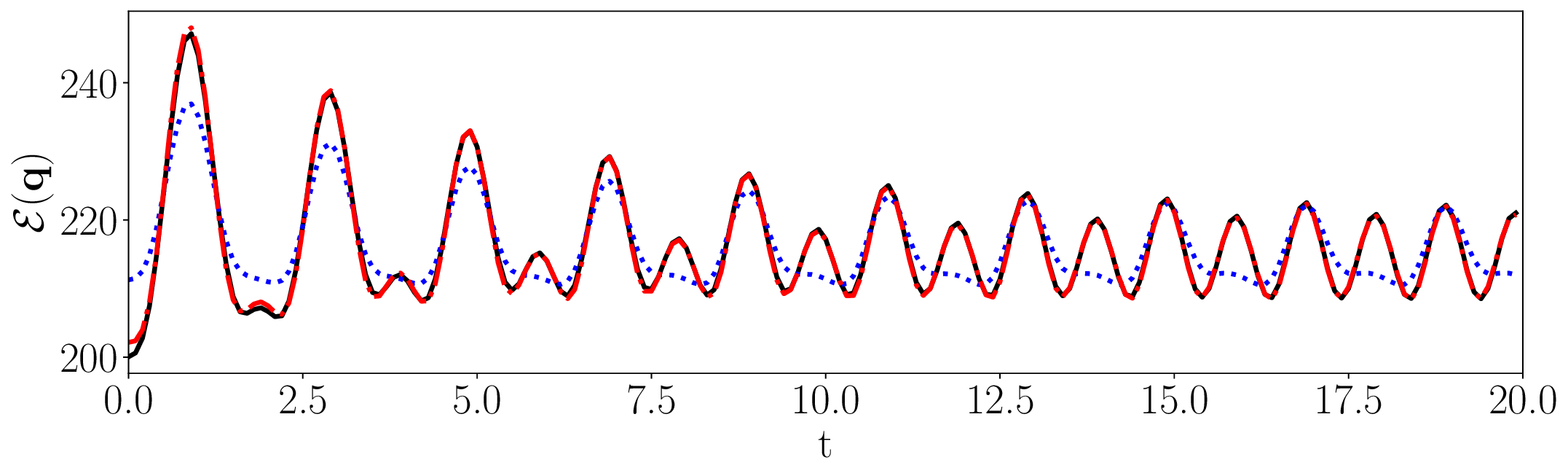}\\
d)  \includegraphics[scale=0.35]{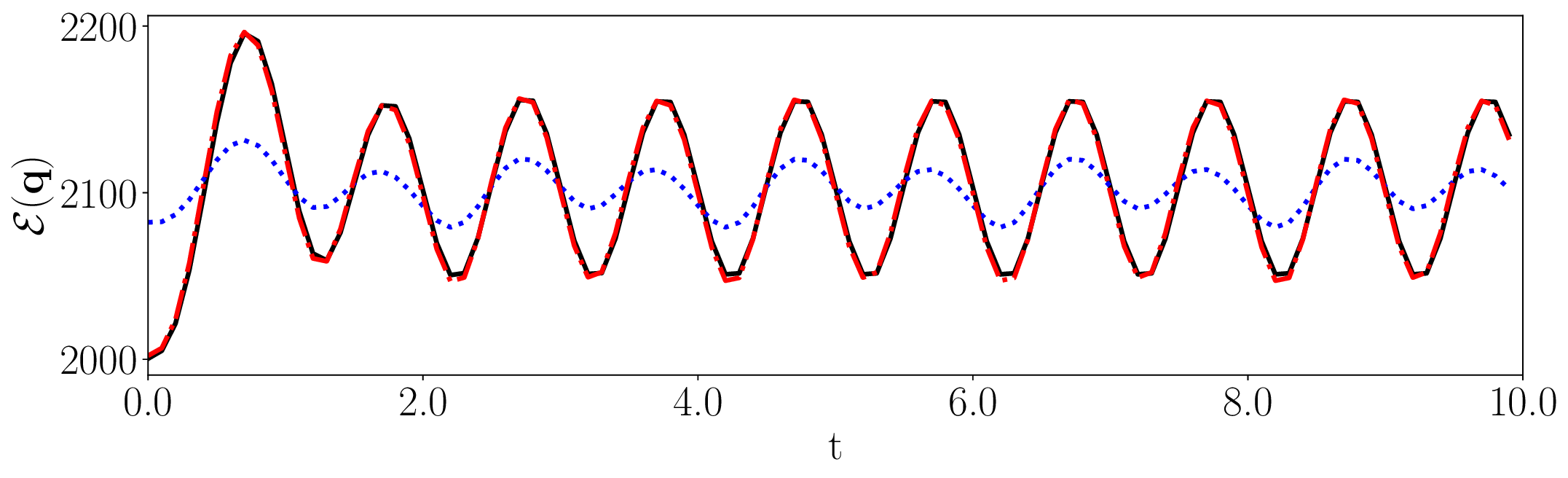}
  \includegraphics[width=.9\textwidth]{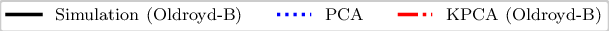}
  \caption{Total mechanical energy (\eqref{eqn:energy}) as function of time, reconstructed using $r=2$ principal components (modes):  
  a) $\theta = 0.001$, b) $\theta = 0.01$, c) $\theta = 100$ and d) $\theta = 1000$.
  }
  \label{energy_oldB_time}
\end{figure}

Finally, figure \ref{recons_OldB} provides a comparison between simulation and $2$-mode reconstructions for distribution of mechanical energy in a fixed snapshot at $t=2.5$ for the $\theta = 0.001$ case. 
Consistent with the previous discussions concerning the error plots in Fig.~\ref{energy_oldB_time}(a), the reconstructed total mechanical energy obtained using the Oldroyd-B kernel exhibits an exceptional agreement with the simulation, even when only two modes are retained. 
On the other hand, the na\"{i}ve PCA-based reconstruction fails to capture the distribution of mechanical energy in the snapshot.

\begin{figure}[!ht]
  \centering
\includegraphics[width=.9\textwidth]{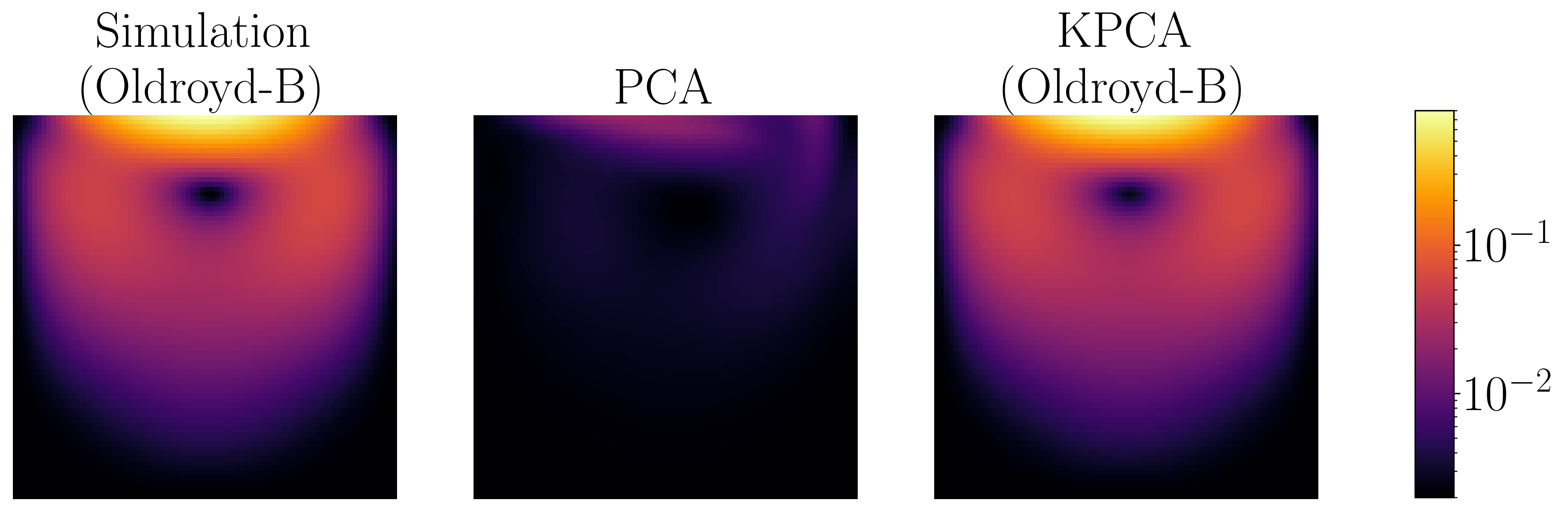}\\
  \caption{Spatial distribution of the total mechanical energy (\eqref{eqn:energy}) at time $t = 2.5$ for $\theta = 0.001$. Results obtained by the simulation (left), PCA (center) and KPCA (right). The reconstructed values were obtained using $r=2$ modes.
  }
  \label{recons_OldB}
\end{figure}

\subsection{Results using nonlinear stress models}

This section explores the use of nonlinear stress models and their associated kernel functions.
We study whether, and to what extent KPCA-based dimensionality reduction and reconstruction benefits from choosing the kernel function in accordance with the underlying stress model used to generate the simulation data.
To do this, we fix $\theta = 1$ by setting $\Rey = 1$, $\Wei = 0.5$, and $\beta=0.5$, and we simulate the lid-driven cavity flow using a variety of nonlinear stress models across a range of parameter values modulating the degree of nonlinearity.
As above, we compare the performance of KPCA-based dimensionality reduction and reconstruction using the kernel function matching the underlying simulation against kernel functions that do not.
In particular, we are interested in when a nonlinear kernel function is necessary, or whether there are parameter regimes where we can simply use the Oldroyd-B kernel, which has the benefit of an explicit bijective feature map given by Eq.~\eqref{eqn:OldB_explicit_feature_map}.
The error $E_H$ in Eq.~\ref{E2} is always computed using the kernel associated with the underlying stress model.

In Fig.~\ref{fene-p_merged}, the simulation was performed using the FENE-P stress model with parameter $L^2 = 5$, and we compare the corresponding kernel against Oldroyd-B.
While the Oldroyd-B kernel yields better reconstructions in the $L^2(\Omega)$ sense measured by $E_F$, the reconstructions using small numbers of kernel principal components, or ``modes'', are worse than those obtained using the FENE-P kernel function in the energetic sense measured by $E_H$ and $E_E$.

\begin{figure}[!ht]
  \centering
 a)
    \includegraphics[width=.3\textwidth]{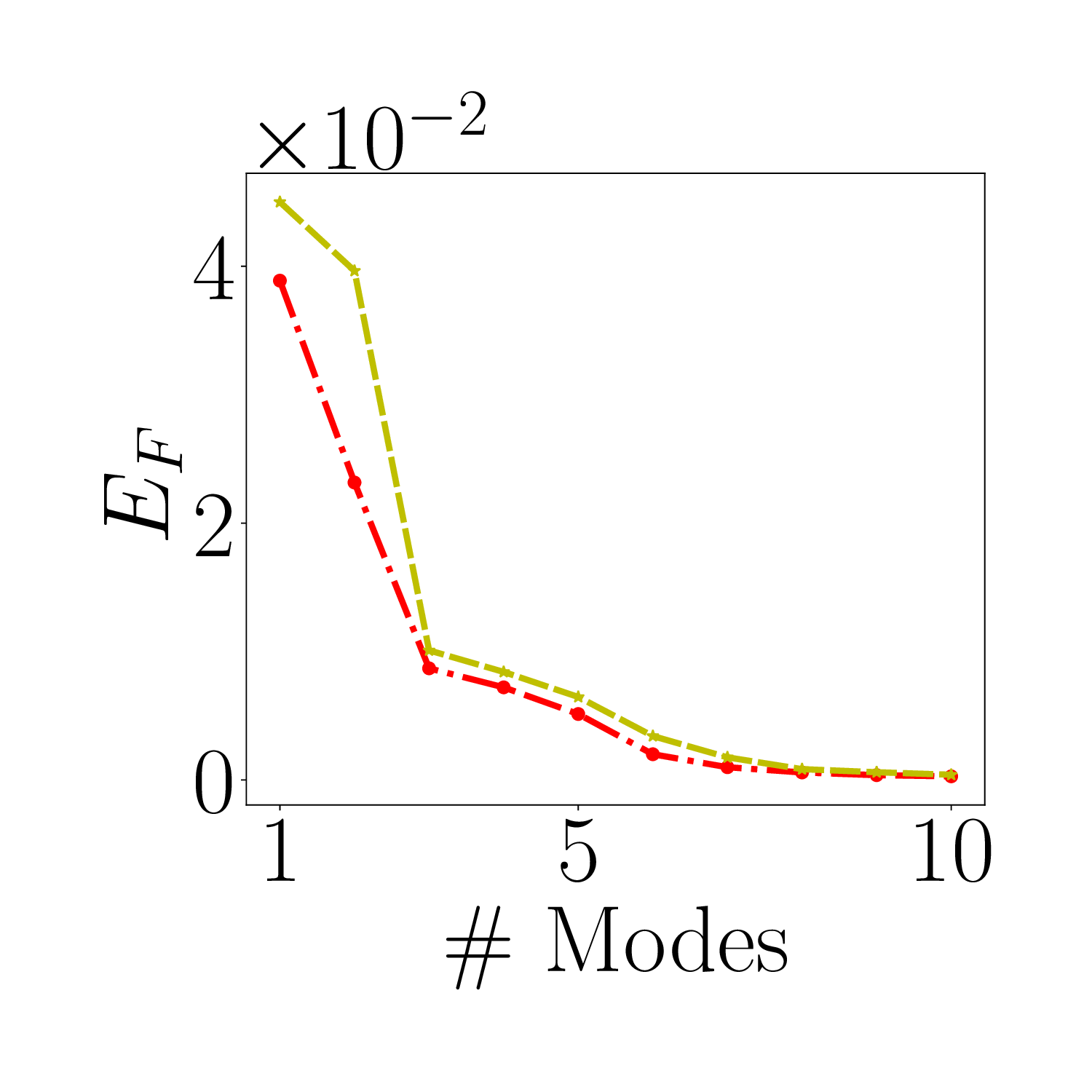}
    \includegraphics[width=.3\textwidth]{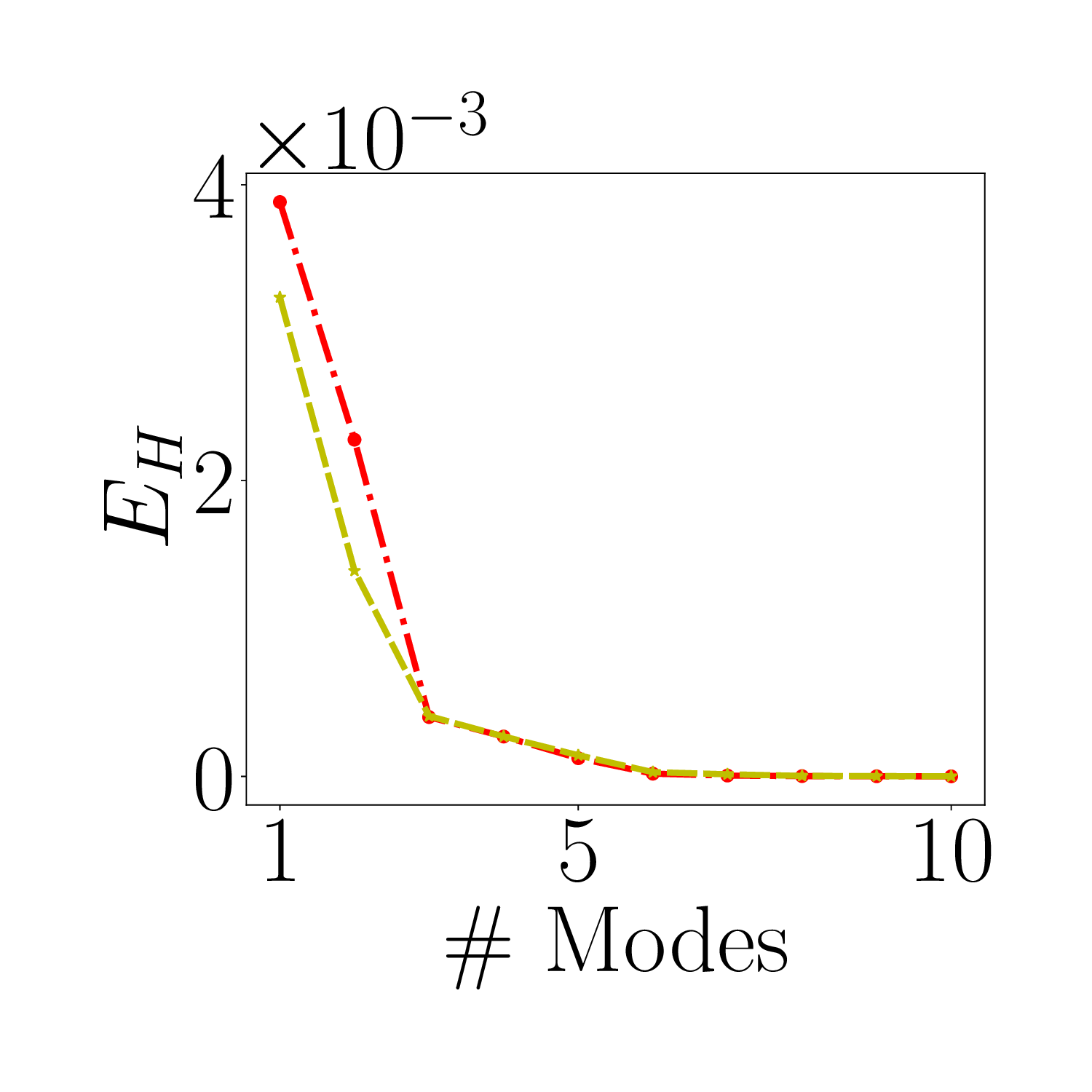} 
    \includegraphics[width=.3\textwidth]{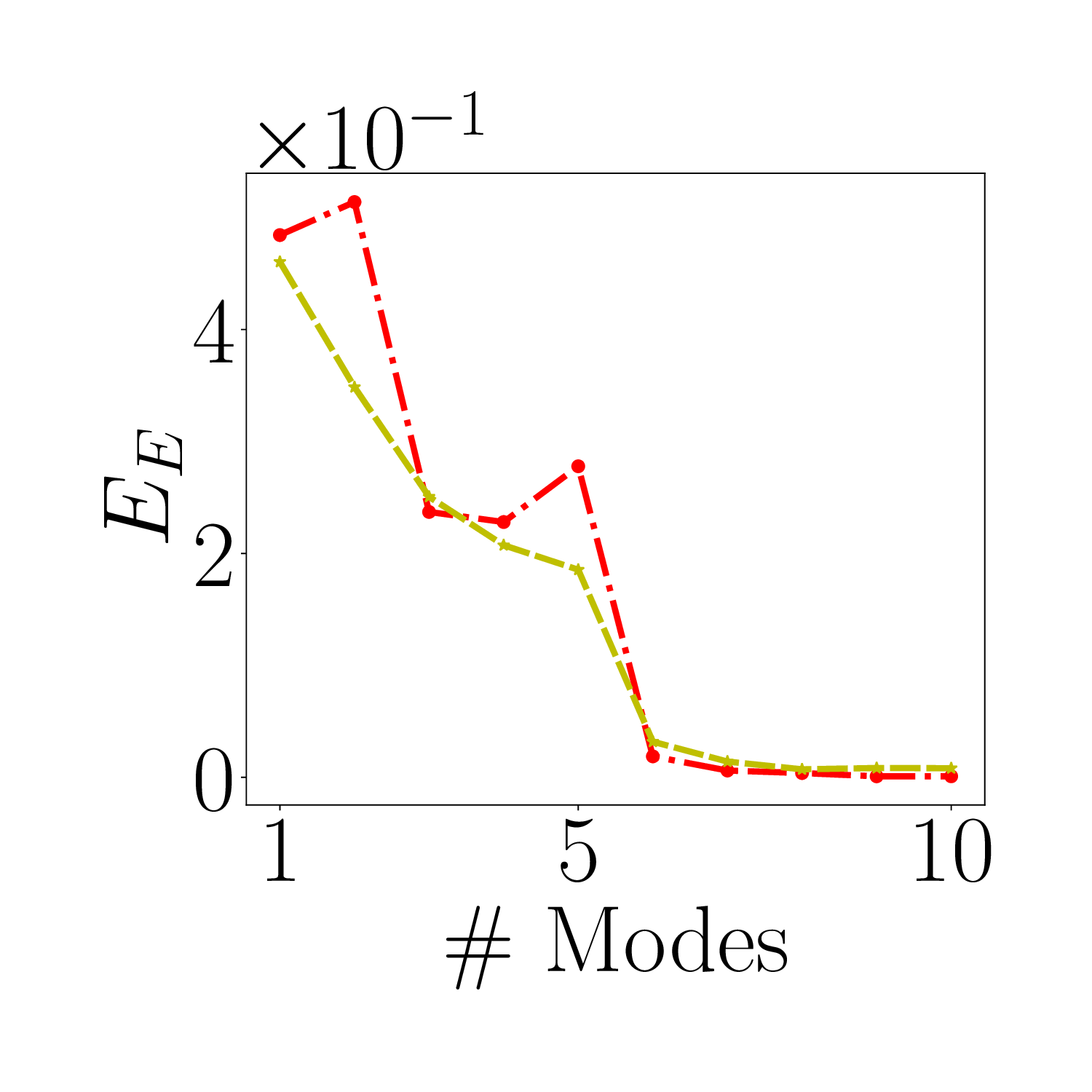} \\
b) \includegraphics[width=.9\textwidth]{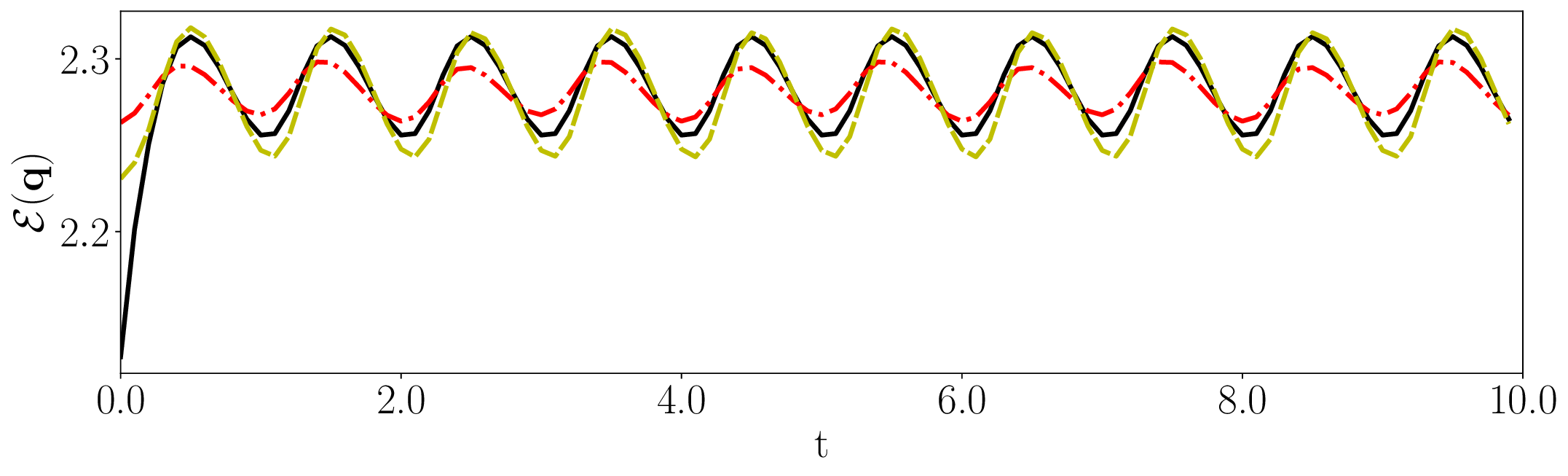}\\
   \includegraphics[width=.9\textwidth]{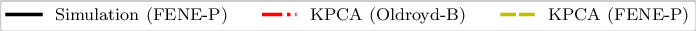}
  \caption{Results for the $\theta=1$ case using the nonlinear FENE-P model with $L^2= 5$. Row (a) shows reconstruction errors computed using Eq.~\ref{E1} (left), Eq.~\ref{E2} (center) and Eq.~\ref{E3} (right). Row (b) shows the total mechanical energy (\eqref{eqn:energy}) as function of time for fields reconstructed using $r=2$ kernel principal components (modes). 
  }
  \label{fene-p_merged}
\end{figure}

We repeat this experiment in Fig.~\ref{PTT_merged} using the nonlinear PTT stress model with parameter $\varepsilon = 0.5$ for the simulation.
A similar trend is observed; matching the kernel function to the stress model is beneficial for reconstruction using small numbers of modes in the energetic sense measured by $E_H$ and $E_E$.
This improvement is dramatic when $r=2$ modes are used for reconstruction, as illustrated by the time histories of total mechanical energy plotted in Fig.~\ref{PTT_merged}(b).
However, for intermediate numbers of modes, the simple Oldroyd-B kernel slightly outperforms the nonlinear PTT kernel in all error metrics.

\begin{figure}[!ht]
  \centering
 a)
    \includegraphics[width=.3\textwidth]{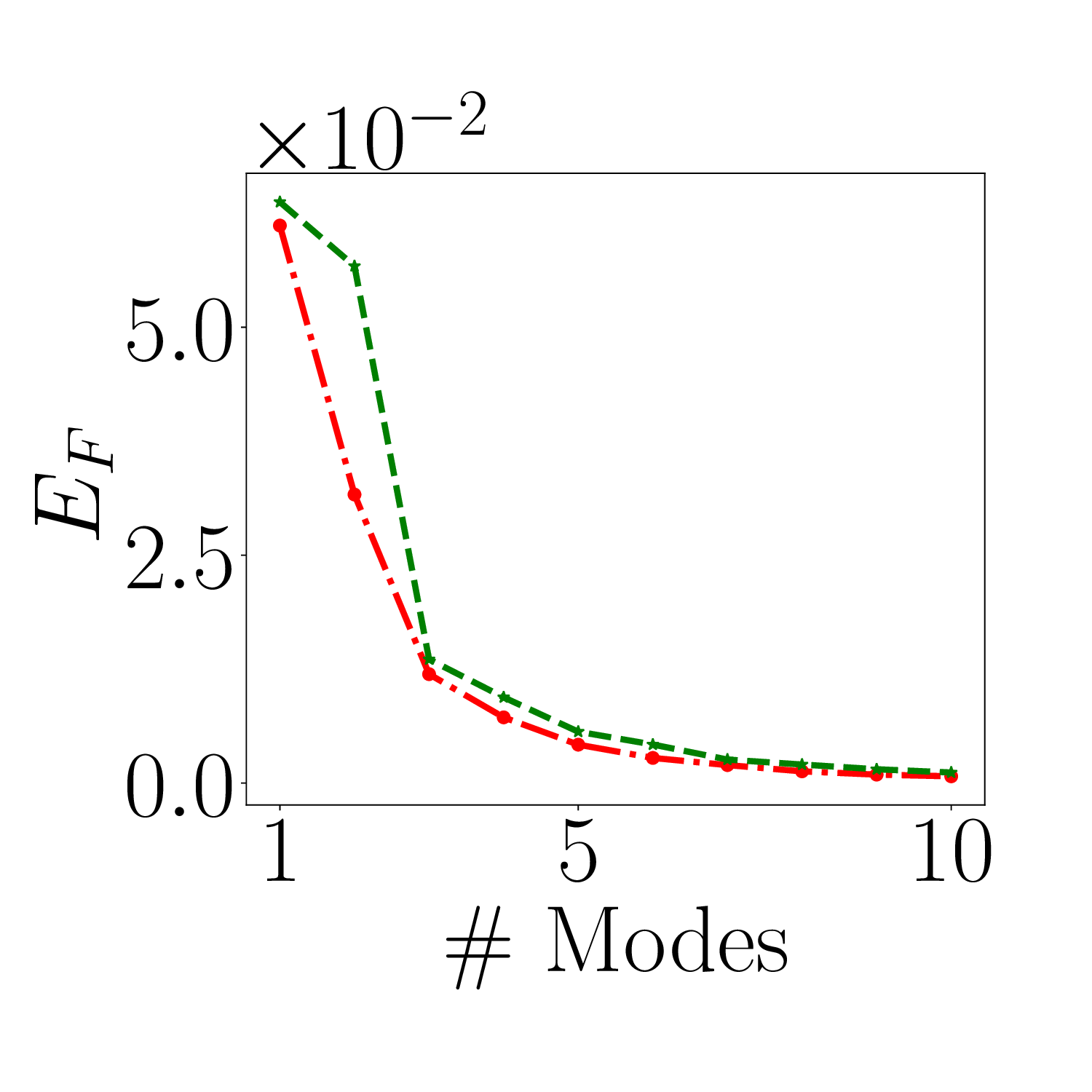}
    \includegraphics[width=.3\textwidth]{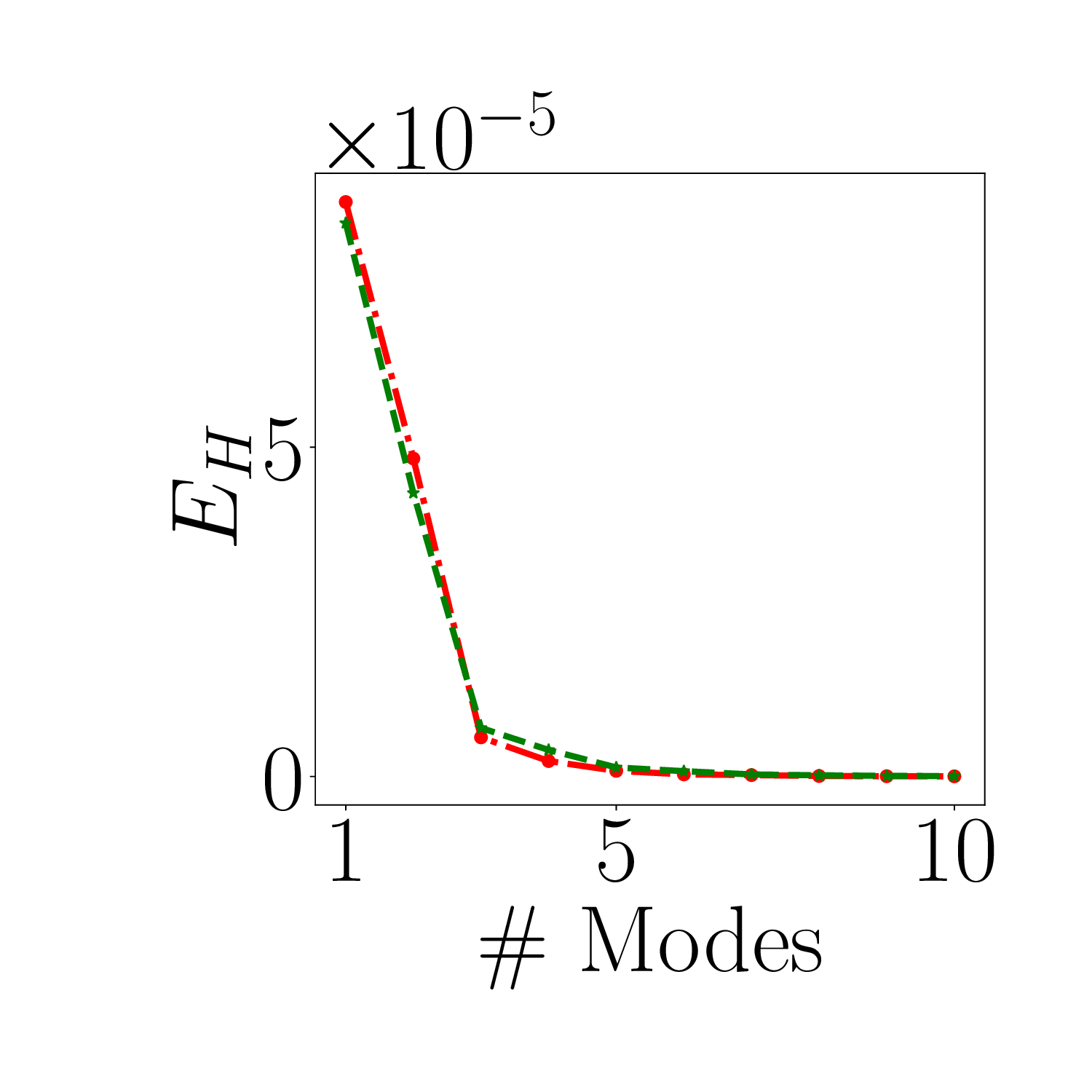} 
    \includegraphics[width=.3\textwidth]{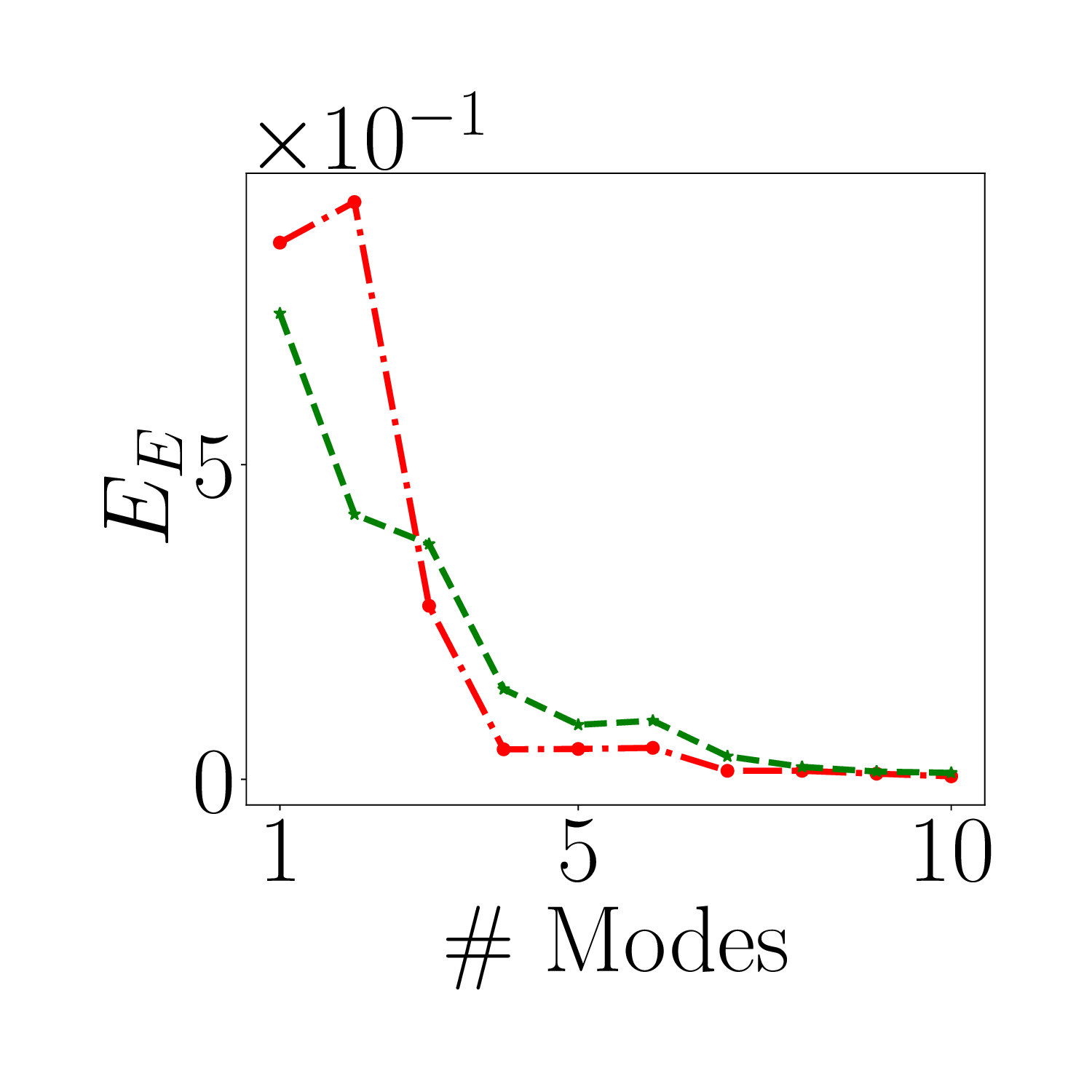} \\
b)  \includegraphics[width=.9\textwidth]{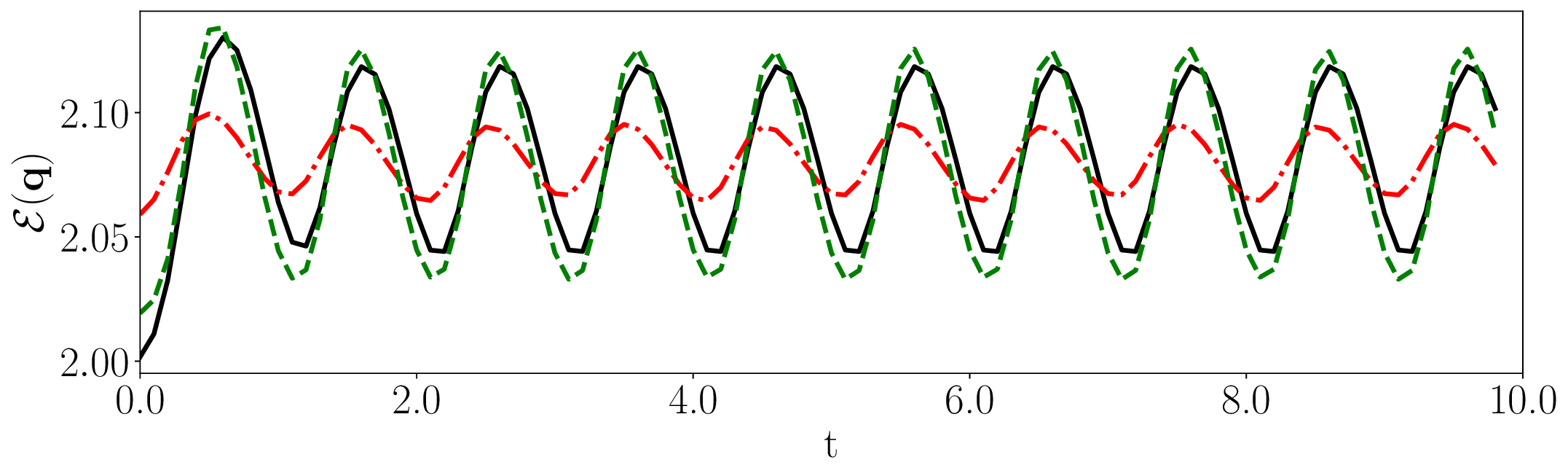}
    \includegraphics[width=.9\textwidth]{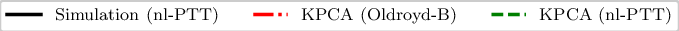}\\
  \caption{Analogue of Fig.~\ref{fene-p_merged} using the nonlinear PTT model (denoted nl-PTT) with $\varepsilon=0.5$.}
  \label{PTT_merged}
\end{figure} 

In Fig.~\ref{G_merged} we present analogous results for the Giesekus model with $\alpha=1$.
Again, for small numbers of modes, using the corresponding Giesekus kernel function tends to yield slight improvements over the Oldroyd-B kernel in the energetic sense measured by $E_H$ and $E_E$.
The performance using the two kernel functions is nearly the same when larger numbers of modes are used to reconstruct.

\begin{figure}[!ht]
  \centering
 a)
    \includegraphics[width=.3\textwidth]{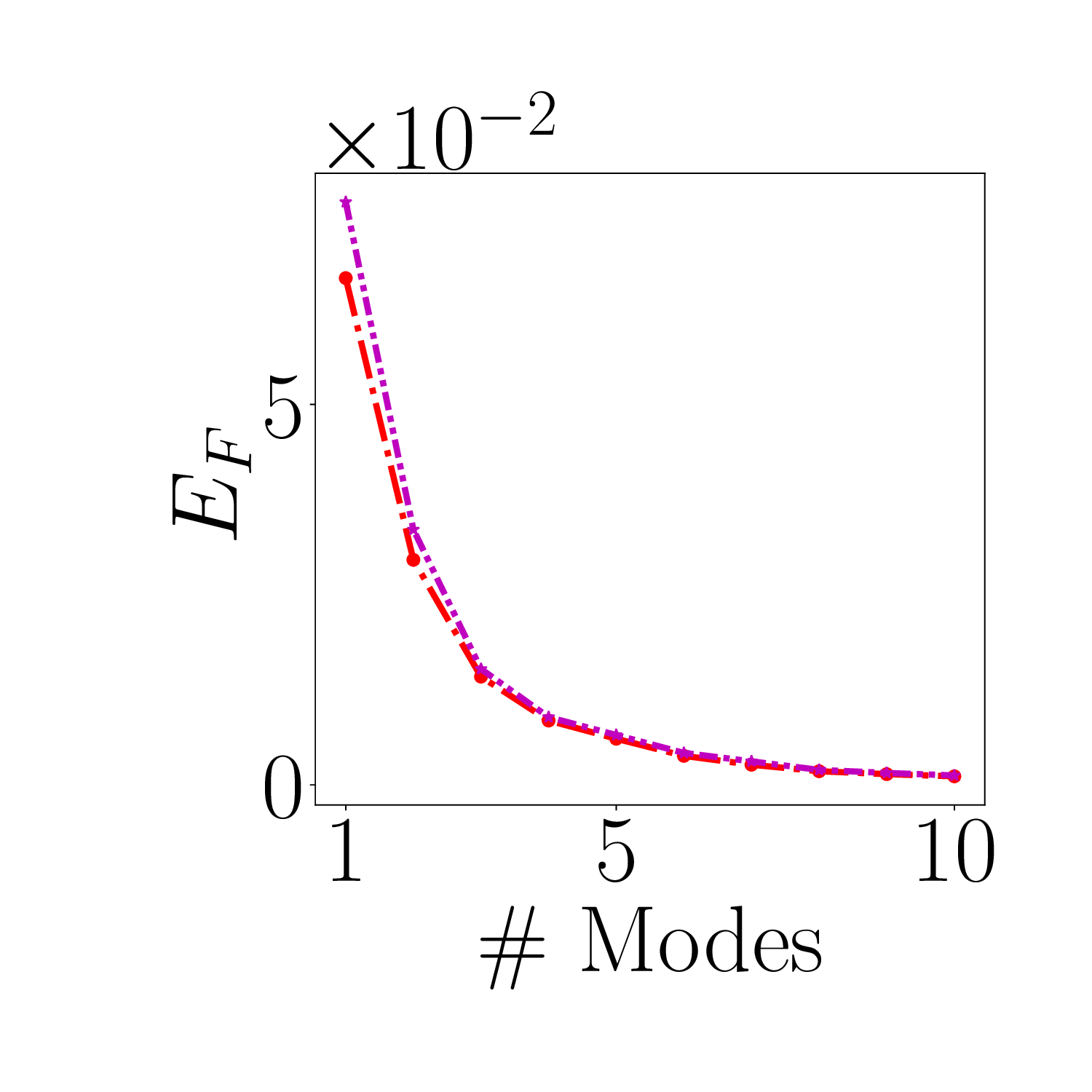}
    \includegraphics[width=.3\textwidth]{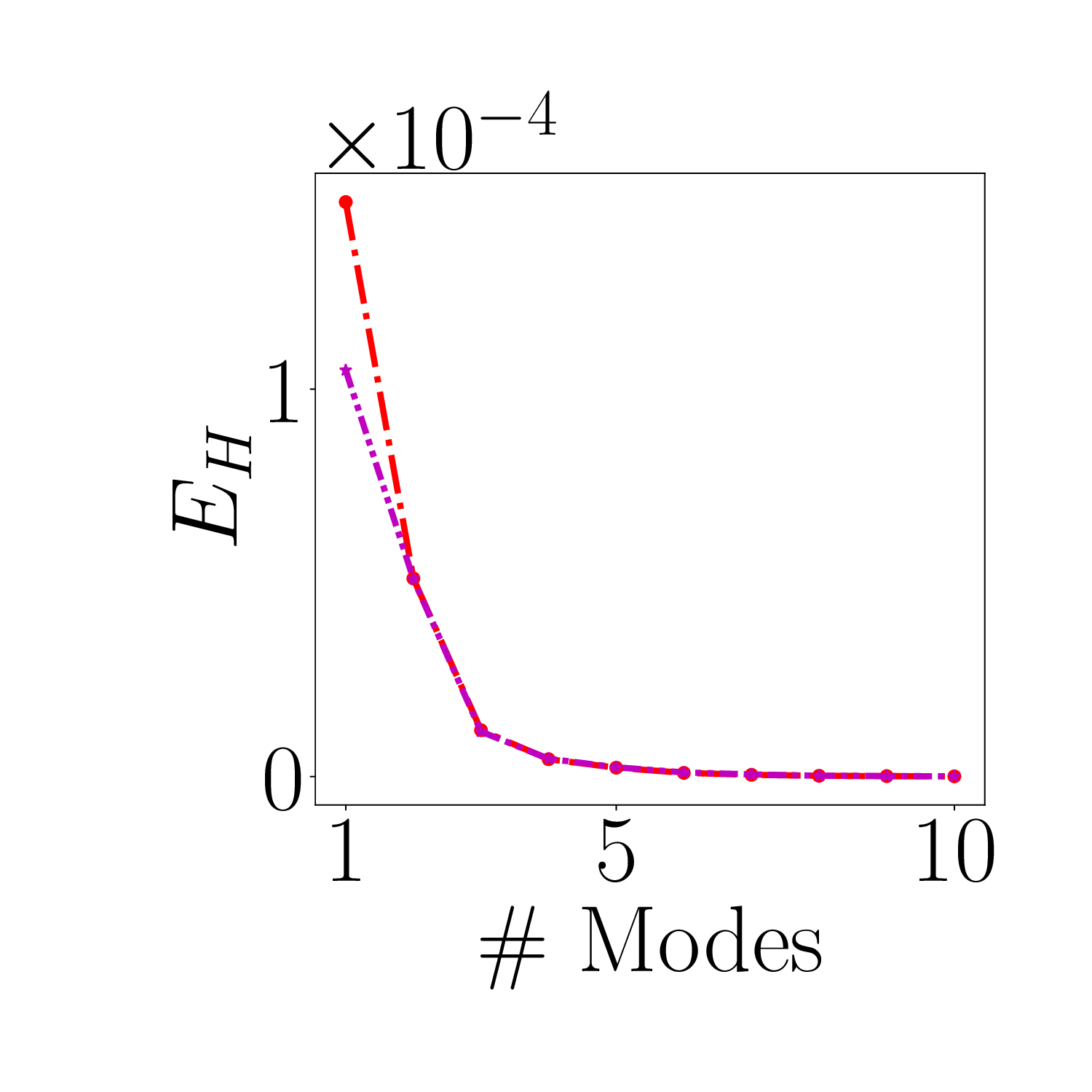} 
    \includegraphics[width=.3\textwidth]{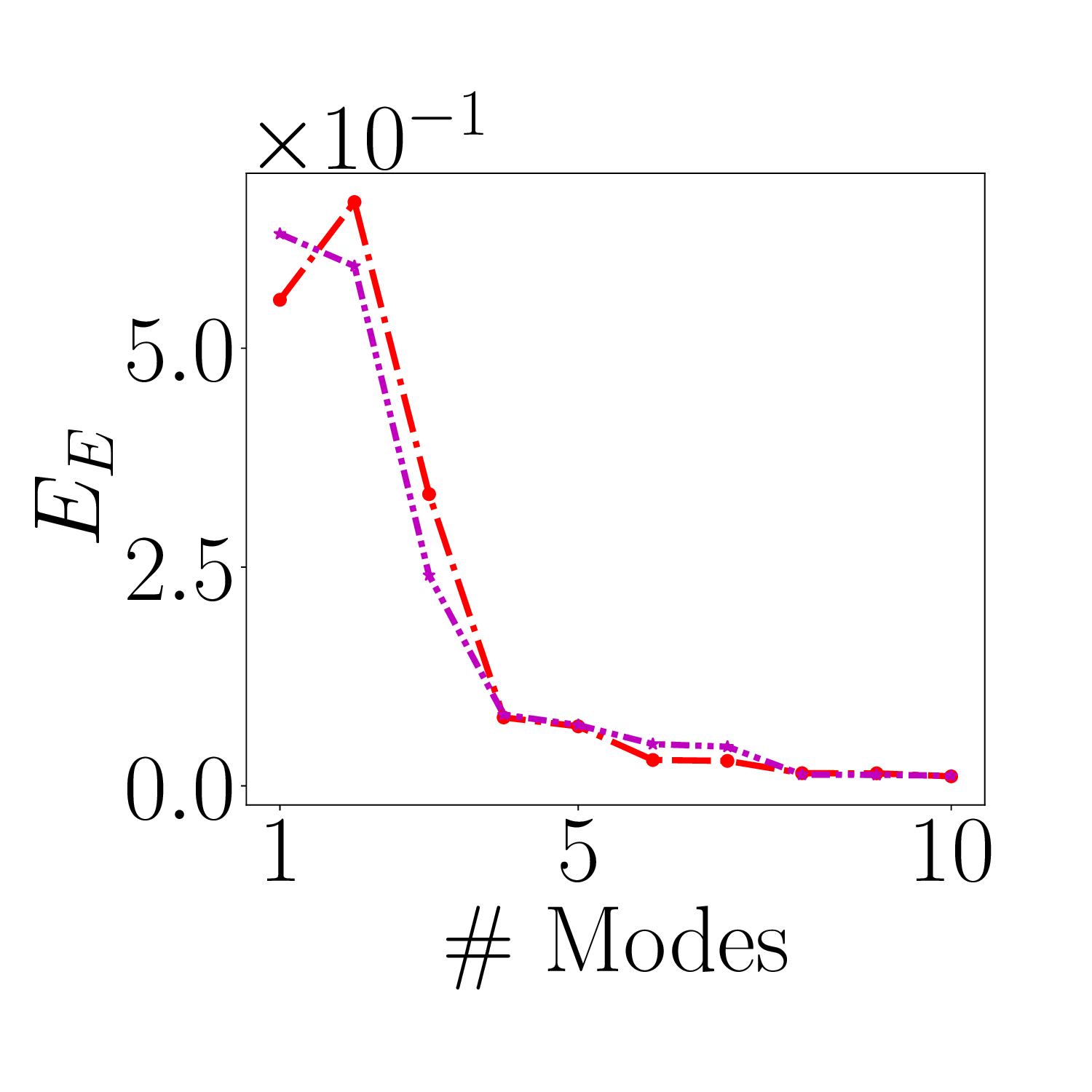} \\
b) \includegraphics[width=.9\textwidth]{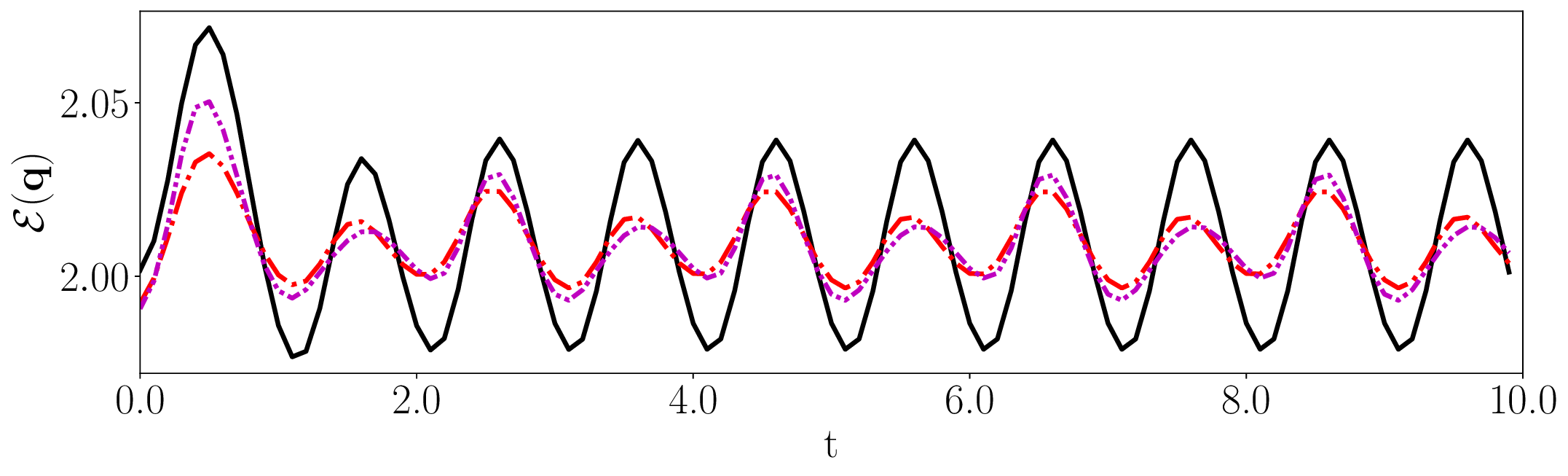}\\
   \includegraphics[width=.9\textwidth]{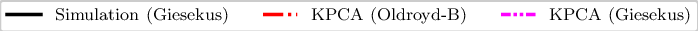}
  \caption{Analogue of Fig.~\ref{fene-p_merged} using the Giesekus model with $\alpha=1$.}
  \label{G_merged}
\end{figure}

We study the effect of the parameters in the FENE-P and nonlinear PTT stress models in Fig.~\ref{PTT_param_variantion}.
The number of kernel principal components (modes) used for reconstruction is fixed at $r=2$.
These nonlinear stress models become equivalent to the Oldroyd-B model in the limits $1/L^2 \to 0$ and $\varepsilon \to 0$.
As expected, the choice of kernel is unimportant at small values of these parameters and becomes more significant at larger values.
In both cases, it becomes important to use the kernel function matching the nonlinear stress model of the simulation only when the degree of nonlinearity is sufficiently high.
We do not present analogous results using the Giesekus model due to numerical challenges associated with simulating the flow at larger values of the parameter $\alpha$.

\begin{figure}[!ht]
  \centering
   a) \includegraphics[width=.3\textwidth]{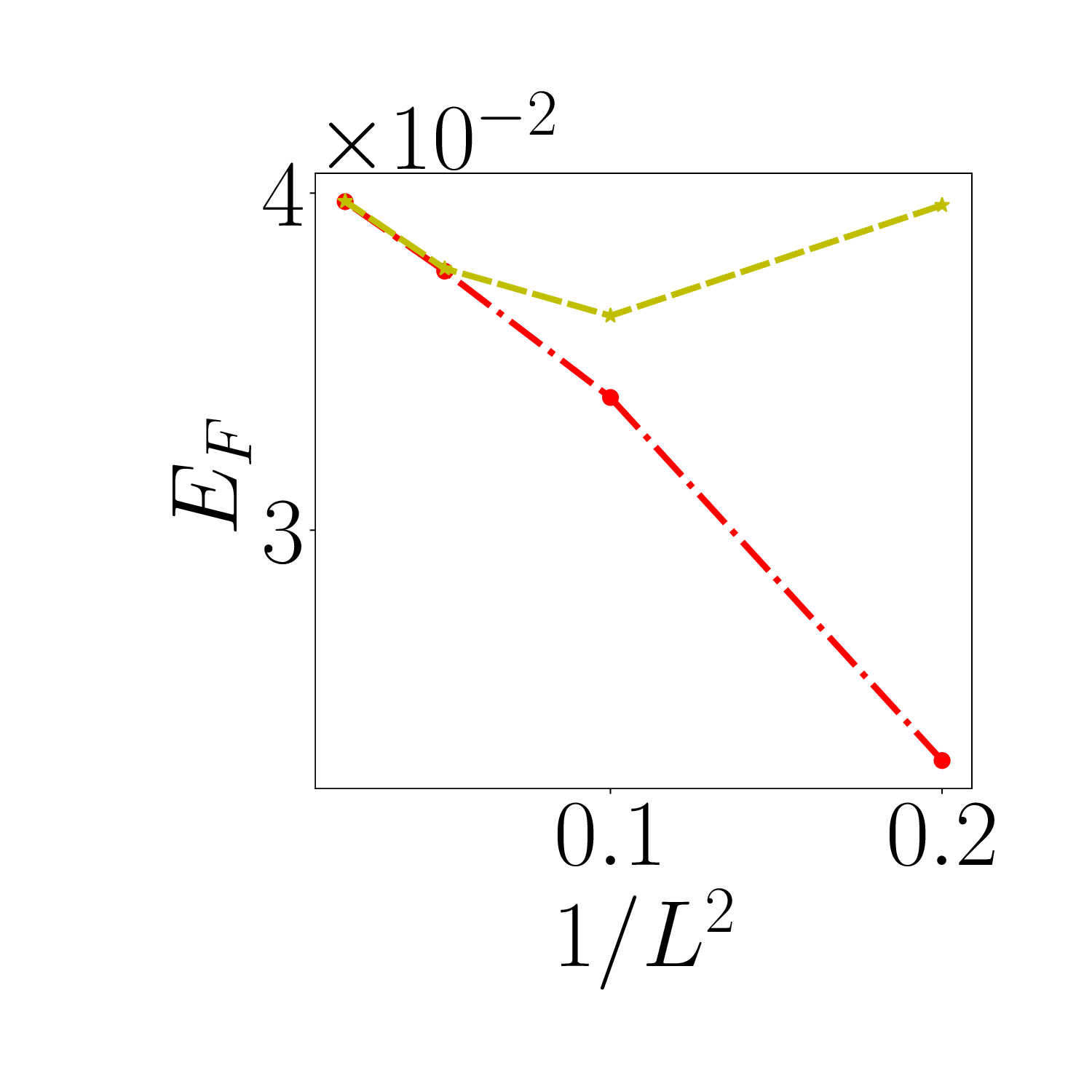}
    \includegraphics[width=.3\textwidth]{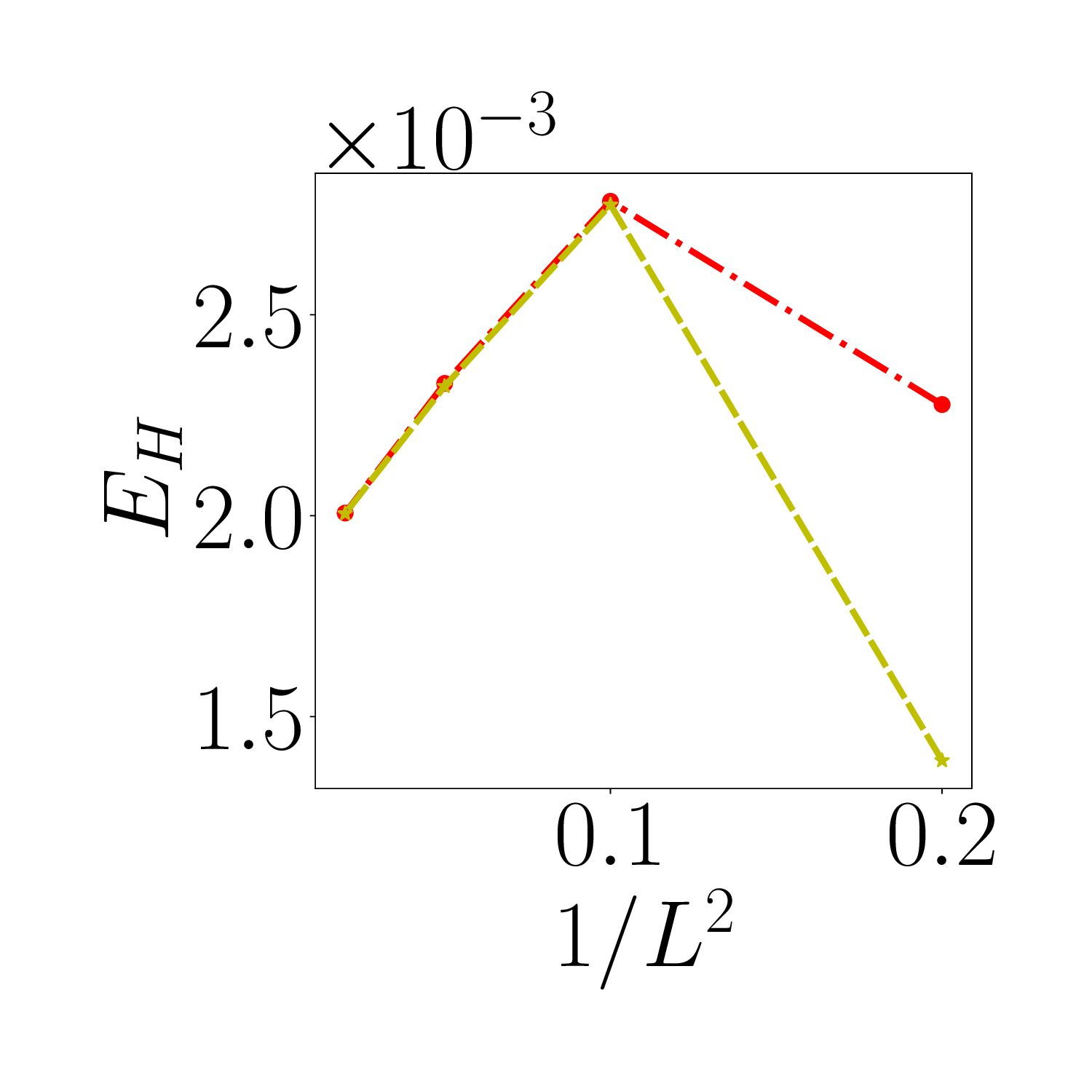} 
    \includegraphics[width=.3\textwidth]{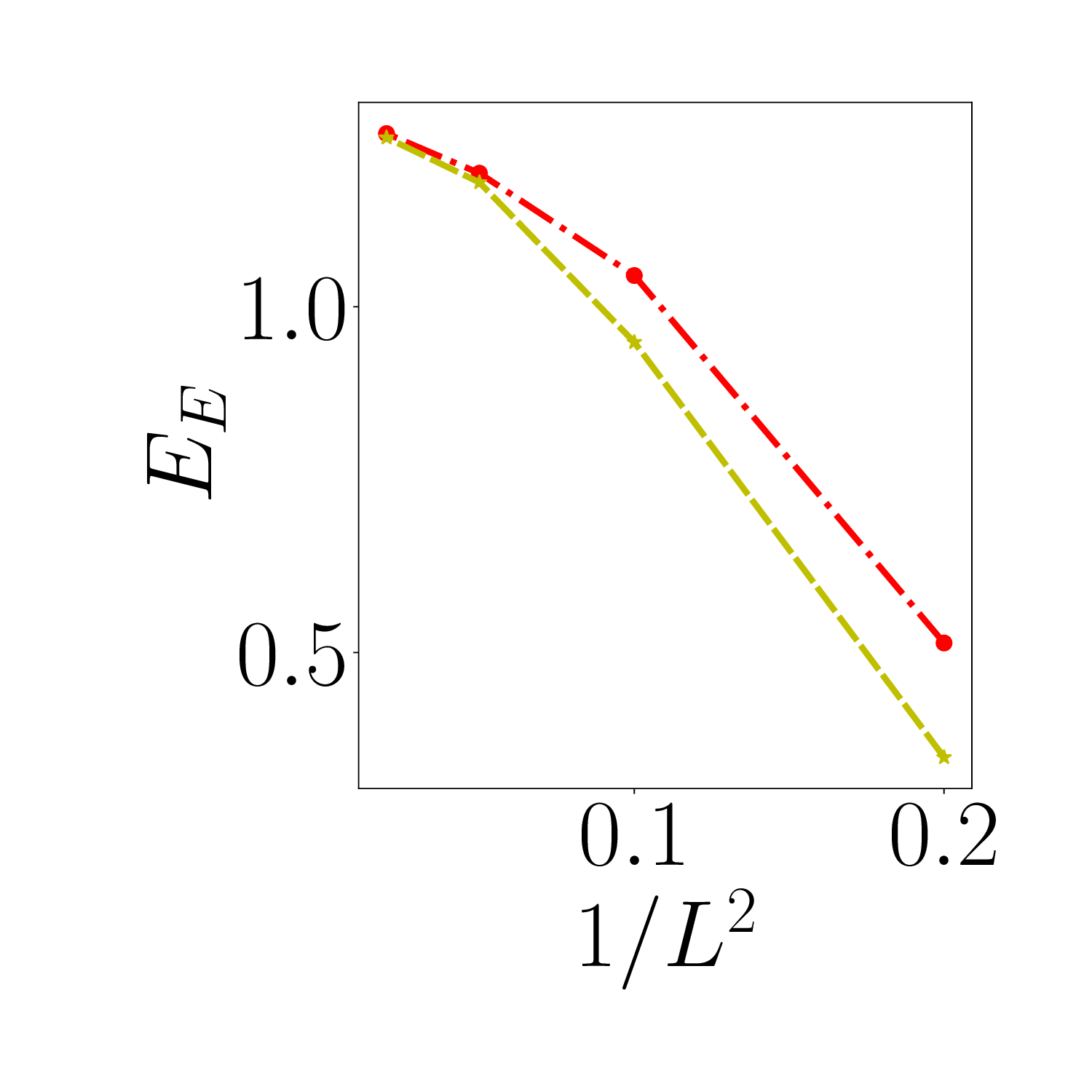} \\
    b)\includegraphics[width=.3\textwidth]{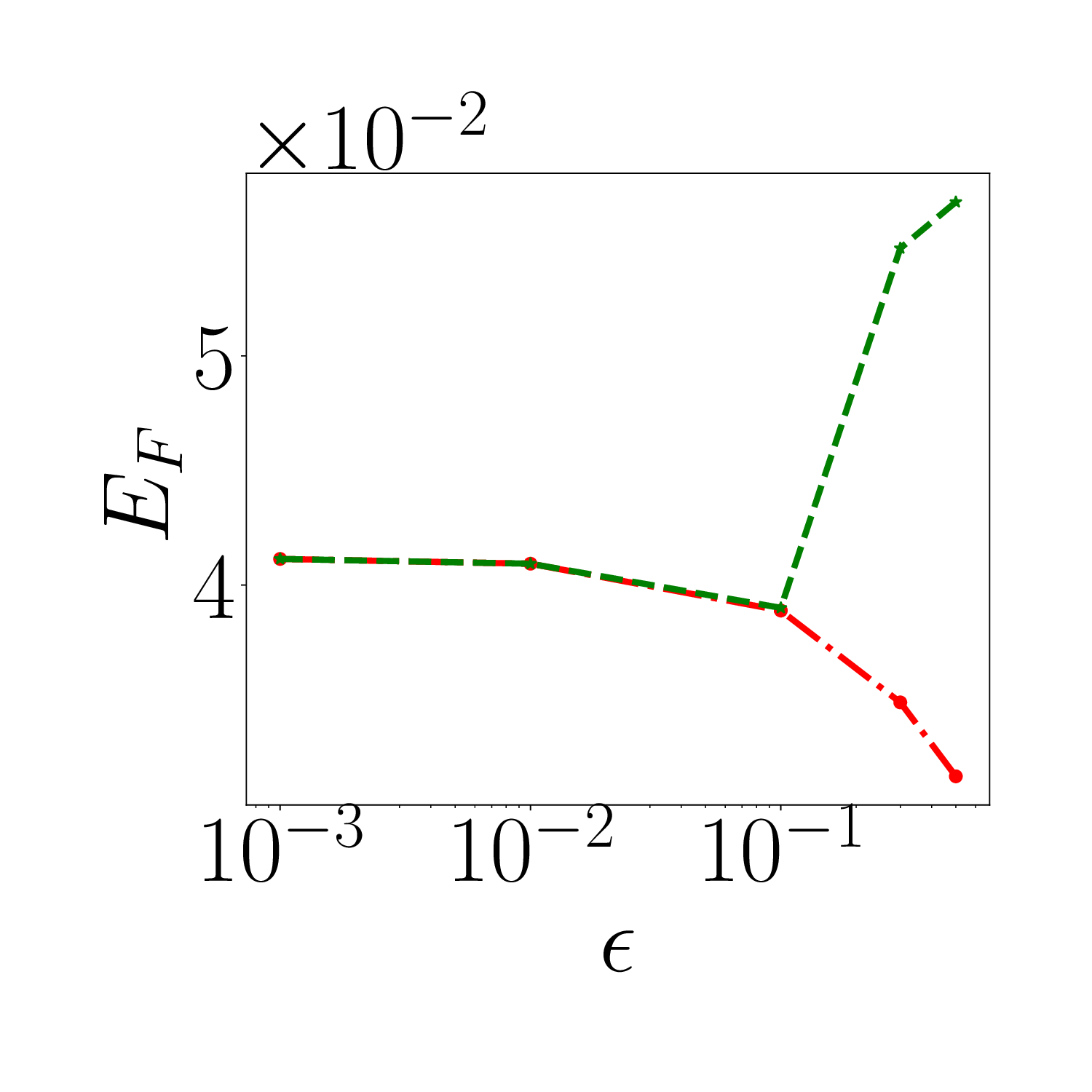}
    \includegraphics[width=.3\textwidth]{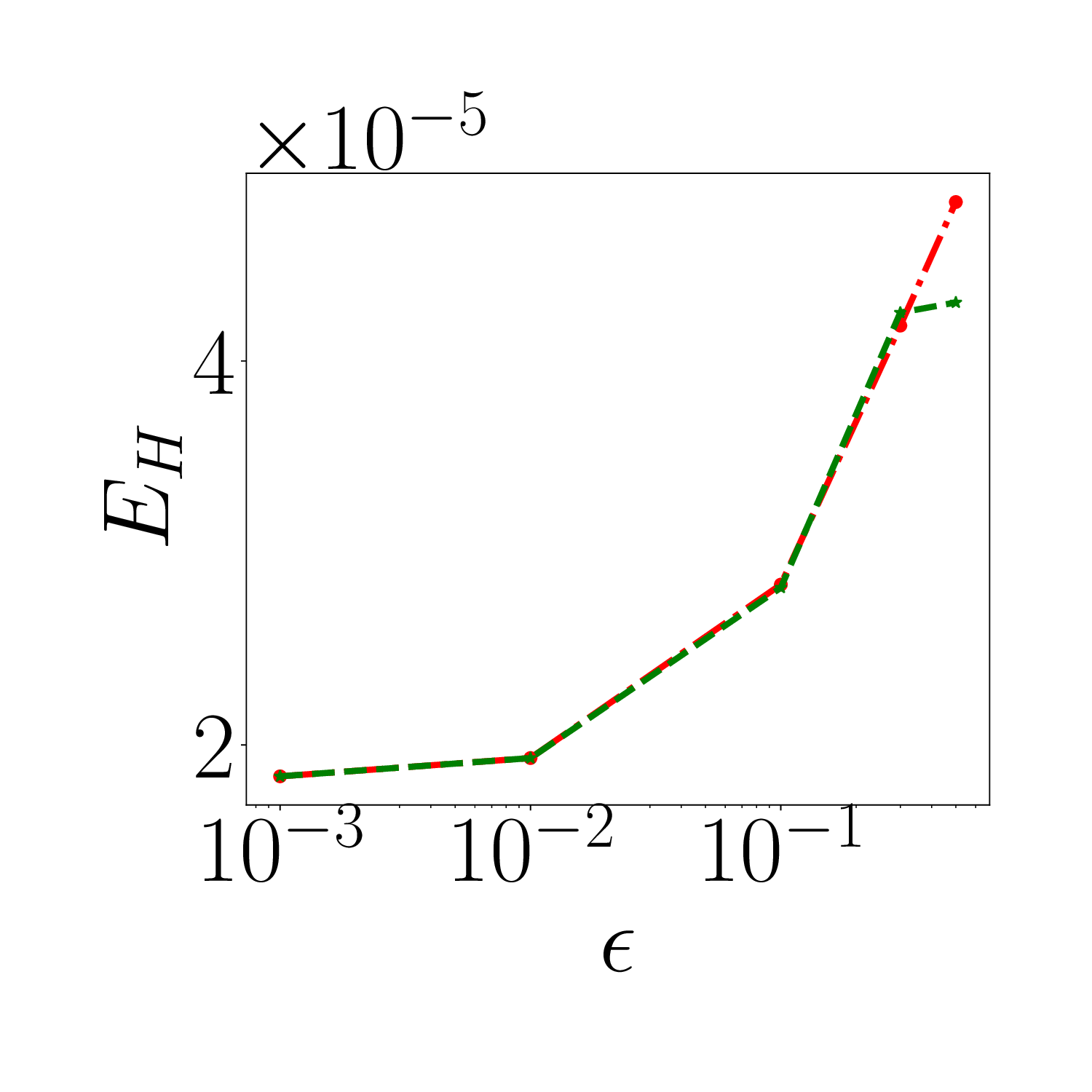} 
    \includegraphics[width=.3\textwidth]{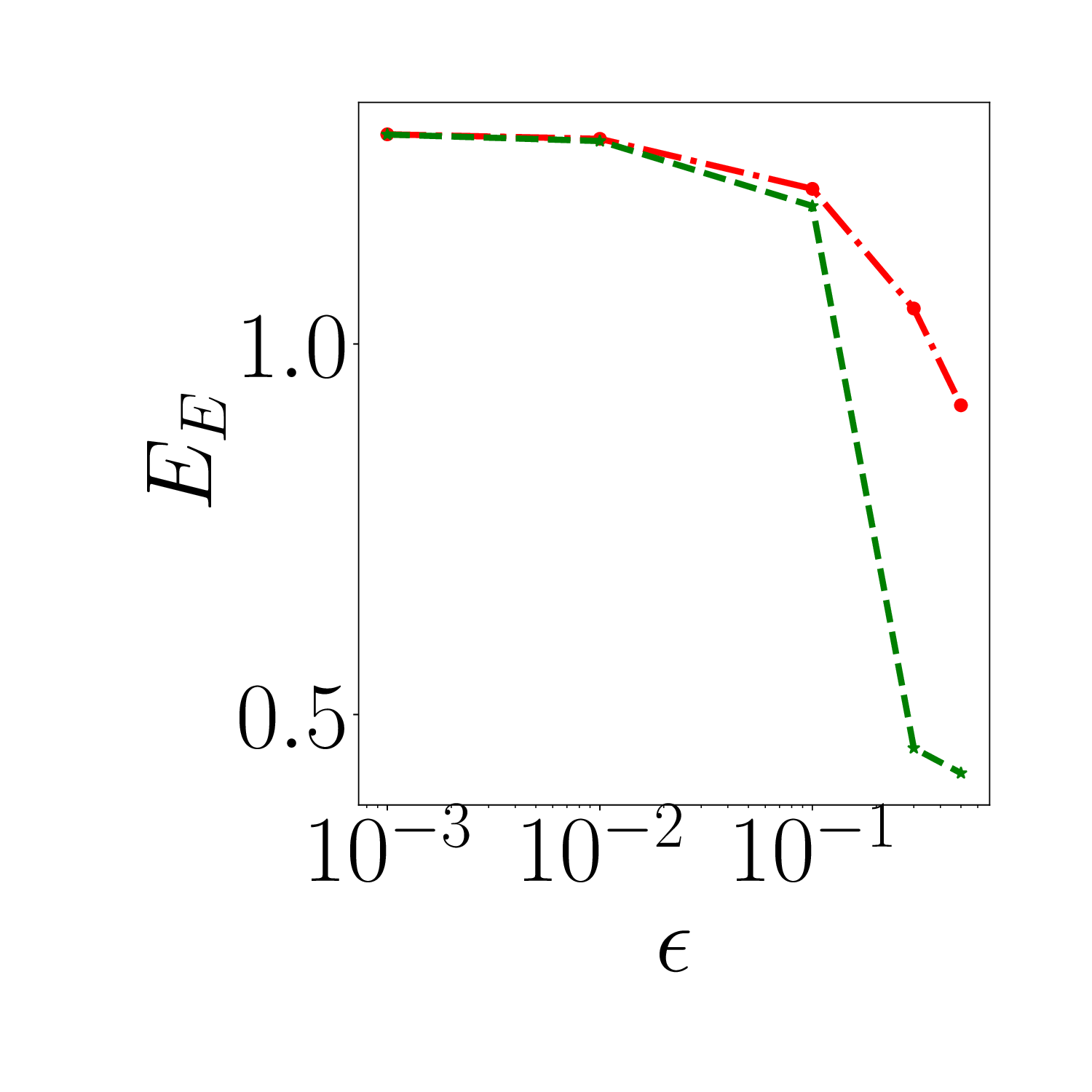} \\

    \includegraphics[width=.7\textwidth]{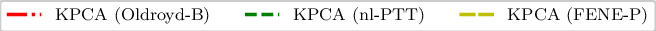}\\
  \caption{The effect of stress model nonlinearity on reconstruction performance using $r=2$ kernel principal components. Simulations in row (a) use the FENE-P model with different values of $\frac{1}{L_2}$ and simulations in row (b) use the nonlinear PTT model with different values of $\varepsilon$. In both cases we have $\theta=1$ and we compare to results obtained using the Oldroyd-B kernel.}
  \label{PTT_param_variantion}
\end{figure}

Finally, two-mode reconstructions of the instantaneous total mechanical energy fields are shown in Fig.~\ref{recons_ptt} at time $t=0.6$ for a simulation performed using the nonlinear PTT model with parameter $\varepsilon = 0.3$.
In this case with $\theta=1$, the dynamic variation is concentrated close to the lid (top of the spatial domain), which is in contrast to the global variations observed in Fig.\ref{recons_OldB} at $\theta = 0.001$.
Here, we see that using the matching kernel yields a qualitatively more accurate reconstruction of the energy field than using the Oldroyd-B kernel.
We also plot the energy field along horizontal slices at $y = 0.475$ and $y=0.975$, further illustrating that the reconstructions using the nonlinear PTT kernel function are in closer agreement with the ground truth than reconstructions obtained using the Oldroyd-B kernel.

\begin{figure}[!ht]
  \centering
  a)\includegraphics[width=.9\textwidth]{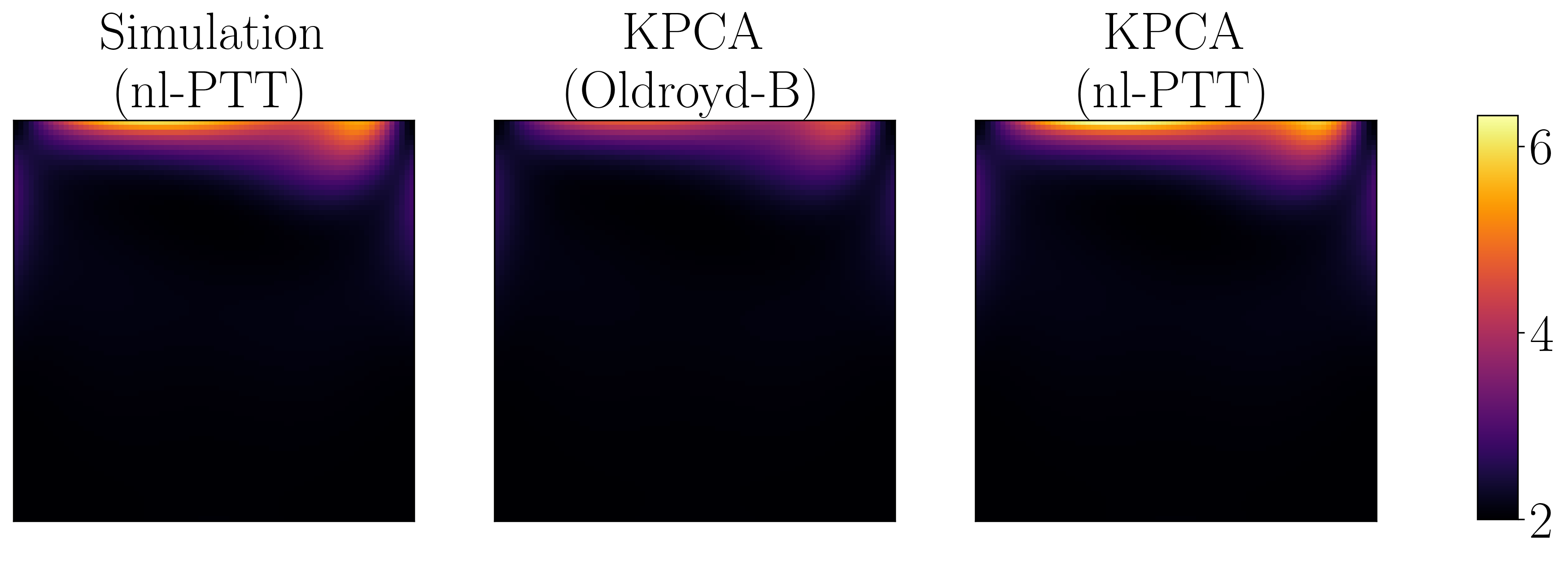}\\
  b)\includegraphics[width=.45\textwidth]{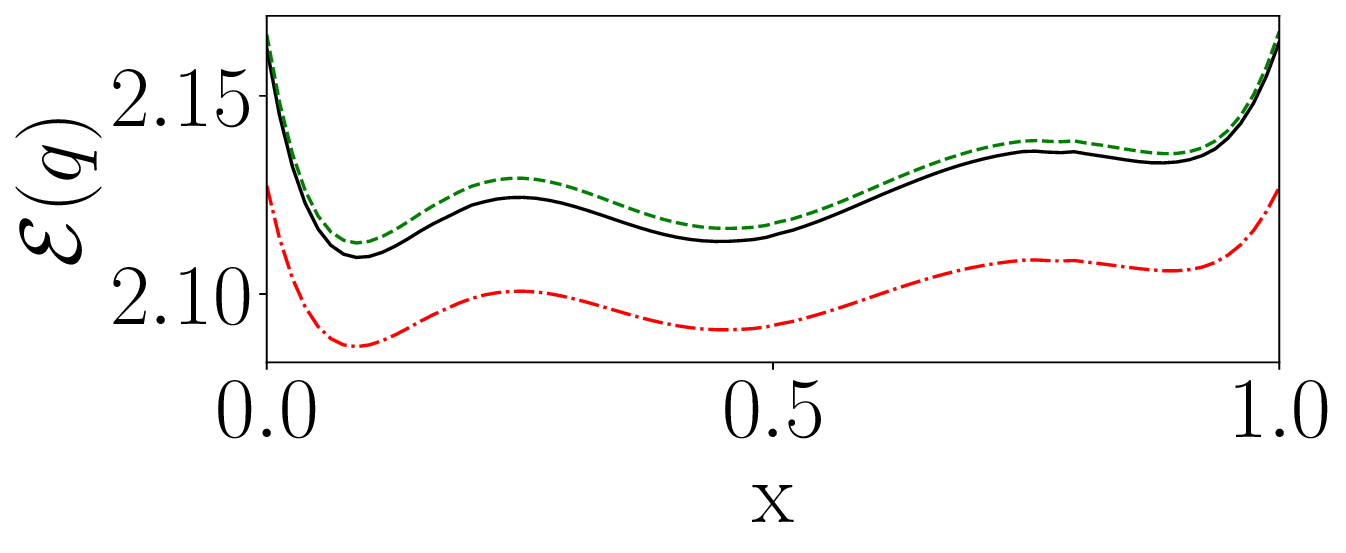}
  \includegraphics[width=.45\textwidth]{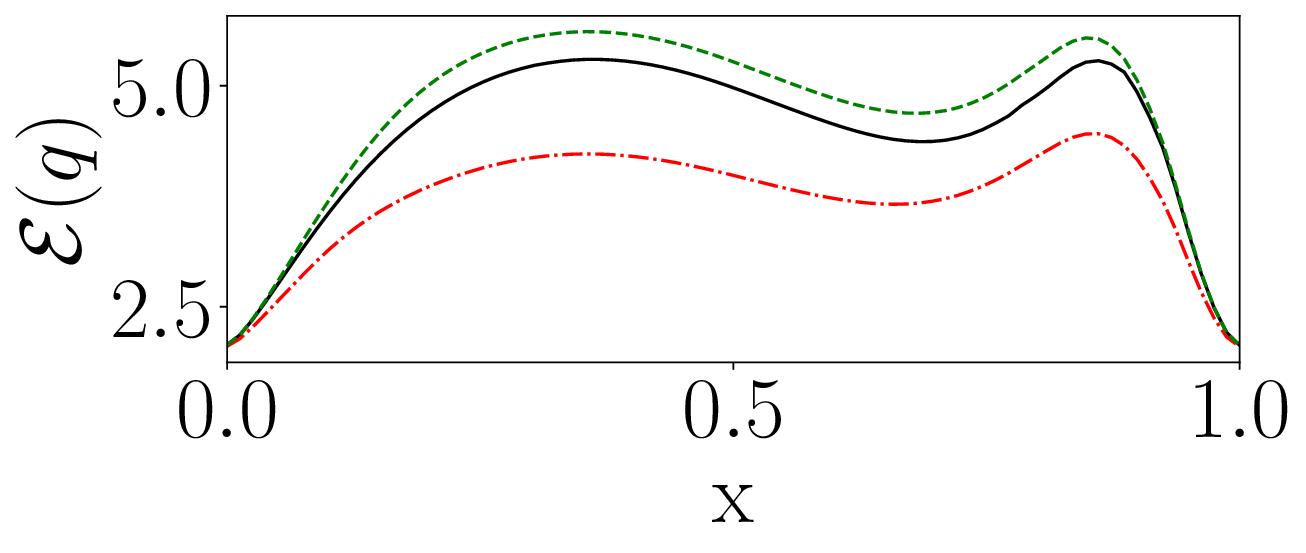}\\
\includegraphics[width=.9\textwidth]{legend_sim_ptt.eps}
  \caption{Spatial distribution of the total mechanical energy (\eqref{eqn:energy}) for a simulation using the nonlinear PTT model with $\varepsilon=0.3$ and $\theta =1$ at time $t = 0.6$. Row (a) compares reconstructed snapshots using $r=2$ kernel principal components obtained using the Oldroyd-B and nonlinear PTT kernels to the ground truth. Row (b) shows the mechanical energy distributions across horizontal slices located at $y = 0.475$ (left) and $y = 0.975$ (right).}
  \label{recons_ptt}
\end{figure}

\section{Discussion}

In summary, our theoretical results provide easily verifiable conditions on a given viscoelastic stress model ensuring that there is a corresponding positive-definite kernel function compatible with the total mechanical energy and turning the space of flowfields with finite mechanical energies into a complete, separable metric space.
Moreover, this kernel function is constructed explicitly from a convergent series representation of the stress model.
Remarkably, these conditions hold for many standard viscoelastic stress models, yielding principled choices for kernel functions which give rise to natural measures for distances and angles between flowfield snapshots.
These geometric notions correspond to implicitly-defined lifted representations of snapshots in a unique reproducing kernel Hilbert space (RKHS) associated with the kernel function.
\revtwo{These kernel functions can then be used to generate large families of energy-compatible kernels by forming products with kernels satisfying a normalization condition.}

The kernel functions and metrics for viscoelastic flows introduced in this paper enable a variety of machine learning algorithms to be employed to extract statistical information from flowfield snapshot data.
In this paper we focus primarily on dimensionality reduction using kernel principal component analysis (KPCA).
Here, a major challenge known as the ``preimage problem'' is to reconstruct flowfields from their RKHS representations or truncated representations in terms of the leading kernel principal components.
We provide a solution to the preimage problem for our viscoelastic kernel functions by showing that the velocity field and square root conformation tensor field can be linearly reconstructed from RKHS representations.
We provide bounds on the reconstruction quality using truncated representations given by the leading kernel principal components.
More generally, our results describe which fields can be linearly reconstructed based on the terms in a series expansion for a given viscoelastic stress model.

We demonstrate the utility of our kernel functions for dimensionality reduction and reconstruction of snapshots from a lid-driven cavity flow.
Here, the flow is simulated with different underlying stress models and we compare the reconstructions obtained using the leading kernel principal components extracted with different choices of kernel function.
Our results underscore the importance of choosing an appropriate metric for reconstruction error.
Ordinary principal component analysis (PCA), which corresponds to a kernel function given by the $L^2$ inner product on the spatial domain, leads to low reconstruction error in an $L^2$ sense, but produces poor reconstructions as measured by the total mechanical energy.
For simulations performed using the (linear) Oldroyd-B stress model, superior reconstructions in an energetic sense were obtained using the kernel principal components extracted using the Oldroyd-B kernel function.
We note that KPCA using the Oldroyd-B kernel is equivalent to ordinary PCA using properly weighted state vectors based on the square root conformation tensor field.

For simulations performed using nonlinear stress models such as FENE-P and nonlinear PTT, we compare the reduction and reconstruction performance using the corresponding kernel functions to the Oldroyd-B kernel.
In each case, using the appropriate nonlinear kernel function was beneficial for reconstructing flow features in an energetic sense from a small number of kernel principal components.
The benefit of using a nonlinear kernel function matching the underlying stress model became more pronounced as the parameters controlling the nonlinearity of the stress model were increased.
However, for our simple lid-driven cavity flow example, this benefit was less significant than the improvements made by using the Oldroyd-B kernel over na\"{i}ve PCA.
The performance of the nonlinear kernel function was comparable to using the simple Oldroyd-B kernel in many cases where the nonlinearity of the stress model was low or moderate and when more principal components were employed for reconstruction.
This suggests that the Oldroyd-B kernel function (or equivalently, appropriately modified state vectors) could be a useful default choice when processing data from viscoelastic flows with only moderately nonlinear stress models.
However, it is possible that accounting for nonlinearities of the stress model when selecting the kernel function will be important for capturing the behavior of more complex viscoelastic flows.
\revtwo{It may also be beneficial to extract nonlinear features using the energy-compatible product kernels introduced in Section~\ref{subsec:creating_new_kernels}.
Systematic explorations of energy-compatible nonlinear kernels for more complex viscoelastic flow problems will be a subject of future work.}
 
\revtwo{Another} exciting avenue for future work involves using \revtwo{energy-compatible} kernel functions to build low-dimensional data-driven reduced-order models approximating the dynamics of viscoelastic fluid flows.
These low-dimensional dynamical systems can then be used for a variety of key scientific and engineering tasks including qualitative analysis of the flow's dynamics and bifurcations, state estimation from limited sensor measurements, real-time forecasting, and feedback control.
Promising approaches could combine the dynamics-informed features extracted using the kernel covariance balancing reduction using adjoint snapshots (K-CoBRAS) method \cite{Otto2022model} with modeling techniques such as Sparse Identification of Nonlinear Dynamics (SINDy) \cite{Brunton2016discovering,oishi2023} or variants of Dynamic Mode Decomposition (DMD) \cite{Williams2015data, Williams2015kernel, Colbrook2023multiverse}.
Other approaches could use \revtwo{energy-based kernel} distance metrics to formulate loss functions for autoencoder-based reduced-order modeling methods such as those introduced in \cite{Champion2019data, Lee2020model, Otto2023nonlinear}.
Our kernels may also be of use for classifying flow regimes using support vector machines (see \cite{Boser1992training, Hofmann2008kernel}).

\section*{Acknowledgements}
\begin{revtwobox} 
We would like to thank the second referee for pointing out the method of generating new energy-compatible kernel functions discussed in Section~\ref{subsec:creating_new_kernels}.
\end{revtwobox} 
C.M. Oishi and F.V.G. Amaral would like to thank the financial support given by Sao Paulo Research Foundation (FAPESP) grants \#2013/07375-0, \#2021/13833-7, \#2021/07034-4 and \#2023/06035-2, and the National Council for Scientific and Technological Development (CNPq), grants \#305383/2019-1 and \#307228/2023-1. The authors acknowledge support from the National Science Foundation AI Institute in Dynamic Systems
(grant number 2112085) and from the Army Research Office (ARO W911NF-19-1-0045).

\bibliographystyle{elsarticle-num}
\bibliography{references}

\appendix

\section{Proof of Theorem~\ref{thm:viscoelastic_kernels} and Proposition~\ref{prop:product_kernel_injectivity}}
\label{app:viscoelastic_kernels_proof}
The proof of the theorem relies on several preliminary lemmas.
First, the Moore-Aronszajn Theorem is a classical result, stated below for completeness, allowing us to associate a unique RKHS to any positive-definite kernel function.
\begin{theorem}[Moore-Aronszajn \cite{Berlinet2011reproducing}]
    \label{thm:Moore-Aronszajn}
    Let $k:\mcal{F}\times\mcal{F} \to \R$ be a function satisfying the positive-definiteness condition in Eq.~\eqref{eqn:PD_cond}.
    Then there is a unique reproducing kernel Hilbert space $\mcal{H}$ of functions on $\mcal{F}$ whose reproducing kernel is $k$.
    In particular, the subspace $\mcal{H}_0$ consisting of finite linear combinations of elements in $\{ K_{\vect{q}} \}_{\vect{q}\in\mcal{F}}$ is dense in $\mcal{H}$ and $\mcal{H}$ is the set of functions that are pointwise limits of Cauchy sequences in $\mcal{H}_0$ with the inner product
    \begin{equation}
        \left\langle \sum_{i=1}^m a_i K_{\vect{q}_i}, \ \sum_{j=1}^n b_j K_{\vect{q}'_j} \right\rangle_{\mcal{H}_0}
        = \sum_{i=1}^m \sum_{j=1}^n a_i b_j k(\vect{q}_i, \vect{q}'_j).
    \end{equation}
\end{theorem}

The following lemma provides several useful rules for combining positive-definite kernels to produce new positive-definite kernels.
\begin{lemma}[Combining kernels \cite{Hofmann2008kernel}]
    \label{lem:combining_kernels}
    Let $k_1, k_2, \ldots$ be real-valued positive-definite kernel functions on an arbitrary nonempty set $\mcal{X}$.
    Then the set of positive-definite kernels on $\mcal{X}$ is a closed convex cone, that is, 
    \begin{enumerate}[leftmargin=1.0cm]
        \item if $\theta_1, \theta_2 \geq 0$, then $\theta_1 k_1 + \theta_2 k_2$ is a positive-definite kernel; and
        \item if $k(x, x') := \lim_{n\to\infty} k_n(x, x')$ exists for all $x,x'\in\mcal{X}$, then $k$ is a positive-definite kernel.
    \end{enumerate}
    The point-wise product $k_1 k_2$ is also a positive-definite kernel.
    Consequently, if $a_0, a_1, \ldots \geq 0$ are non-negative constants and
    \begin{equation}
        \psi(t) = \sum_{n=0}^{\infty} a_n t^n,
    \end{equation}
    converges for every $t \in k(\mcal{X}\times\mcal{X})$ then $\psi \circ k$ is a positive-definite kernel.
\end{lemma}

We also require conditions on the kernel function ensuring that the associated feature map $\Phi$ is injective.
That is, for every pair of distinct states $\vect{q}_1, \vect{q}_2\in\mcal{F}$ we have
\begin{equation}
    \label{eqn:injectivity_condition}
    k(\vect{q}_1, \vect{q}_1) - 2 k(\vect{q}_1, \vect{q}_2) + k(\vect{q}_2, \vect{q}_2) 
    = \Vert \Phi(\vect{q}_1) - \Phi(\vect{q}_2) \Vert_{\mcal{H}}^2 > 0.
\end{equation}
We call positive-definite kernel functions with the above property ``injective kernels''.
In the following lemma, we provide some useful rules for combining injective kernels.
\begin{lemma}[Combining injective kernels]
    \label{lem:combining_injective_kernels}
    Let $k$ and $k'$ be real-valued positive-definite kernel functions on an arbitrary nonempty set $\mcal{X}$ and suppose that $k$ is injective.
    Then the following hold:
    \begin{enumerate}[leftmargin=1.0cm]
        \item if $\theta_1 > 0$ and $\theta_2 \geq 0$ then $\theta_1 k + \theta_2 k'$ is injective;
        \item if $p \geq 1$ is odd then $k^p$ is injective; and
        \item if $k$ takes only non-negative values, then $k^p$ is injective for every integer $p \geq 1$.
    \end{enumerate}
    Consequently, if $k$ is non-negative-valued and $\psi: k(\mcal{X}\times\mcal{X}) \to \R$ can be expressed as a convergent power series
    \begin{equation}
        \psi(t) = \sum_{n=0}^{\infty} a_n t^n,
    \end{equation}
    with every $a_n \geq 0$ and there being some $n \geq 1$ with $a_n > 0$ then $\psi \circ k$ is injective.  
\end{lemma}
\begin{proof}[Proof of Lemma~\ref{lem:combining_injective_kernels}]
    \label{prf:combining_injective_kernels}
    The first statement follows immediately from Eq.~\eqref{eqn:injectivity_condition}.
    To prove the second and third statements, consider an integer $p \geq 1$ and choose two distinct points $x,y \in \mcal{X}$.
    Since $t \mapsto t^p$ is convex on $[0,\infty)$, Jensen's inequality (Theorem~3.3 in \cite{Rudin1987real} with $\mu$ being a sum of two Dirac measures) yields
    \begin{equation}
        \frac{1}{2}\left[ k(x,x)^p - 2 k(x,y)^p + k(y,y)^p \right] 
        \geq \left( \frac{k(x,x) + k(y,y)}{2} \right)^p - k(x,y)^p.
    \end{equation}
    If $k(x,y) \geq 0$ then Eq.~\eqref{eqn:injectivity_condition} implies that
    \begin{equation}
        \left( \frac{k(x,x) + k(y,y)}{2} \right)^p - k(x,y)^p > 0
    \end{equation}
    because $t \mapsto t^p$ is strictly monotone increasing on $[0,\infty)$.
    This proves the third statement.
    On the other hand, if $k(x,y) < 0$ and $p$ is odd then $k(x,y)^p < 0$, which immediately implies that
    \begin{equation}
        k(x,x)^p - 2 k(x,y)^p + k(y,y)^p > 0,
    \end{equation}
    proving the second statement.
    To prove the final statement, we observe the sum defining $\psi$ converges absolutely since all of the terms are nonnegative.
    Supposing that $a_m > 0$ for some $m \geq 1$ we have
    \begin{equation}
        \psi \circ k = a_m k^m + \sum_{n\neq m} a_n k^n,
    \end{equation}
    where the first term is an injective kernel by the argument above.
    Since the sum in the second term converges, it defines a positive-definite kernel by Lemma~\ref{lem:combining_kernels}.
    Therefore the sum of the two terms is an injective kernel.
\end{proof}

The following lemma allows us to convert terms such as $\Tr(\mat{c})$ and $\Tr\big[ (\mat{c} - \mat{I})^2 \big]$ appearing in the energy function into injective kernel functions.
\begin{lemma}
    \label{lem:matrix_function_kernel}
    Let $D \subset \mathbb{S}^d$ be a set of symmetric matrices and let $f$ be a real-valued function on $\sigma(D) := \bigcup_{\mat{c}\in D} \sigma(\mat{c})$.
    With the action of this function on a matrix $\mat{c} \in D$ defined by Eq.~\eqref{eqn:matrix_functional_calculus},
    \begin{equation}
        k(\mat{c}_1, \mat{c}_2) = \Tr\left[ f(\mat{c}_1) f(\mat{c}_2) \right] 
    \end{equation}
    is a positive-definite kernel function on $D\times D$.
    If $f$ is injective on $\sigma(D)$, then the kernel function is injective.
\end{lemma}
\begin{proof}[Proof of Lemma~\ref{lem:matrix_function_kernel}]
    Positive-definiteness is obvious.
    If $f$ is injective on $\sigma(D)$ then its inverse $f^{-1}$ can be defined on $f(\sigma(D))$.
    Since $\sigma(f(D)) \subset f(\sigma(D))$, we have $f^{-1}(f(\mat{c})) = \mat{c}$ for every $\mat{c} \in D$, meaning that $f$ is injective on $D$.
    If $\mat{c}_1, \mat{c}_2$ are distinct elements in $D$, then $f(\mat{c}_1) \neq f(\mat{c}_2)$ and we have
    \begin{equation}
        k(\mat{c}_1, \mat{c}_1) - 2 k(\mat{c}_1, \mat{c}_2) + k(\mat{c}_2, \mat{c}_2)
        = \| f(\mat{c}_1) - f(\mat{c}_2) \|_{F}^2 > 0.
    \end{equation}
    This proves that $k$ is an injective kernel function.
\end{proof}

Finally, we are ready to prove Theorem~\ref{thm:viscoelastic_kernels}.
\begin{proof}[Proof of Theorem~\ref{thm:viscoelastic_kernels}]
    \label{prf:viscoelastic_kernels}
    The series in Eq.~\eqref{eqn:integrand_kernel} converges absolutely because
    \begin{equation}
    \begin{aligned}
        \sum_{i,p = 0}^{\infty} c_{i,p} \left\vert\Tr\left( f_i(\mat{c}_1) f_i(\mat{c}_2) \right)\right\vert^p
        &\leq \sum_{i,p = 0}^{\infty} \sqrt{c_{i,p} \Tr\big(f_i(\mat{c}_1) f_i(\mat{c}_1) \big)^{p}} \sqrt{ c_{i,p} \Tr\big(f_i(\mat{c}_2) f_i(\mat{c}_2)\big)^{p}} \\
        &\leq \sqrt{ \sum_{i,p = 0}^{\infty} c_{i,p} \Tr\big(f_i(\mat{c}_1) f_i(\mat{c}_1)\big)^{p} } \sqrt{ \sum_{j,q = 0}^{\infty} c_{j,q} \Tr\big(f_j(\mat{c}_2) f_j(\mat{c}_2)\big)^{q} },
    \end{aligned}
    \end{equation}
    thanks to two applications of the Cauchy-Schwarz inequality.
    The two terms in the rightmost product are finite by assumption.
    By Lemma~\ref{lem:combining_kernels} and Lemma~\ref{lem:matrix_function_kernel} it follows that $\tilde{k}$ defined by Eq.~\eqref{eqn:main_kernel_fun} is a positive-definite kernel function on $D(\mat{s})$.
    It is then easy to see that Eq.~\eqref{eqn:main_kernel_fun} satisfies the positive-definiteness condition in Eq.~\eqref{eqn:PD_cond}, and is therefore a positive-definite kernel on $\mcal{F}$.
    Since we have
    \begin{equation}
        \mcal{E}(\vect{q}) 
        = \frac{1}{2} \int_{\Omega} \left[ \big\vert \vect{u}(\vect{x}) \big\vert^2 + \theta h(\mat{c}) \right] \td \vect{x},
        \label{eqn:shifted_energy}
    \end{equation}
    where $h(\mat{c}) := -\Wei \cdot \Tr\big(\mat{s}(\mat{c})\big) + c = \tilde{k}(\mat{c}, \mat{c})$,
    it follows that $\mcal{E}(\vect{q}) = k(\vect{q}, \vect{q})$.

    Suppose that there is a coefficient $c_{i,p} > 0$ with $p \geq 1$, $f_i$ injective, and $p$ odd or $f_i$ nonnegative.
    We first show that $\tilde{k}$ is an injective kernel function on $D(\mat{s})$.
    The case when $p$ is odd follows immediately from Lemma~\ref{lem:combining_injective_kernels}(i, ii), and Lemma~\ref{lem:matrix_function_kernel}.
    In the case that $f_i$ is nonnegative, every $f_i(\mat{c})$ is a positive semi-definite matrix and it follows that the kernel function
    \begin{equation}
        (\mat{c}_1, \mat{c}_2) \mapsto \Tr\big[ f_i(\mat{c}_1) f_i(\mat{c}_2) \big] 
        = \Tr\left[ \sqrt{f_i(\mat{c}_2)} f_i(\mat{c}_1) \sqrt{f_i(\mat{c}_2)} \right]
    \end{equation}
    takes only nonnegative values on $\mathbb{S}^d\times\mathbb{S}^d$.
    By Lemma~\ref{lem:combining_injective_kernels}(i, iii) and Lemma~\ref{lem:matrix_function_kernel} it follows that $\tilde{k}$ is an injective kernel function.

   Using injectivity of $\tilde{k}$, we prove that $k$ is an injective kernel function on $\mcal{F}$.
    We observe that
    \begin{equation}
        \vert \vect{u}_1 \vert^2
        -2 \vect{u}_1 \cdot \vect{u}_2
        + \vert \vect{u}_2 \vert^2 \geq 0
    \end{equation}
    by the Cauchy-Schwarz inequality for the dot product on $\R^d$ and
    \begin{equation}
        \tilde{k}(\mat{c}_1, \mat{c}_1)
        -2 \tilde{k}(\mat{c}_1, \mat{c}_2)
        + \tilde{k}(\mat{c}_2, \mat{c}_2) \geq 0
    \end{equation}
    by the Cauchy-Schwarz inequality for the positive-definite kernel $\tilde{k}$.
    Thus, if
    \begin{equation}
        k(\vect{q}_1, \vect{q}_1) - 2 k(\vect{q}_1, \vect{q}_2) + k(\vect{q}_2, \vect{q}_2) = 0,
    \end{equation}
    then for almost every $\vect{x}\in\Omega$ we have
    \begin{equation}
        \big\vert \vect{u}_1(\vect{x}) \big\vert^2
        -2 \vect{u}_1(\vect{x}) \cdot \vect{u}_2(\vect{x})
        + \big\vert \vect{u}_2(\vect{x}) \big\vert^2 = 0,
    \end{equation}
    which implies that $\vect{u}_1(\vect{x}) = \vect{u}_2(\vect{x})$,
    and
    \begin{equation}
        \tilde{k}\big(\mat{c}_1(\vect{x}), \mat{c}_1(\vect{x})\big)
        -2 \tilde{k}\big(\mat{c}_1(\vect{x}), \mat{c}_2(\vect{x})\big)
        + \tilde{k}\big(\mat{c}_2(\vect{x}), \mat{c}_2(\vect{x})\big) = 0,
    \end{equation}
    which implies that $\mat{c}_1(\vect{x}) = \mat{c}_2(\vect{x})$.
    Therefore, $\vect{q}_1 = \vect{q}_2$ in $\mcal{F}$, proving that $k$ is an injective kernel function.
\end{proof}
\begin{revtwobox}
    
\begin{proof}[Proof of Proposition~\ref{prop:product_kernel_injectivity}]
    Suppose that
    \begin{equation}\label{product_injectivity_eqn1}
        k\tilde{k}(\vect{q}_1, \vect{q}_1) - 2 k\tilde{k}(\vect{q}_1, \vect{q}_2) + k\tilde{k}(\vect{q}_2, \vect{q}_2) = 0.
    \end{equation}
    Then by the assumed property of $\tilde{k}$, the arithmetic and geometric means (AM-GM) inequality, and the Cauchy-Schwarz inequality, we obtain
    \begin{equation}
        k(\vect{q}_1, \vect{q}_2)\tilde{k}(\vect{q}_1, \vect{q}_2)
        = \frac{1}{2}\left( k(\vect{q}_1, \vect{q}_1) + k(\vect{q}_2, \vect{q}_2) \right)
        \geq \sqrt{k(\vect{q}_1, \vect{q}_1) k(\vect{q}_2, \vect{q}_2)}
        \geq k(\vect{q}_1, \vect{q}_2).
    \end{equation}
    This means that $\tilde{k}(\vect{q}_1, \vect{q}_2) \geq 1$.
    By Cauchy-Schwarz and the assumed property of $\tilde{k}$, we have
    \begin{equation}
        1 \leq \tilde{k}(\vect{q}_1, \vect{q}_2)
        \leq \sqrt{\tilde{k}(\vect{q}_1, \vect{q}_1) \tilde{k}(\vect{q}_2, \vect{q}_2)} \leq 1,
    \end{equation}
    meaning that $\tilde{k}(\vect{q}_1, \vect{q}_2) = 1$.
    Substituting this into \eqref{product_injectivity_eqn1} and using the assumed property of $\tilde{k}$ yields
    \begin{equation}
        k(\vect{q}_1, \vect{q}_1) - 2 k(\vect{q}_1, \vect{q}_2) + k(\vect{q}_2, \vect{q}_2) = 0,
    \end{equation}
    which by injectivity of $k$, means that $\vect{q}_1 = \vect{q}_2$.
    This completes the proof.
\end{proof}
\end{revtwobox}

\section{Proof of Theorem~\ref{thm:completeness_of_metric}}
\label{app:completeness_proof}

Our proof relies on the following technical lemma.
\begin{lemma}
    \label{lem:kernel_power_inequality}
    Let $x, y$ be vectors in a Hilbert space.
    If $p \geq 1$ is odd or $\langle x, y \rangle \geq 0$, then
    \begin{equation}
        \langle x, x \rangle^p - 2 \langle x, y \rangle^p + \langle y, y \rangle^p
        \geq \frac{1}{2^{2(p-1)}} \| x - y \|^{2p}.
    \end{equation}
\end{lemma}
\begin{proof}[Proof of Lemma~\ref{lem:kernel_power_inequality}]
    When $x = y$, the statement is vacuously true, so we assume that $x \neq y$ and denote
    \[
        \alpha = \frac{2\langle x, y \rangle}{\| x - y \|^2}.
    \]
    By Jensen's inequality (Theorem~3.3 in \cite{Rudin1987real} with $\mu$ being a sum of two Dirac measures) we have
    \begin{equation}
        \langle x, x \rangle^p - 2 \langle x, y \rangle^p + \langle y, y \rangle^p
        \geq 2\left[ \left( \frac{\langle x, x \rangle + \langle y, y \rangle}{2} \right)^p - \langle x, y \rangle^p \right]
        = 2\left[ \left( \frac{1}{2} \| x - y \|^2 + \langle x, y \rangle  \right)^p - \langle x, y \rangle^p \right].
    \end{equation}
    Dividing through by $2(\frac{1}{2}\| x - y\|^2)^p$ gives
    \begin{equation}
        \frac{\langle x, x \rangle^p - 2 \langle x, y \rangle^p + \langle y, y \rangle^p}{2^{1-p}\| x - y\|^{2p}}
        \geq (1 + \alpha)^p - \alpha^p =: f(\alpha),
    \end{equation}
    and so it remains to lower bound $f(\alpha)$ by a positive constant, specifically $2^{1-p}$.
    When $\langle x, y \rangle \geq 0$, we have $\alpha \geq 0$, and it is easy to see that $f(\alpha ) \geq 1 \geq 2^{1-p}$.
    
    Now we assume that $p \geq 1$ is odd.
    Differentiating $f(\alpha)$, we find
    \begin{equation}
        f'(\alpha) 
        = p (1 + \alpha)^{p-1} - p \alpha^{p-1}
        = p | 1 + \alpha |^{p-1} - p | \alpha |^{p-1}
    \end{equation}
    since $p-1$ is even.
    When $\alpha \geq -\frac{1}{2}$, we evidently have $f'(\alpha) \geq 0$.
    Likewise, when $\alpha \leq -\frac{1}{2}$ we have $f'(\alpha) \leq 0$.
    Therefore, for every real $\alpha$ we have
    \begin{equation}
        f(\alpha) \geq f\left(-\tfrac{1}{2}\right) = 2(\tfrac{1}{2})^p = 2^{1-p},
    \end{equation}
    which completes the proof of the lemma.
\end{proof}

With this lemma in hand, we are ready to prove the theorem.
\begin{proof}[Proof of Theorem~\ref{thm:completeness_of_metric}]
Thanks to Theorem~\ref{thm:viscoelastic_kernels}, $\mcal{F}$ is a metric space with metric $d_{\mcal{E}}$ defined by Eq.~\eqref{eqn:kernel_metric}.
Let $\{ \vect{q}_n = (\vect{u}_n, \mat{c}_n) \}_{n=1}^{\infty} \in \mcal{F}$ be a Cauchy sequence. 
Since $L^2(\Omega)$ is complete (see Theorem~3.11 in \cite{Rudin1987real}) and
\begin{equation}
    d_{\mcal{E}}(\vect{q}_m, \vect{q}_n) \geq \frac{1}{\sqrt{2}} \| \vect{u}_m - \vect{u}_n \|_{L^2(\Omega)},
\end{equation}
it follows that $\vect{u}_n \to \vect{u}$ in $L^2(\Omega)$ for a unique velocity field $\vect{u}$.

Considering the term $c_{i,p} > 0$ and letting $\theta = (1-\beta)/(\Rey\Wei)$, we have
\begin{multline}
    \left(\frac{2}{c_{i,p} \theta}\right) d_{\mcal{E}}(\vect{q}_m, \vect{q}_n)^2
    \geq \int_{\Omega} \left\{ \left[\Tr\left( f_i(\mat{c}_m)^2 \right)\right]^p - 2 \left[\Tr\left( f_i(\mat{c}_m) f_i(\mat{c}_n) \right)\right]^p + \left[\Tr\left( f_i(\mat{c}_n)^2 \right)\right]^p \right\} \td \vect{x} \\
    \geq \frac{1}{2^{2(p-1)}} \int_{\Omega} \| f_i(\mat{c}_m) - f_i(\mat{c}_n) \|_F^{2p} \td \vect{x},
\end{multline}
thanks to Lemma~\ref{lem:kernel_power_inequality}.
It follows from the completeness theorem for $L^{2p}(\Omega)$, specifically Theorems~3.11~and~3.12 in \cite{Rudin1987real}, that $f_i(\mat{c}_n)$ converges in $L^{2p}(\Omega)$ to a limit $\mat{f}$ and that there is a subsequence $\{f_i(\mat{c}_{n_k})\}_{k=1}^{\infty}$ converging pointwise almost everywhere to $\mat{f}$ in $\Omega$.
Since the function $f_i$ is injective, it follows that the conformation tensor field $\mat{c} = f_i^{-1}(\mat{f}): \Omega \to \mathbb{S}_{+}^d$ satisfies $f_i(\mat{c}_n) \to f_i(\mat{c})$ in $L^{2p}(\Omega)$ and $\mat{c}_{n_k}(\vect{x}) \to \mat{c}(\vect{x})$ for almost every $\vect{x}\in\Omega$.

It remains to show that $\vect{q} = (\vect{u}, \mat{c})$ satisfies $\mcal{E}(\vect{q}) < \infty$ and that $d_{\mcal{E}}(\vect{q}_n, \vect{q}) \to 0$.
Both of these are accomplished by means of Lebesgue's dominated convergence theorem (Theorem~1.34 in \cite{Rudin1987real}).
We pass to a further subsequence, still denoted with indices $n_k$, such that
\begin{equation}
    d_{\mcal{E}}(\vect{q}_{n_k}, \vect{q}_{n_{k+1}}) < 2^{-k}.
\end{equation}
We let $\tilde{\Phi}: D(\mat{s}) \to \tilde{\mcal{H}}$ denote the feature map associated with the reproducing kernel $\tilde{k}:D(\mat{s})\times D(\mat{s}) \to \R$.
We define a function $G : \Omega \to \R$ by
\begin{equation}
    G(\vect{x}) = 
   \big\| \tilde{\Phi}(\mat{c}_{n_1}(\vect{x})) \big\|_{\tilde{\mcal{H}}}
    + \sum_{j=2}^{\infty} \big\| \tilde{\Phi}(\mat{c}_{n_{j}}(\vect{x})) - \tilde{\Phi}(\mat{c}_{n_{j-1}}(\vect{x})) \big\|_{\tilde{\mcal{H}}}.
\end{equation}
We observe that $G \in L^{2}(\Omega)$ because
\begin{equation}
    \sqrt{\frac{\theta}{2}} \| G \|_{L^2(\Omega)}
    \leq \sqrt{\mcal{E}(\vect{q}_{n_1})}
    + \sum_{j=2}^{\infty} d_{\mcal{E}}(\vect{q}_{n_j}, \vect{q}_{n_{j-1}})
    < \infty.
\end{equation}
Moreover, by construction we have 
\begin{equation}
    \tilde{k}\big( \mat{c}_{n_k}(\vect{x}), \mat{c}_{n_k}(\vect{x}) \big) =
    \big\| \tilde{\Phi}(\mat{c}_{n_k}(\vect{x})) \big\|_{\tilde{\mcal{H}}}^2 \leq G(\vect{x})^2
\end{equation}
for every $k \geq 1$.
Therefore, the dominated convergence theorem (Theorem~1.34 in \cite{Rudin1987real}) yields
\begin{equation}
    \lim_{k\to\infty} \int_{\Omega} \tilde{k}\big( \mat{c}_{n_k}(\vect{x}), \mat{c}_{n_k}(\vect{x}) \big) \td \vect{x}
    = \int_{\Omega} \tilde{k}\big( \mat{c}(\vect{x}), \mat{c}(\vect{x}) \big) \td \vect{x}
    \leq \| G \|_{L^2(\Omega)}^2 < \infty.
\end{equation}
Combined with the $L^2(\Omega)$ convergence of $\vect{u}_{n_k}$ to $\vect{u}$, this gives
\begin{equation}
    \lim_{k\to\infty} \mcal{E}(\vect{q}_{n_k}) = \mcal{E}(\vect{q}) < \infty,
\end{equation}
meaning that $\vect{q} \in \mcal{F}$.
Next we observe that for almost every $\vect{x}\in\Omega$,
\begin{equation}
    \left| \tilde{k}(\mat{c}_{n_k}(\vect{x}), \mat{c}_{n_k}(\vect{x})) - 2 \tilde{k}(\mat{c}_{n_k}(\vect{x}), \mat{c}(\vect{x})) + \tilde{k}(\mat{c}(\vect{x}), \mat{c}(\vect{x})) \right| \\
    \leq 4 G(\vect{x})^2
\end{equation}
by the Cauchy-Schwarz inequality in $\tilde{\mcal{H}}$.
Therefore, by another application of the dominated convergence theorem we obtain
\begin{equation}
    d_{\mcal{E}}(\vect{q}_{n_k}, \vect{q}) \to 0.
\end{equation}
Since $\vect{q}_{n}$ is Cauchy, we obtain $d_{\mcal{E}}(\vect{q}_{n}, \vect{q}) \to 0$, proving that $\mcal{F}$ is a complete metric space.

Finally, if $\Phi(\vect{q}_n) \to \phi$ in $\mcal{H}$ then $\{\vect{q}_n\}_{n=1}^{\infty}$ is a Cauchy sequence in $\mcal{F}$.
Since $\mcal{F}$ is a complete metric space, there exists $\vect{q}\in\mcal{F}$ satisfying
\begin{equation}
    d_{\mcal{E}}(\vect{q}_n, \vect{q}) 
    = \| \Phi(\vect{q}_n) - \Phi(\vect{q}) \|_{\mcal{H}} \to 0.
\end{equation}
Since limits in a Hilbert space are unique, we must have $\phi = \Phi(\vect{q})$, proving that $\Phi(\mcal{F})$ is closed in $\mcal{H}$.
\end{proof}

\section{Proof of Theorem~\ref{thm:reconstruction} and Proposition~\ref{prop:KPCA_field_reconstruction}}
\label{app:reconstruction_proofs}
First, we establish a lemma relating powers of the trace to the trace of iterated Kronecker products.
\begin{lemma}
    \label{lem:trace_Kronecker_power}
    For every $\mat{A}, \mat{B} \in \mathbb{R}^{d\times d}$ and integer $p \geq 0$ we have
    \begin{equation}
        \left[ \Tr(\mat{A} \mat{B}) \right]^p = \Tr\left[ \mat{A}^{\otimes p} \mat{B}^{\otimes p} \right],
    \end{equation}
    with the convention that $\mat{A}^{\otimes 0} = 1$.
\end{lemma}
\begin{proof}[Proof of Lemma~\ref{lem:trace_Kronecker_power}]
    The cases $p=0$ and $p=1$ are trivial.
    Suppose that the result holds for a given $p > 1$.
    Then we have
    \begin{equation}
    \begin{aligned}
        \left[ \Tr(\mat{A} \mat{B}) \right]^{p+1} 
        &= \Tr(\mat{A} \mat{B}) \Tr\left[ \mat{A}^{\otimes p} \mat{B}^{\otimes p} \right] \\
        &= \Tr\left[ (\mat{A} \mat{B}) \otimes  \big( \mat{A}^{\otimes p} \mat{B}^{\otimes p} \big) \right] \\
        &= \Tr\left[ \big( \mat{A} \otimes \mat{A}^{\otimes p} \big) \big( \mat{B} \otimes \mat{B}^{\otimes p}\big) \right]
        = \Tr\left[ \mat{A}^{\otimes (p+1)} \mat{B}^{\otimes (p+1)} \right].
    \end{aligned}
    \end{equation}
    The first line holds by the induction hypothesis, the second line is due to the trace property of the Kronecker product, and the third line follows from the mixed product property of the Kronecker product \cite{Horn1991topics, Taboga2021properties}.
    Therefore the stated result holds for all integers $p \geq 0$ by induction on $p$. 
\end{proof}

\begin{proof}[Proof of Theorem~\ref{thm:reconstruction}]
    With a countable set of indices
    \begin{equation}
        \mathbb{I} = \{ (i, p, j, k) \in \mathbb{Z}^4 \ : \ i \geq 0, \ p\geq 0, \ 1\leq j,k \leq d^p \}
    \end{equation}
    we consider the space $L^2(\Omega \times \mathbb{I})$ with inner product
    \begin{equation}
        \langle f, \ g \rangle_{L^2(\Omega \times \mathbb{I})} := \sum_{(i,p,j,k)\in\mathbb{I}} \int_{\Omega} f(\vect{x}, i, p, j, k) g(\vect{x}, i, p, j, k) \td \vect{x}.
    \end{equation}
    Note that the sum and the integral can be exchanged thanks to Fubini's theorem (Theorem~8.8 in \cite{Rudin1987real}).
    Defining $\Psi: \mcal{F} \to H := L^2(\Omega;\ \R^3) \times L^2(\Omega \times \mathbb{I})$ by
    \begin{equation}
        \Psi(\vect{q}) = \frac{1}{\sqrt{2}}(\vect{u}, \ \sqrt{\theta} \tilde{\Psi}(\vect{q})),
        \qquad
        \tilde{\Psi}(\vect{q})(\vect{x}, i, p, j, k) = \sqrt{c_{i,p} } \left[ f_i(\mat{c}(x))^{\otimes p} \right]_{j,k}
    \end{equation}
    and applying Lemma~\ref{lem:trace_Kronecker_power}, we obtain
    \begin{equation}
        k(\vect{q}_1, \vect{q}_2) 
        = \langle \Psi(\vect{q}_1), \ \Psi(\vect{q}_2) \rangle_{H} 
        = \frac{1}{2}\int_{\Omega}\left\{ \vect{u}_1(\vect{x})^T \vect{u}_2(\vect{x}) + \theta \sum_{i=0}^{\infty} \sum_{p=0}^{\infty} c_{i,p} \Tr\left[ f_i(\mat{c}_1(\vect{x}))^{\otimes p} f_i(\mat{c}_2(\vect{x}))^{\otimes p} \right] \right\} \td \vect{x}
    \end{equation}
    for every $\vect{q}_1, \vect{q}_2 \in \mcal{F}$.
    Therefore, by the Moore-Aronszajn theorem (Theorem~\ref{thm:Moore-Aronszajn}), there is a unique to a linear isometry $U: \mcal{H} \to H$ satisfying $U \Phi(\vect{q}) = \Psi(\vect{q})$ for every $\vect{q}\in\mcal{F}$. 
    
    Incidentially, the above shows that the metric space $(\mcal{F}, d_{\mcal{E}})$ is separable.
    This is because $\mcal{F}$ is isometric (via the map $U \Phi)$ to a subset of $H$, which is a separable Hilbert space.

    For every $f \in L^2(\Omega \times \mathbb{I})$, let $\mat{f}_{i,p}:\Omega \to \R^{d^p \times d^p}$ be defined by $[\mat{f}_{i,p}(\vect{x})]_{j,k} = f(\vect{x}, i, p, j, k)$.
    The linear map $T: H \to L^2(\Omega)$ defined by
    \begin{equation}
        T(\vect{u} , f) = \sqrt{2} \vect{a}_0^T \vect{u} + \sqrt{\frac{2}{\theta}} \sum_{i=0}^{\infty} \sum_{p=0}^{\infty} \Tr\left( \mat{A}_{i,p}^T \mat{f}_{i,p} \right)
    \end{equation}
    clearly satisfies
    \begin{equation}
        T U \Phi(\vect{q}) = T \Psi(\vect{q}) = \vect{a}_0^T \vect{u} + \sum_{i=0}^{\infty} \sum_{p=0}^{\infty} \sqrt{c_{i,p}} \Tr\left[ \mat{A}_{i,p}^T f_i(\mat{c})^{\otimes p} \right] = \psi \circ \vect{q}
    \end{equation}
    for every $\vect{q} = (\vect{u}, \mat{c}) \in \mcal{F}$.
    The operator $T$ is bounded because
    \begin{equation}
    \begin{aligned}
        \| T (\vect{u}, f) \|_{L^2(\Omega)}^2 
        &= \int_{\Omega} \left\vert \sqrt{2} \vect{a}_0^T \vect{u}(\vect{x}) + \sqrt{\frac{2}{\theta}} \sum_{i=0}^{\infty} \sum_{p=0}^{\infty} \Tr\left[ \mat{A}_{i,p}^T \mat{f}_{i,p}(\vect{x}) \right] \right\vert^2 \td \vect{x} \\
        &\leq \int_{\Omega} \left\vert \sqrt{2} \|\vect{a}_0 \|_2 \| \vect{u}(\vect{x}) \|_2 + \sqrt{\frac{2}{\theta}} \sum_{i=0}^{\infty} \sum_{p=0}^{\infty} \| \mat{A}_{i,p} \|_F \| \mat{f}_{i,p}(\vect{x}) \|_F \right\vert^2 \td \vect{x} \\
        &\leq \underbrace{\left( 2 \| \vect{a}_0 \|_2^2 + \frac{2}{\theta} \sum_{i=0}^{\infty} \sum_{p=0}^{\infty} \| \mat{A}_{i,p} \|_F^2 \right)}_{A^2 < \infty \mbox{ by assumption}} \underbrace{\int_{\Omega} \left( \| \vect{u}(\vect{x}) \|_2^2 + \sum_{i=0}^{\infty} \sum_{p=0}^{\infty} \| \mat{f}_{i,p}(\vect{x}) \|_F^2  \right) \td \vect{x}}_{\| (\vect{u}, f) \|_{H}^2}
    \end{aligned}
    \end{equation}
    thanks to two applications of the Cauchy-Schwarz inequality.
    It follows that the linear operator $R_{\psi} := T U$ is bounded, with operator norm $\| R_{\psi} \| \leq A$.
    Since the span of vectors $\{ \Phi(\vect{q}) \}_{\vect{q}\in\mcal{F}}$ is dense in $\mcal{H}$ by the by the Moore-Aronszajn theorem (Theorem~\ref{thm:Moore-Aronszajn}), the bounded linear operator $R_{\psi}$ is uniquely defined by the relation $R_{\psi} \Phi(\vect{q}) = \psi \circ \vect{q}$ for all $\vect{q}\in\mcal{F}$.
\end{proof}

\begin{proof}[Proof of Proposition~\ref{prop:KPCA_field_reconstruction}]
    Let $P_r = U_r U_r^*$ denote the orthogonal projection onto the span of $u_1, \ldots, u_r$ in $\mcal{H}$.
    Thanks to Corollary~\ref{cor:reconstruction}, we have
    \begin{equation}
    \label{eqn:single_state_KPCA_rec_derivation}
    \begin{aligned}
        \left\| \big( \vect{u},\ f_i(\mat{c})^{\otimes p} \big) - R_{i,p} U_r \vect{z}_r(\vect{q}) \right\|_{L^2(\Omega)}^2 
        &= \| R_{i,p} \Phi(\vect{q}) - R_{i,p} P_r \Phi(\vect{q}) \|_{L^2(\Omega)}^2 \\
        &\leq \| R_{i,p} \|^2 \| (I - P_r) \Phi(\vect{q}) \|_{\mcal{H}}^2 \\
        &= \| R_{i,p} \|^2 \left( \| \Phi(\vect{q}) \|_{\mcal{H}}^2 - \| P_r \Phi(\vect{q}) \|_{\mcal{H}}^2 \right).
    \end{aligned}
    \end{equation}
    Since $\| \Phi(\vect{q}) \|_{\mcal{H}}^2 = \mcal{E}(\vect{q})$ and $ \| P_r \Phi(\vect{q}) \|_{\mcal{H}}^2 = \| U_r^* \Phi(\vect{q}) \|_2^2 = \| \vect{z}_r(\vect{q}) \|_2^2, $
    we obtain Eq.~\eqref{eqn:single_state_KPCA_rec_bound} using the bound on $\| R_{i,p} \|$ stated in Corollary~\ref{cor:reconstruction}.
    Integrating, applying Parseval's theorem (Theorem~II.6 in \cite{Reed1980functional}), and using Fubini's theorem (Theorem~8.8 in \cite{Rudin1987real}) to exchange summation and integration yields
    \begin{equation}
    \begin{aligned}
        \int_{\mcal{F}} \| (I - P_r) \Phi(\vect{q}) \|_{\mcal{H}}^2 \td\mu(\vect{q})
        &= \int_{\mcal{F}} \sum_{j=r+1}^{\infty} \left|\langle u_j, \ \Phi(\vect{q})\rangle_{\mcal{H}}\right|^2 \td\mu(\vect{q}) \\
        &= \sum_{j=r+1}^{\infty} \int_{\mcal{F}} \langle u_j, \ \Phi(\vect{q})\rangle_{\mcal{H}} \langle \Phi(\vect{q}), u_j \rangle_{\mcal{H}}  \td\mu(\vect{q}) \\
        &= \sum_{j=r+1}^{\infty} \langle u_j, \ C_{\mu} u_j \rangle_{\mcal{H}} 
        = \sum_{j=r+1}^{\infty} \sigma_j^2.
    \end{aligned}
    \end{equation}
    Combining this with Eq.~\eqref{eqn:single_state_KPCA_rec_derivation} and the bound on $\| R_{i,p} \|$ in Corollary~\ref{cor:reconstruction} yeilds Eq.~\eqref{eqn:average_KPCA_rec_bound}, completing the proof.
\end{proof}

\end{document}